\documentclass{amsart}
\usepackage{mathrsfs}
\usepackage{amssymb}
\usepackage{amsmath}
\usepackage{amsfonts}
\usepackage{enumerate}
\usepackage{pifont}
\usepackage{enumerate}
\usepackage[pdftex]{graphicx}
\usepackage{verbatim}
\usepackage{amsthm}
\usepackage[all]{xy}

\newtheorem{thm}{Theorem}[section]
\newtheorem{prop}[thm]{Proposition}
\newtheorem{cor}[thm]{Corollary}
\newtheorem{lem}[thm]{Lemma}
\newtheorem{defi}[thm]{Definition}
\newtheorem{rk}[thm]{Remark}

\numberwithin{equation}{section}

\newcommand{\real}{{\mathbb R}}

\newcommand{\8}{\infty}
\newcommand{\el}{\ell}

\newcommand{\be}{\begin{eqnarray*}}
\newcommand{\ee}{\end{eqnarray*}}
\newcommand{\beq}{\begin{equation}}
\newcommand{\eeq}{\end{equation}}
\newcommand{\beqn}{\begin{equation*}}
\newcommand{\eeqn}{\end{equation*}}

\begin{document}

\title{An infinite linear hierarchy for the incompressible Navier-Stokes equation and application}


\author{Zeqian Chen}

\address{Wuhan Institute of Physics and Mathematics, Chinese
Academy of Sciences, 30 West District, Xiao-Hong-Shan, Wuhan 430071, China}


\date{}
\maketitle
\markboth{Zeqian Chen}%
{Navier-Stokes equation}

\begin{abstract}
This paper introduces an infinite linear hierarchy for the homogeneous, incompressible three-dimensional Navier-Stokes equation. The Cauchy problem of the hierarchy with a factorized divergence-free initial datum is shown to be equivalent to that of the incompressible Navier-Stokes equation in $\mathcal{H}^1.$ This allows us to present an explicit formula for solutions to the incompressible Navier-Stokes equation under consideration. The obtained formula is an expansion in terms of binary trees encoding the collision histories of the ``particles" in a concise form. Precisely, each term in the summation of $n$ ``particles" collision is expressed by a $n$-parameter singular integral operator with an explicit kernel in Fourier space, describing a kind of processes of two-body interaction of $n$ ``particles". Therefore, this formula is a physical expression for the solutions of the incompressible Navier-Stokes equation.
\end{abstract}

\tableofcontents


\section{Introduction}\label{Intro}

We are concerned with the homogeneous, incompressible Navier-Stokes equation in $\mathbb{R}^3$
\beq\label{eq:NSE}
\left \{\begin{split}
& \partial_t u + (u \cdot \bigtriangledown ) u = \triangle u - \bigtriangledown p,\\
& \bigtriangledown \cdot u = 0, \end{split} \right.
\eeq
with the initial data $u (0) = u_0$ satisfying $\bigtriangledown \cdot u_0 =0.$ Recall that $u = u (t,x) \in \real^3$ is the velocity of the fluid at position $x \in \real^3$ and time $t > 0,$ and $p = p(t,x)$ is a scalar field called the pressure of the fluid, while $u_0 = u_0 (x) \in \real^3, x \in \real^3,$ is a given initial velocity vector. By eliminating the pressure $p,$ the equation \eqref{eq:NSE} is reformulated as
\beq\label{eq:NSEpressure-free}
\left \{\begin{split}
& \partial_t u = \triangle u - W (u \otimes u), \\
& \bigtriangledown \cdot u = 0, \end{split} \right.
\eeq
where $W (u \otimes u) = P \bigtriangledown \cdot (u \otimes u)$ with $P$ being the Leray projection on $[ L^2 (\mathbb{R}^3)]^3.$ This formulation was traced back to Leray \cite{Leray1934} (see also \cite[Chapter 11]{LR2002} for details).

From the quantum-mechanical point of view, the nonlinear term of the first equation of \eqref{eq:NSEpressure-free} that involves a two-fold tensor function should indicate the on-site effect of many-body pair interaction in some sense. This observation allows us to introduce a sequence of symmetric tensor functions $u^{(k)} (t, \vec{x}_k): = \otimes^k_{j=1} u (t, x_j)$ with $\vec{x}_k = (x_1, \ldots, x_k) \in (\mathbb{R}^3 )^k$ for all $k \ge 1.$ It then follows that $\mathfrak{U} = ( u^{(k)} )_{k \ge 1}$ satisfies an infinite hierarchy of linear equations that follows
\beq\label{eq:NSEhierarchy}
\begin{split}
\partial_t u^{(k)} (t, \vec{x}_k) = \sum^k_{j=1} \triangle_j u^{(k)} (t, \vec{x}_k) - \sum^k_{j=1} W_j u^{(k+1)} (t, \vec{x}_k, x_j)
\end{split}
\eeq
for all $k \ge 1,$ where $\triangle_j$ and $W_j$ denote respectively the operators $\triangle$ and $W$ acting on $x_j \in \mathbb{R}^3$ for every $j \ge 1.$ Conversely, a symmetric tensor solution $\mathfrak{U} = ( u^{(k)} )_{k \ge 1}$ to \eqref{eq:NSEhierarchy} with a factorized divergence-free initial data leads necessarily to a solution to \eqref{eq:NSEpressure-free}, thanks to the uniqueness of solutions to the initial problem for the hierarchy \eqref{eq:NSEhierarchy} (see Section \ref{Unique} below). We thereby can investigate the Cauchy problem of \eqref{eq:NSE} through using the infinite linear hierarchy \eqref{eq:NSEhierarchy}. In what follows, we will call this hierarchy the {\it Navier-Stokes hierarchy}, since it can be obtained from the Navier-Stokes equation \eqref{eq:NSE}.

By the linearity of the Navier-Stokes hierarchy \eqref{eq:NSEhierarchy}, its solution with an initial datum $( u^{(k)}_0)_{k \ge 1}$ can be formally expanded in a Duhamel-type series, i.e., for any $k \ge 1,$
\beq\label{eq:DuhamelExpression}
\begin{split}
u^{(k)} (t) = e^{t \triangle^{(k)}} u^{(k)}_0 + \sum^n_{j=1} \int^t_0 d t_1 \int^{t_1}_0 & d t_2 \cdots  \int^{t_{j-1}}_0 d t_j e^{( t- t_1)\triangle^{(k)}}  W^{(k)}\cdots \\
& \times e^{( t_{j-1}-t_j) \triangle^{(k+j-1)}} W^{(k+j-1)} e^{ t_j \triangle^{(k+j)}} u^{(k+j)}_0\\
+ \int^t_0 d t_1 \int^{t_1}_0 & d t_2 \cdots \int^{t_n}_0 d t_{n+1} e^{ ( t-t_1) \triangle^{(k)}} W^{(k)} \cdots \\
& \times e^{ ( t_n - t_{n+1} ) \triangle^{(k+n)}} W^{(k + n)} u^{(k + n +1)} (t_{n+1}) \end{split}
\eeq
for every $n \ge 1,$ with the convention $t_0 =t,$ where $\triangle^{(m)} = \sum^m_{j=1} \triangle_j$ and $W^{(m)} = - \sum^m_{j=1} W_j$ for $m \ge 1.$ Given a fixed $k \ge 1,$ from the definition of $W^{(m)}$ there are about $k (k+1) \cdots (k+n) \sim n!$ terms in the summation and remainder expressions on the right hand of \eqref{eq:DuhamelExpression}. For handling the integration terms in this expression, a natural method is to perform an iterative estimate involving subsequent one-parameter space-time dispersive bounds. Unfortunately, the present author was unable to prove {\it a prior} space-time estimates for $W^{(m)}$'s cancelling the factor $n!$ at the moment of this writing. For this reason, instead we manage to present an expansion in terms of binary trees as follows
\beq\label{eq:NSEhierarchyTreeExpanssionIntro}
\begin{split}
u^{(k)} (t) = e^{t \triangle^{(k)}} u^{(k)}_0 + \sum^n_{j=1} \sum_{\mathbb{T} \in \mathfrak{T}_{j,k}} C_{\mathbb{T}, t} u^{(k+j)}_0 - \sum_{\mathbb{T} \in \mathfrak{T}_{n+1,k}} \int^t_0 d s R_{\mathbb{T}, t-s} u^{(k+n+1)} (s)
\end{split}
\eeq
for any $k \ge 1$ and for every $n \ge 1.$ Here, $\mathfrak{T}_{m,k}$ is the set of $k$-rooted binary trees (see Section \ref{App} below) encoding the collision ways of $k+m$ ``particles" with $|\mathfrak{T}_{m,k}| \lesssim C^m,$ where $C$ is a constant depending only on $k;$ and both $ C_{\mathbb{T}, t}$ and $R_{\mathbb{T}, t-s}$ are multi-parameter singular integral operators on tensor product spaces, indicating the two-body interaction of $k+m$ ``particles" in a concise form. The key novelty of \eqref{eq:NSEhierarchyTreeExpanssionIntro} is to reduce the number of $n!$ terms in the expansion \eqref{eq:DuhamelExpression} to $C^n.$

A suitable strategy for using the Navier-Stokes hierarchy \eqref{eq:NSEhierarchy} to investigate the Navier-Stokes equation \eqref{eq:NSEpressure-free} is to establish {\it a prior} space-time estimates for the interaction operators $C_{\mathbb{T}, t}$ and $R_{\mathbb{T}, t},$ which are singular integral operators from multi-parameter product spaces to one-parameter spaces. Although $C_{\mathbb{T}, t}$ and $R_{\mathbb{T}, t}$ are expressed by explicit kernels in Fourier space, to the best of my knowledge, both seem not to fall into an existing theory for either multi-parameter singular integral operators on tensor product spaces or multi-linear singular integral operators on Cartesian product spaces. In fact, the argument we proceed, following \cite{ESY2007} in spirit, is elementary and very involved. The proofs are quite technical and complicated, but essentially everything is based on two main ideas: integrate $\delta$-functions and estimate integration for rational functions with parameters.

In this paper, we will prove the equivalence between the Cauchy problem of \eqref{eq:NSEpressure-free} and that of \eqref{eq:NSEhierarchy} with a factorized divergence-free initial datum in $\mathcal{H}^1.$ As an application, we obtain a solution formula for \eqref{eq:NSE} with an initial datum $u_0 \in \mathcal{H}^1 (\mathbb{R}^3)$ that follows
\beq\label{eq:NSESolutionFormulaIntr}
u (t) = e^{t \triangle} u_0 + \sum^\8_{n=1} \sum_{\mathbb{T} \in \mathfrak{T}_{n,1}} C_{\mathbb{T}, t} u^{\otimes^{n+1}}_0
\eeq
in the sense of distributions for small $t>0.$ As noted above, every $\mathbb{T} \in \mathfrak{T}_{n,1}$ indicates a kind of processes of two-body interaction of $n+1$ ``particles". Note that every $C_{\mathbb{T}, t}$ encodes the two-body interaction of ``particles" in a concise form. Therefore, this solution formula is a physical expression and should be of helpful implications in computing the incompressible Navier-Stokes equation.

There are extensive works on the incompressible Navier-Stokes equation \eqref{eq:NSE}, we refer to \cite{LR2002, LR2016} and references therein (also see arXiv for more recent works). However, it seems that this is the first time to introduce the Navier-Stokes hierarchy \eqref{eq:NSEhierarchy} as a framework of studying \eqref{eq:NSE}. We expect that this framework will shed some new lights on the study of the incompressible Navier-Stokes equation, such as the multi-parameter singular integral operators will play a role in this study. In fact, the hierarchy \eqref{eq:NSEhierarchy} exhibits a certain kind of interference behavior arising from linear superposition of many-mode fluids. We will explain the physical meaning of \eqref{eq:NSEhierarchy} together with \eqref{eq:NSEhierarchyTreeExpanssionIntro} and \eqref{eq:NSESolutionFormulaIntr} elsewhere.

The organization of this paper is as follows. In the next section, we introduce some function spaces that will be used throughout the paper. In Section \ref{NShierarchy} we show that the interaction operators in the Navier-Stokes hierarchy are well-defined in the sense of distributions, and introduce the notion of solution to the Navier-Stokes hierarchy. In Section \ref{Unique}, we give a uniqueness theorem for the Navier-Stokes hierarchy and show the equivalence between the Cauchy problem of \eqref{eq:NSEpressure-free} and that of \eqref{eq:NSEhierarchy} with a factorized divergence-free initial datum in $\mathcal{H}^1.$ In Sections \ref{GraphCollisionoperator} and \ref{GraphicNShierachy}, we prove the formula \eqref{eq:NSEhierarchyTreeExpanssionIntro}. The aim of Section \ref{SpacetimeEstimate} is to establish {\it a prior} space-time estimates for the interaction operator and then present the proof of the uniqueness theorem mentioned above. Finally, in Section \ref{SolutionFormulaNSE}, we prove the main result of this paper, that is the solution formula \eqref{eq:NSESolutionFormulaIntr}. We include preliminaries on binary trees and some technical inequalities in Appendix.

\

{\it Preliminary notation.}\;Throughout the paper, we denote by $x= (x^1, x^2, x^3)$ a general variable in $\mathbb{R}^3$ and by $\vec{x}_k=(x_1,\ldots, x_k)$ a point in $\mathbb{R}^{3 k} = ( \mathbb{R}^3 )^k.$ For any $x , y \in \mathbb{R}^3$ we denote by $x \cdot y = \sum^3_{i=1} x^i y^i,$ and $|x| = (x \cdot x)^{\frac{1}{2}}$ with the convenience notation $x^2 =|x|^2.$ Moreover, we use the notation $\langle x \rangle = (1 + x^2)^{\frac{1}{2}}$ for all $x \in \mathbb{R}^3.$ For any $\vec{x}_k, \vec{y}_k \in \mathbb{R}^{3 k},$ we set $\langle \vec{x}_k, \vec{y}_k \rangle = \sum^k_{j=1} x_j \cdot y_j,$ and $| \vec{x}_k | = \langle \vec{x}_k, \vec{x}_k \rangle^{\frac{1}{2}}$ with the convenience notation $\vec{x}_k^2 = x^2_1 + \cdots + x^2_k.$

The following are some notations that will be used throughout the paper.

\begin{itemize}

\item $L^2 (\mathbb{R}^n)$ -- the Hilbert space of square integrable functions in $\mathbb{R}^n.$

\item $\mathcal{D} (\mathbb{R}^n)$ and $\mathcal{D}'(\mathbb{R}^n)$ -- the space of all smooth (i.e., infinitely differentiable) functions on $\mathbb{R}^n$ with compact support and its (locally convex) topological dual space.


\item $\mathcal{S} (\mathbb{R}^n)$ and $\mathcal{S}' (\mathbb{R}^n)$ -- the Schwarz space of all smooth functions of rapid decrease and the space of tempered distributions equipped with the Schwartz topology, respectively.

\item $H^\alpha (\mathbb{R}^n)$ -- $\alpha$-order Sobolev spaces for $\alpha \in \mathbb{R},$ defined as the closure of the Schwartz functions in $\mathbb{R}^n$ under their norms
$\| f \|_{H^\alpha (\mathbb{R}^n)} : = \|(1 - \triangle)^{\frac{\alpha}{2}} f \|_{L^2 (\mathbb{R}^n)}.$

\item $\hat{f}$ -- the Fourier transform of $f,$ defined by the formula
\beq\label{eq:Fouriertrans}
\hat{f} (\xi): = \int_{\mathbb{R}^n} d x f (x) e^{- \mathrm{i} x \cdot \xi},\quad \forall \xi \in \mathbb{R}^n.
\eeq
Here and in the following, $\mathrm{i} = \sqrt{-1}.$

\item $\mathbb{L}^2 (\real^3) = [L^2 (\real^3)]^3$ -- the Hilbert space of square integrable vector fields $u =(u_1, u_2, u_3)$ in $\mathbb{R}^3.$

\item $\bigtriangledown = (\partial_{x^1}, \partial_{x^2}, \partial_{x^3})$ -- the gradient operator in $\mathbb{R}^3.$

\item $R_{x^i} = \frac{\partial_{x^i}}{(- \triangle)^\frac{1}{2}}$ ($i =1,2,3$) -- the Riesz transform, i.e., for $f \in L^2 (\real^3)$
\beq\label{eq:Riesztrans}
\widehat{R_{x^i} f} (\xi) = \frac{\mathrm{i} \xi^i }{|\xi|} \hat{f} (\xi),\quad \forall \xi \in \mathbb{R}^3.
\eeq

\item The Leray projection $P$ on $\mathbb{L}^2 (\real^3)$ is defined by $P = \mathrm{id} + R \otimes R$ with $R = (R_{x^1}, R_{x^2}, R_{x^3}),$ i.e.,
\beq\label{eq:Lerayproj}
(P u)_j = u_j +  \sum^3_{i=1} R_{x^j}R_{x^i} u_i, \quad j =1,2,3,
\eeq
for $u = (u_1, u_2, u_3) \in \mathbb{L}^2 (\real^3).$

\item $u^{(k)} = u^{(k)} (\vec{x}_k)$ -- $k$-fold tensor functions in $\mathbb{R}^{3 k}$ with the convention notation
\be
u^{(k)} = (u_{i_1, \ldots, i_k}): = \big ( u_{i_1, \ldots, i_k} (\vec{x}_k):\; 1 \le i_1, \ldots, i_k \le 3 \big ).
\ee

\item $a^i b_i = \sum^3_{i=1} a^i b_i$-- Einstein's summation notation.

\item $p, q, r$ -- Fourier (momentum) variables in $\mathbb{R}^3$ with the convenience notation $p^2 = |p|^2.$

\item $f(p), f(q),$ or $f(r)$ -- the Fourier transform of $f$ in $\mathbb{R}^3,$ i.e., the usual hat indicating the Fourier transform is omitted.

\item $\vec{p}_k=(p_1,\ldots, p_k)$ -- a point in $\mathbb{R}^{3 k} = ( \mathbb{R}^3 )^k$ with the convenience notation $\vec{p}_k^2 = p^2_1 + \cdots + p^2_k.$

\item $u^{(k)} (\vec{p}_k), u^{(k)} (\vec{q}_k),$ or $u^{(k)} (\vec{r}_k)$ -- the Fourier transform of `$k$-body' velocity $u^{(k)} (\vec{x}_k)$ in position space, i.e.,
\be
u^{(k)} (\vec{p}_k) := \int d \vec{x}_k u^{(k)} (\vec{x}_k) e^{- \mathrm{i} \langle \vec{x}_k, \vec{p}_k \rangle}
\ee
with the slight abuse of notation of omitting the hat on the left hand side.

\item $d p$ -- the integration measure for the momentum variable $p$ which is always divided by $(2 \pi)^3,$ i.e.,
\beq\label{eq:momentummeasure}
d p = \frac{1}{(2 \pi)^3} d_{\mathrm{Leb}} p\; ,
\eeq
where $d_{\mathrm{Leb}}$ denotes the usual Lebesgue measure in $\mathbb{R}^3.$ With this notation, we have the Fourier inversion formula
\be
f (x) = \int d p f(p) e^{\mathrm{i} x \cdot p}, \quad \forall x \in \mathbb{R}^3.
\ee

\item $\delta (p)$ -- the delta function in the momentum space $\mathbb{R}^3$ corresponding to the measure $d p$ above, i.e.,
\be
\int d p f (p) \delta (p -q) = f (q),\quad \forall q \in \mathbb{R}^3
\ee
for smooth functions $f$ in the momentum space. The delta function in the position space $\mathbb{R}^3,$ $\delta (x),$ remains subordinated to the usual Lebesgue measure in $\mathbb{R}^3.$

\item $\tau$ -- frequency variables (dual variables to the time $t$) with the convenience notation
\beq\label{eq:frequencemeasure}
d \tau = \frac{1}{2 \pi} d_{\mathrm{Leb}} \tau,
\eeq
to which the delta function $\delta (\tau)$ of $\tau$-variables is subordinated.

\item Without specified otherwise, the integrals are over $\mathbb{R}^3,$ $\mathbb{R}^{3 k},$ or on $\mathbb{R},$ if the measure is $d x,\; d p;\; d \vec{x}_k,\; d \vec{p}_k;$ or $d \tau$ etc.

\end{itemize}

We use $X \lesssim Y$ to denote the inequality $X \le C Y$ for an absolute constant $C>0,$ and use $X \thicksim Y$ as shorthand for $X \lesssim Y \lesssim X.$ Also, $X \lesssim_{p,s,\ldots} Y$ denotes the inequality $X \le C_{p,s,\ldots} Y$ for some constant $C_{p,s,\ldots} >0$ depending on $p, s, \ldots .$


\section{Function spaces}\label{FuncSpace}

Given $k \ge 1,$ we define
\be
\mathbb{L}^2_{(k)} (\mathbb{R}^3) : = \big \{u^{(k)} = \big ( u_{i_1, \ldots, i_k} \big ):\; u_{i_1, \ldots, i_k} \in L^2 (\mathbb{R}^{3 k}),\; 1 \leq i_1, \ldots, i_k \leq 3 \big \},
\ee
equipped with the norm
\be
\| u^{(k)} \|_{\mathbb{L}^2_{(k)}} = \Big ( \sum_{1 \leq i_1, \ldots, i_k \leq 3} \| u_{i_1, \ldots, i_k} \|^2_{L^2} \Big )^{\frac{1}{2}},
\ee
and the associated inner product is given by
\be
\langle u^{(k)}, v^{(k)} \rangle_{\mathbb{L}^2_{(k)}} = \sum_{1 \leq i_1, \ldots, i_k \leq 3} \big \langle u_{i_1, \ldots, i_k} , v_{i_1, \ldots, i_k} \big \rangle_{L^2 (\mathbb{R}^{3k})}
\ee
for any $u^{(k)} = ( u_{i_1, \ldots, i_k} ), v^{(k)} = ( v_{i_1, \ldots, i_k} ) \in \mathbb{L}^2_{(k)} (\mathbb{R}^3).$ Note that $\mathbb{L}^2_{(1)} (\real^3) = \mathbb{L}^2 (\real^3),$ and we simply write $u^{(1)} =u.$

For $u^{(k)} \in \mathbb{L}^2_{(k)} (\mathbb{R}^3)$ we let
\be
\Theta_{\sigma} u^{(k)} (\vec{x}_k) = \big ( u_{i_{\sigma (1)}, \ldots, i_{\sigma (k)} } (x_{\sigma (1)}, \ldots , x_{\sigma (k)})
\big )
\ee
for a permutation $\sigma \in \Pi_k$ ($\Pi_k$ denotes the set of permutations on $k$ elements). Each $\Theta_{\sigma}$ is a unitary operator on $\mathbb{L}^2_{(k)} (\mathbb{R}^3).$ We define
\be
\mathfrak{L}^2_{(k)} ( \mathbb{R}^3 ) : = \big \{ u^{(k)} \in \mathbb{L}^2_{(k)} ( \mathbb{R}^3):\;
\Theta_{\sigma} u^{(k)} = u^{(k)}, ~~\forall \sigma \in \Pi_k \big \}
\ee
equipped with the norm $\| \cdot \|_{\mathbb{L}^2_{(k)}}.$ Then for each $u^{(k)} = \big ( u_{i_1, \cdots, i_k} \big ) \in \mathbb{L}^2_{(k)} ( \mathbb{R}^3),$ $u^{(k)} \in \mathfrak{L}^2_{(k)} (\mathbb{R}^3)$ if and only if for every $1 \le i_1, \cdots, i_k \le 3,$
\be
u_{i_1, \cdots, i_k} (x_1, \ldots, x_k ) = u_{i_{\sigma (1)}, \cdots, i_{\sigma (k)}} (x_{\sigma (1)}, \ldots, x_{\sigma (k)} )
\ee
for all $\sigma \in \Pi_k.$

We remark that for every $k \ge 2,$ $\mathbb{L}^2_{(k)} ( \mathbb{R}^3)$ (resp., $\mathfrak{L}^2_{(k)} (\mathbb{R}^3)$) is identified with the $k$-fold Hilbert tensor (resp., symmetric tensor) product space of $\mathbb{L}^2 (\real^3).$ In the sequel, we will write $\mathbb{L}^2_{(k)} = \mathbb{L}^2_{(k)} ( \mathbb{R}^3)$ and $\mathfrak{L}^2_{(k)} = \mathfrak{L}^2_{(k)} (\mathbb{R}^3)$ as shorthand.

For $k \ge 1,$ we denote by $\mathcal{D}_{(k)} (\mathbb{R}^3)$ (resp., $\mathcal{S}_{(k)} (\mathbb{R}^3)$) the space of $k$-fold tensor smooth and compactly supported functions (resp., Schwarz functions), i.e.,
\be
\mathcal{D}_{(k)} (\mathbb{R}^3) = \big \{ \phi^{(k)} = (\phi_{i_1, \ldots, i_k})_{1\le i_1, \ldots, i_k \le 3}:\; \phi_{i_1, \ldots, i_k} \in \mathcal{D} (\mathbb{R}^{3 k}) \big \}
\ee
and
\be
\mathcal{S}_{(k)} (\mathbb{R}^3) = \big \{ \phi^{(k)} = (\phi_{i_1, \ldots, i_k})_{1\le i_1, \ldots, i_k \le 3}:\; \phi_{i_1, \ldots, i_k} \in \mathcal{S} (\mathbb{R}^{3 k}) \big \}.
\ee
We may define the generalized function space $\mathcal{D}_{(k)}' (\mathbb{R}^3)$ as the topological dual space of $\mathcal{D}_{(k)} (\mathbb{R}^3)$ and the Schwarz generalized function space $\mathcal{S}_{(k)}' (\mathbb{R}^3)$ as that of $\mathcal{S}_{(k)} (\mathbb{R}^3),$ respectively. Similarly, for any $T>0,$ we can define $\mathcal{D}_{(k)} ((0,T) \times \mathbb{R}^3)$ and $\mathcal{S}_{(k)} ((0, T) \times \mathbb{R}^3).$

Given $k \ge 1,$ we write as shorthand $S_j = (1 - \bigtriangleup_j)^{\frac{1}{2}}$ for $1 \le j \le k,$ and
\be
S^{(k)}_\alpha = \prod^k_{j=1} S^\alpha_j
\ee
for any $k \ge 1$ and for any $\alpha \in \mathbb{R},$ with the convention $S^{(k)}= S^{(k)}_1.$ Here and in the following, $\Delta_j$ refers to the usual Laplace operator with respect to the $j$-th variables $x_j \in {\mathbb R}^3.$ We then define $\mathrm{H}^\alpha_{(k)} (\mathbb{R}^3)$ for $\alpha \in \mathbb{R}$ as the closure of the Schwartz functions $u \in \mathcal{S} (\mathbb{R}^{3 k})$ under the norm
\be
\| u \|_{\mathrm{H}^\alpha_{(k)}}: = \| S^{(k)}_\alpha u \|_{L^2 (\mathbb{R}^{3 k})},
\ee
These spaces are Hilbert spaces under the natural inner products. In fact, each $\mathrm{H}^\alpha_{(k)} (\mathbb{R}^3)$ is identified with the $k$-fold Hilbert tensor product spaces of the usual Sobolev space $H^\alpha (\mathbb{R}^3).$ We will write $\mathrm{H}^\alpha_{(k)} = \mathrm{H}^\alpha_{(k)} (\mathbb{R}^3)$ as shorthand.

Next, we define $\mathbb{H}^\alpha_{(k)} (\mathbb{R}^3)$ for $\alpha \in \mathbb{R}$ to be the space
\be
\mathbb{H}^\alpha_{(k)} (\mathbb{R}^3): = \big \{ u^{(k)} = (u_{i_1, \ldots, i_k}):\; u_{i_1, \ldots, i_k} \in \mathrm{H}^\alpha_{(k)} (\mathbb{R}^3) \big \}
\ee
with the norm
\be
\| u^{(k)} \|_{\mathbb{H}^\alpha_{(k)}} = \bigg ( \sum_{1 \le i_1, \ldots, i_k \le 3} \| u_{i_1, \ldots, i_k} \|^2_{\mathrm{H}^\alpha_{(k)}} \bigg )^{\frac{1}{2}}.
\ee
Moreover, we define
\be
\mathfrak{H}^\alpha_{(k)} (\mathbb{R}^3) = \big \{ u^{(k)} = (u_{i_1, \ldots, i_k}) \in \mathbb{H}^\alpha_{(k)} (\mathbb{R}^3):\; \Theta_{\sigma} u^{(k)} = u^{(k)}, ~~\forall \sigma \in \Pi_k \big \}.
\ee
Thus $\mathbb{H}^\alpha_{(k)} (\mathbb{R}^3)$'s (resp., $\mathfrak{H}^\alpha_{(k)} (\mathbb{R}^3)$'s) generalize the spaces $\mathbb{L}^2_{(k)} (\mathbb{R}^3)$ (resp., $\mathfrak{L}^2_{(k)} (\mathbb{R}^3)$), which correspond to the cases $\alpha =0.$ It can be shown that for $\alpha > 0,$ the Banach space dual of $\mathbb{H}^\alpha_{(k)} (\mathbb{R}^3)$ (resp., $\mathfrak{H}^\alpha_{(k)} (\mathbb{R}^3)$) is identified with $\mathbb{H}^{-\alpha}_{(k)} (\mathbb{R}^3)$ (resp., $\mathfrak{H}^{-\alpha}_{(k)} (\mathbb{R}^3)$). In what follows, we will write as shorthand respectively $\mathbb{H}^\alpha_{(k)}$ for $\mathbb{H}^\alpha_{(k)} (\mathbb{R}^3),$ $\mathfrak{H}^\alpha_{(k)}$ for $\mathfrak{H}^\alpha_{(k)} (\mathbb{R}^3),$ $\mathbb{H}^\alpha = \mathbb{H}^\alpha_{(1)},$ $\mathfrak{H}^\alpha = \mathfrak{H}^\alpha_{(1)},$ etc.

\section{The Navier-Stokes hierarchy}\label{NShierarchy}

\subsection{Interaction operator}

Given $k \ge 1,$ we define
\beq\label{eq:nabla_j}
\bigtriangledown_j \cdot u^{(k+1)} : = \big ( \partial_{x^i_j} u_{i_1, \ldots, i_k, i} \big ),\quad 1 \le j \le k +1,
\eeq
with Einstein's summation notation $ \partial_{x^i_j} u_{i_1, \ldots, i_k, i} = \sum^3_{i =1} \partial_{x^i_j} u_{i_1, \ldots, i_k, i},$ and
\beq\label{eq:P_j}
P_j u^{(k)} : = u^{(k)} + \big ( R_{x^{i_j}_j} R_{x^\el_j} u_{i_1, \ldots, i_{j-1}, \el, i_{j+1}, \ldots, i_k} \big ), \quad 1 \le j \le k,
\eeq
where Einstein's summation notation is used again for the index $\el.$ For any $1 \le j \le k,$ put
\beq\label{eq:W_j+}
\big ( W^+_{j,k+1} u^{(k+1)} \big ) (\vec{x}_k ) := - \bigtriangledown_j \cdot u^{(k+1)} (\vec{x}_k, x_j) = \big ( - \partial_{x^i_j} u_{i_1, \ldots, i_k, i} (\vec{x}_k, x_j)\big )
\eeq
and
\beq\label{eq:W_j-}
\big ( W^-_{j,k +1} u^{(k+1)} \big ) (\vec{x}_k ) := \big ( - R_{x^{i_j}_j} R_{x^\el_j} \partial_{x^i_j} u_{i_1, \ldots, i_{j-1}, \el, i_{j+1}, \ldots, i_k, i} (\vec{x}_k, x_j)\big ).
\eeq
We then define
\beq\label{eq:W_j}
\big ( W_{j, k+1} u^{(k+1)} \big ) (\vec{x}_k ):= W^+_{j, k+1} u^{(k+1)} (\vec{x}_k ) + W^-_{j,k+1} u^{(k+1)} (\vec{x}_k )= - P_j \bigtriangledown_j \cdot u^{(k+1)} (\vec{x}_k, x_j),
\eeq
for any $1 \le j \le k.$

We introduce the interaction operator $W^{(k)}$ as
\beq\label{eq:Woperatordf}
W^{(k)}: = \sum^k_{j=1} W_{j, k+1}
\eeq
which describes interactions between the first $k$ `particles' and the $(k+1)$-th `particle'. The action of $W^{(k)}$ is defined through a limiting procedure. Since the expression of $W^{(k)}$ acting on smooth tensor functions $u^{(k+1)} = \big ( u_{i_1, \ldots, i_k, i_{k+1}} \big ) \in \mathcal{S}_{(k+1)} (\mathbb{R}^3)$ is
\be
\big ( W^{(k)} u^{(k+1)} \big ) (\vec{x}_k) = - \sum^k_{j=1} P_j \bigtriangledown_j \cdot u^{(k+1)} ( \vec{x}_k, x_j ),
\ee
the action of $W^{(k)}$ for general $u^{(k+1)}$ is then formally given by
\beq\label{eq:Woperator}
\big ( W^{(k)} u^{(k+1)} \big ) (\vec{x}_k) = - \sum^k_{j=1} P_j \bigtriangledown_j \cdot \int d x_{k+1} \delta (x_j - x_{k+1}) u^{(k+1)} ( \vec{x}_k, x_{k+1}).
\eeq
Since $\mathcal{S}_{(k+1)} (\mathbb{R}^3)$ is dense in $\mathbb{L}^2_{(k+1)},$ $W_{j, k+1}$ is a densely defined operator from $\mathbb{L}^2_{(k+1)}$ into $\mathbb{L}^2_{(k)}$ for any $1 \le j \le k,$ and so does $W^{(k)}.$

In the following, we show that $W^{(k)}$ is well defined in $\mathbb{H}^1_{(k+1)}$ by \eqref{eq:Woperator} in the sense of distributions. To this end, we choose a nonnegative function $h \in \mathcal{D} (\mathbb{R}^3)$ supported in the unit ball $B = \{x \in \mathbb{R}^3:\; |x| \le 1\}$ such that $\int h d x =1.$ For any $\epsilon >0,$ we set $\delta_\epsilon (x) = \epsilon^{-3} h (\epsilon^{-1} x).$ Then for $u^{(k+1)} \in \mathbb{L}^2_{(k+1)},$ we define
\beq\label{eq:k-WoperatorLim}
\big ( W^{(k)} u^{(k+1)} \big ) (\vec{x}_k) = - \lim_{\epsilon \to 0} \sum^k_{j=1} P_j \bigtriangledown_j \cdot \int d x_{k+1} \delta_\epsilon (x_j - x_{k+1}) u^{(k+1)} (\vec{x}_k, x_{k+1})
\eeq
in the sense of distributions, i.e., for every $\phi^{(k)} = (\phi_{i_1, \ldots, i_k}) \in \mathcal{D}_{(k)} (\mathbb{R}^3),$
\be\begin{split}
\big \langle & \phi_{i_1, \ldots, i_k}, \big ( W^{(k)} u^{(k+1)} \big )_{i_1, \ldots, i_k} \big \rangle_{L^2 (\mathbb{R}^{3 k})}\\
& = \lim_{\epsilon \to 0} \sum^k_{j=1} \bigg [ \sum^3_{i=1} \int d x_{k+1} \big \langle \partial_{x^i_j} \phi_{i_1, \ldots, i_k},  \delta_\epsilon (x_j - x_{k+1}) u_{i_1, \ldots, i_k, i} (\cdot, x_{k+1}) \big \rangle_{L^2 (\mathbb{R}^{3 k})}\\
& \quad + \sum^3_{\el =1} \sum^3_{i=1} \int d x_{k+1} \big \langle R_{x^{i_j}_j} R_{x^{\el}_j} \partial_{x^i_j} \phi_{i_1, \ldots, i_k}, \delta_\epsilon (x_j - x_{k+1}) u_{i_1, \ldots, i_{j-1}, \el, i_{j+1}, \ldots, i_k, i} (\cdot, x_{k+1}) \big \rangle_{L^2 (\mathbb{R}^{3 k})} \bigg ]
\end{split}
\ee
for all $1 \le i_1, \ldots, i_k \le 3.$

\begin{prop}\label{prop:Wdfwell}
For every $k \ge1,$ $W^{(k)}$ is well defined for all $u^{(k+1)} \in \mathbb{H}^1_{(k+1)}$ in the sense of distributions, such that
\beq\label{eq:Woperatorweak-estimate}
\big | \big \langle \phi^{(k)}, W^{(k)} u^{(k+1)} \big \rangle_{\mathbb{L}^2_{(k)}} \big | \le  C k^\frac{1}{2} \| \phi^{(k)} \|_{\mathbb{H}^1_{(k)}} \| u^{(k+1)} \|_{\mathbb{H}^1_{(k+1)}},
\eeq
where $C>0$ is a universal constant.

Consequently, for any $k \ge 1$ the operator $W^{(k)},$ originally defined on Schwarz functions, can be extended to a bounded operator from $\mathbb{H}^1_{(k+1)}$ into $\mathbb{H}^{-1}_{(k)}.$
\end{prop}

\begin{proof} By a standard argument, it follows from Lemmas \ref{lem:WoperatordfWell} and \ref{lem:Potentialestimate} that the limit \eqref{eq:k-WoperatorLim} exists for every $u^{(k+1)} \in \mathbb{H}^1_{(k+1)}$ in the sense of distributions, and is independent of the choice of $h \in \mathcal{D} (\mathbb{R}^3).$ Hence, the operator $W^{(k)}$ is well defined for all $u^{(k+1)} \in \mathbb{H}^1_{(k+1)}.$

To prove \eqref{eq:Woperatorweak-estimate}, by Lemma \ref{lem:Potentialestimate} again, we have for each term in \eqref{eq:k-WoperatorLim},
\be\begin{split}
\big | \big \langle & \phi_{i_1, \ldots, i_k}, \big ( W^{(k)} u^{(k+1)} \big )_{i_1, \ldots, i_k} \big \rangle_{L^2 (\mathbb{R}^{3 k})} \big |\\
& \le C \sum^k_{j=1} \sum^3_{i=1} \| \partial_{x^i_j} \phi_{i_1, \ldots, i_k} \|_{L^2 (\mathbb{R}^{3k})} \big [ \| (1- \triangle_j)^{\frac{1}{2}} (1- \triangle_{k+1})^{\frac{1}{2}} u_{i_1, \ldots, i_k, i} \|_{L^2 (\mathbb{R}^{3(k+1)})}\\
& \quad \quad \quad \quad \quad \quad + \sum^3_{\el =1} \| (1- \triangle_j)^{\frac{1}{2}} (1- \triangle_{k+1})^{\frac{1}{2}} u_{i_1, \ldots, i_{j-1}, \el, i_{j+1}, \ldots, i_k, i} \|_{L^2 (\mathbb{R}^{3(k+1)})} \big ].
\end{split}\ee
This yields \eqref{eq:Woperatorweak-estimate} and completes the proof.
\end{proof}

\subsection{Definition of solution}

In terms of $W^{(k)}$'s, we can rewrite the Navier-Stokes hierarchy \eqref{eq:NSEhierarchy} as
\beq\label{eq:NSEhierarchyW}
\partial_t u^{(k)} (t) = \triangle^{(k)} u^{(k)} (t) + W^{(k)} u^{(k+1)} (t)
\eeq
for all $k \ge 1,$ where and in the following,
\be
\triangle^{(k)}: = \sum^k_{j=1} \triangle_j.
\ee

In the sequel, we define the notion of weak solution to the Navier-Stokes hierarchy \eqref{eq:NSEhierarchyW}. At first, with the help of Proposition \ref{prop:Wdfwell}, we give a restriction assumption that will be proposed on suitable solutions to \eqref{eq:NSEhierarchyW}.

\begin{defi}\label{df:Admisscond}
A sequence $(u^{(k)})_{k \ge 1} \in \prod_{k \ge 1} \mathbb{H}^1_{(k)}$ is said to be consistent if
\beq\label{eq:Admisscondition}
\langle u^{(k)}, W^+_{j, k+1} u^{(k+1)} \rangle_{\mathbb{L}^2_{(k)}} = 0
\eeq
for every $k \ge 1$ and for all $1 \le j \le k.$
\end{defi}

\begin{rk}\label{rk:ConsistentCondition}\rm
Note that if $u \in \mathbb{H}^1 (\mathbb{R}^3)$ with $\nabla \cdot u =0,$ then $(u^{\otimes k})_{k \ge 1}$ is consistent.
\end{rk}

We refer to \cite{DU1977} for the details of strongly measurable functions with values in a Banach space and the Bochner and Pettis integrals for them.

\begin{defi}\label{df:NShierarchyweaksolution}
Let $0 < T \le \8.$ A weak solution on $(0, T)$ for the Navier-Stokes hierarchy \eqref{eq:NSEhierarchyW} is defined as a sequence of strongly measurable functions $u^{(k)} (t)$ on $(0, T)$ with values in $\mathfrak{H}^1_{(k)}$ for $k \ge 1,$ satisfying the following properties:
\begin{enumerate}[{\rm $1)$}]


\item For every $k \ge 1,$ one has
\beq\label{eq:NShierarchyweaksolution}
\int^T_0 \big [ \langle \partial_t \phi + \triangle^{(k)} \phi, u^{(k)} \rangle_{\mathbb{L}^2_{(k)}} + \langle \phi, W^{(k)} u^{(k+1)} \rangle_{\mathbb{L}^2_{(k)}} \big ] d t = 0,
\eeq
for any $\phi \in \mathcal{D}_{(k)} ((0, T) \times \mathbb{R}^3)$ with the divergence free property that follows
\beq\label{eq:TextfunctDivfreeCondit}
\sum^3_{\el =1} \partial_{x^{\el}_j} \phi_{i_1, \ldots ,i_{j-1}, \el , i_{j+1}, \ldots,  i_k} (t) = 0
\eeq
for all $0 < t <T$ and for every $1 \le i_1, \ldots, i_k \le 3.$

\item  For any $k \ge 1$ the divergence free conditions
\beq\label{eq:NShierarchyDivfreeCondit}
\sum^3_{\el =1} \partial_{x^{\el}_j} u^{(k)}_{i_1, \ldots ,i_{j-1}, \el , i_{j+1}, \ldots,  i_k} (t) = 0
\eeq
hold true for every $t \in (0, T)$ and all $1 \le j \le k,$ where $1 \le i_1, \ldots, i_k \le 3.$

\item The sequences $(u^{(k)}(t))_{k \ge 1}$ are consistent for every $t \in (0, T).$

\end{enumerate}

As for $T = \8,$ the solution is called a global weak solution.
\end{defi}

The equality \eqref{eq:NShierarchyweaksolution} means that $( u^{(k)} (t) )_{k \geq 1}$ satisfies \eqref{eq:NSEhierarchyW} in the sense of distributions. Note that $\langle \phi, W^{(k)} u^{(k+1)} \rangle_{(k)}$ is well defined in \eqref{eq:NShierarchyweaksolution}, since $W^{(k)} u^{(k+1)} \in \mathbb{H}^{-1}_{(k)}$ for $u^{(k+1)} \in \mathbb{H}^1_{(k+1)}.$

\begin{rk}\label{rk:NShinvariance}\rm
Any weak solution for the Navier-Stokes hierarchy is shift-invariant, i.e., if $(u^{(k)} (t, \vec{x}_k) )_{k \ge 1} $ is a weak solution on $(0, T),$ so does $(u^{(k)} (t + t_0, \vec{x}_k- \vec{x}_{k0}) )_{k \ge 1}$ on $(0, T - t_0),$  where $t_0 \in (0, T)$ and $\vec{x}_{k0} \in (\mathbb{R}^3)^k$ are fixed for all $k \ge 1.$ Moreover, for every $\lambda >0$ putting
\beq\label{eq:scalingsolution}
u^{(k)}_\lambda (t, \vec{x}_k) = \lambda^k u^{(k)} (\lambda^2 t, \lambda \vec{x}_k), \quad \forall k \ge 1
\eeq
one has that for each $\lambda >0,$ $ (u^{(k)}_\lambda)_{k \ge 1}$ is a weak solution $(0, T)$ to the Navier-Stokes hierarchy \eqref{eq:NSEhierarchyW} when $(u^{(k)})_{k \ge 1}$ does on $(0, \lambda^2 T),$ and vice versa. Thus, the Navier-Stokes hierarchy \eqref{eq:NSEhierarchyW} has the usual space-time dilation invariance.
\end{rk}

\begin{defi}\label{df:NShierarchyCauchyprob}
Let $0 < T \le \8.$ Suppose $\mathfrak{U}_0 = (u^{(k)}_0)_{k \ge 1} \in \prod_{k \ge 1} \mathfrak{L}^2_{(k)} (\mathbb{R}^3)$ such that for every $k \ge 1,$
\beq\label{eq:NShierarchyIVdivfree}
\sum^3_{\el =1} \partial_{x^{\el}_j} (u^{(k)}_0)_{i_1, \ldots ,i_{j-1}, \el , i_{j+1}, \ldots,  i_k} = 0
\eeq
in the sense of distributions for all $j =1, \ldots, k,$ and for any $1 \le i_1, \ldots, i_k \le 3.$ A weak solution on $[0, T)$ for the Cauchy problem of the Navier-Stokes hierarchy \eqref{eq:NSEhierarchyW} with the initial data $\mathfrak{U}_0$ that is
\beq\label{eq:NSherarchyIC}
u^{(k)} (0) = u^{(k)}_0
\eeq
for any $k \ge 1,$ is by definition a weak solution $( u^{(k)} (t) )_{k \geq 1}$ on $(0,T)$ to the Navier-Stokes hierarchy \eqref{eq:NSEhierarchyW} such that for every $k \ge 1,$ $u^{(k)} \in L^\8 \big ((0, T); \mathbb{L}^2_{(k)} (\mathbb{R}^3) \big )$ and
\be
\lim_{t \to 0} u^{(k)} (t) = u^{(k)}_0
\ee
in the weak topology of $\mathbb{L}^2_{(k)} (\mathbb{R}^3).$

When $T = \8,$ the solution is said to be a global solution.
\end{defi}

We shall study the Cauchy problem of \eqref{eq:NSEhierarchyW} with the initial value \eqref{eq:NSherarchyIC} by transforming \eqref{eq:NSEhierarchyW} into the integral Navier-Stokes hierarchy that follows
\beq\label{eq:NShierarchyIntEqua}
u^{(k)} (t) = \mathcal{T}^{(k)} (t) u^{(k)}_0 + \int^t_0 d s \mathcal{T}^{(k)} ( t-s) W^{(k)} u^{(k+1)} (s),\quad \forall k \ge 1,
\eeq
where the free evolution operator $\mathcal{T}^{(k)} (t)$ is defined on $\mathbb{L}^2_{(k)}$ for every $t \ge 0$ by
\beq\label{eq:FreeEvoOpe}
\mathcal{T}^{(k)} (t) u^{(k)}: = \big ( e^{t \triangle^{(k)}} u_{i_1, \ldots, i_k} \big )
\eeq
for every $u^{(k)} = \big (u_{i_1, \ldots, i_k} \big ) \in \mathbb{L}^2_{(k)}.$

\begin{prop}\label{prop:Tcontra}
Given any fixed $k \ge 1$ and $t \ge 0,$ $\mathcal{T}^{(k)} (t)$ is a contraction on $\mathbb{H}^\alpha_{(k)}$ for any $\alpha \in \mathbb{R},$ i.e.,
\be
\| \mathcal{T}^{(k)} (t) u^{(k)} \|_{\mathbb{H}^\alpha_{(k)}} \le \| u^{(k)} \|_{\mathbb{H}^\alpha_{(k)}}
\ee
for all $u^{(k)} = \big (u_{i_1, \ldots, i_k} \big ) \in \mathbb{H}^\alpha_{(k)}.$
\end{prop}
\begin{proof}
Fix $k \ge 1$ and let $t \ge 0.$ Note that for any $1 \le i_1, \ldots, i_k \le 3,$
\be
\| e^{t \triangle^{(k)}} u_{i_1, \ldots, i_k} \|_{L^2 (\mathbb{R}^{3 k})} \le \| u_{i_1, \ldots, i_k} \|_{L^2 (\mathbb{R}^{3 k})},
\ee
it follows that $\mathcal{T}^{(k)} (t)$ is a contraction on $\mathbb{L}^2_{(k)}.$ Since $e^{t \triangle^{(k)}}$ commutates with $S_j = (1 - \triangle_j)^\frac{1}{2}$ for all $1 \le j \le k,$ it follows that $\mathcal{T}^{(k)} (t)$ is a contraction on $\mathbb{H}^\alpha_{(k)}$  for every $\alpha \in \mathbb{R}.$
\end{proof}

By Propositions \ref{prop:Wdfwell} and \ref{prop:Tcontra}, for $k \ge 1$ and $t \ge s > 0$ one has
\be
\mathcal{T}^{(k)} ( t-s) W^{(k)} u^{(k+1)} (s) \in \mathbb{H}^{-1}_{(k)}
\ee
if $u^{(k+1)}(s) \in \mathbb{H}^1_{(k+1)}.$ Thus, when $t \mapsto u^{(k)} (t)$ is strongly measurable on $(0, T)$ with values in $\mathbb{H}^1_{(k)}$ for $k \ge 1,$ the equality \eqref{eq:NShierarchyIntEqua} can be expressed as that for each $t > 0,$ there is a representation of
\be
\int^t_0 d s \mathcal{T}^{(k)} ( t-s) W^{(k)} u^{(k+1)} (s)
\ee
which lies in $\mathbb{H}^1_{(k)}$ and such that \eqref{eq:NShierarchyIntEqua} holds in the sense of distributions, due to the fact that both $u^{(k)} (t)$ and $\mathcal{T}^{(k)} (t) u^{(k)}_0$ are in $\mathbb{H}^1_{(k)}$ for $t > 0,$ where $u^{(k)}_0 \in \mathbb{L}^2_{(k)}.$

\begin{defi}\label{df:NShierarchyMildsolution}
Let $0 < T \le \8.$ Suppose $\mathfrak{U}_0 = (u^{(k)}_0)_{k \ge 1} \in \prod_{k \ge 1} \mathfrak{L}^2_{(k)} (\mathbb{R}^3)$ satisfies the divergence-free condition \eqref{eq:NShierarchyIVdivfree} for all $k \ge 1.$ A mild solution on $[0, T)$ to the integral Navier-Stokes hierarchy \eqref{eq:NShierarchyIntEqua} with the prescribed initial condition $\mathfrak{U}_0$ is defined as a sequence of strongly measurable functions $u^{(k)} (t)$ on $(0, T)$ with values in $\mathfrak{H}^1_{(k)} (\mathbb{R}^3)$ for $k \ge 1,$ satisfying the following conditions:
\begin{enumerate}[{\rm $1)$}]


\item The integral equation \eqref{eq:NShierarchyIntEqua} holds for all $t \in (0, T )$ in the sense of distributions.

\item The sequences $(u^{(k)}(t))_{k \ge 1}$ are consistent for every $t \in (0, T).$

\item For every $k \ge 1,$ $u^{(k)} \in L^\8 \big ((0, T); \mathbb{L}^2_{(k)} (\mathbb{R}^3) \big )$ and
\be
\lim_{t \to 0} u^{(k)} (t) = u^{(k)}_0
\ee
in the weak topology of $\mathbb{L}^2_{(k)} (\mathbb{R}^3).$

\end{enumerate}

When $T = \8,$ the solution is called a global mild solution.
\end{defi}

Next, we show the equivalence between a weak solution to the Cauchy problem of the Navier-Stokes hierarchy \eqref{eq:NSEhierarchyW} and a mild solution to the integral Navier-Stokes hierarchy \eqref{eq:NShierarchyIntEqua}.

\begin{prop}\label{prop:solutionequivalence}
Let $u^{(k)}_0 \in \mathfrak{L}^2_{(k)} (\mathbb{R}^3)$ satisfying \eqref{eq:NShierarchyIVdivfree} for all $k \ge 1.$ A sequence of strongly measurable functions $u^{(k)} (t)$ on $(0, T)$ with values in $\mathfrak{H}^1_{(k)} (\mathbb{R}^3)$ for $k \ge 1$ is a weak solution to the Cauchy problem of the Navier-Stokes hierarchy \eqref{eq:NSEhierarchyW} with the initial data $(u^{(k)} (0))_{k \ge 1} = (u^{(k)}_0)_{k \ge 1},$ if and only if it is a mild solution to the integral Navier-Stokes hierarchy \eqref{eq:NShierarchyIntEqua} with the prescribed initial condition $(u^{(k)}_0)_{k \ge 1}.$
\end{prop}

\begin{proof}
The proof is straightforward by using the standard argument (cf. \cite[Theorem 11.2]{LR2002}). We include the details for the sake of convenience. Let $0 < T \le \8.$ Suppose $\mathfrak{U}_0 = (u^{(k)}_0)_{k \ge 1} \in \prod_{k \ge 1} \mathfrak{L}^2_{(k)} (\mathbb{R}^3)$ satisfying the divergence-free condition \eqref{eq:NShierarchyIVdivfree} for all $k \ge 1.$ At first, we assume that $(u^{(k)} (t))_{k \ge 1}$ is a weak solution on $[0, T)$ to the Cauchy problem for \eqref{eq:NSEhierarchyW} with the initial datum $\mathfrak{U}_0.$ By Propositions \ref{prop:Wdfwell} and \ref{prop:Tcontra}, for $k \ge 1$ and for $t \in (0,T)$ one has
\be
s \mapsto \mathcal{T}^{(k)} ( t-s) W^{(k)} u^{(k+1)} (s)
\ee
is strongly measurable on $(0, t)$ with values in $\mathbb{H}^{-1}_{(k)},$ since $s \mapsto u^{(k+1)} (s)$ is strongly measurable on $(0, T)$ with values in $\mathbb{H}^1_{(k+1)}.$ We put
\be
F ( u^{(k)}) = \mathcal{T}^{(k)} (t) u^{(k)}_0 + \int^t_0 d s \mathcal{T}^{(k)} ( t-s) W^{(k)} u^{(k+1)} (s)
\ee
in the sense of distributions. Then we have
\be
\partial_t (u^{(k)} - F ( u^{(k)})) = \triangle^{(k)} (u^{(k)} - F ( u^{(k)}) )
\ee
with $\lim_{t \to 0} (u^{(k)} - F ( u^{(k)}) ) =0$ in the sense of distributions. By the standard argument (cf. \cite[p. 113]{LR2002}), we conclude that $u^{(k)} = F ( u^{(k)})$ and, therefore $(u^{(k)} (t))_{k \ge 1}$ is a mild solution to \eqref{eq:NShierarchyIntEqua} with the prescribed initial condition $\mathfrak{U}_0.$

Conversely, let $(u^{(k)} (t))_{k \ge 1}$ be a mild solution on $[0, T)$ to \eqref{eq:NShierarchyIntEqua} with the prescribed initial condition $\mathfrak{U}_0.$ For any $k \ge 1$ we have
\be
\partial_t F ( u^{(k)}) = \triangle^{(k)} F ( u^{(k)}) + W^{(k)} u^{(k+1)} (t)
\ee
in the sense of distributions. Since $u^{(k)} = F ( u^{(k)})$ for all $k \ge 1,$ it follows that $(u^{(k)} (t))_{k \ge 1}$ is a weak solution to \eqref{eq:NSEhierarchyW} with the initial data $\mathfrak{U}_0.$
\end{proof}

\section{Uniqueness and equivalence for the Navier-Stokes hierarchy}\label{Unique}

The following is a uniqueness theorem for the Navier-Stokes hierarchy.

\begin{thm}\label{th:unique}
Let $T>0.$ Assume that $\mathfrak{U}_0 = (u^{(k)}_0)_{k \ge 1} \in \prod_{k \ge 1} \mathfrak{H}^1_{(k)}$ is consistent and for every $k \ge 1,$ $u^{(k)}_0 $ satisfies the divergence-free condition \eqref{eq:NShierarchyIVdivfree} and
\beq\label{eq:Initialvaluebound}
\| u^{(k)}_0 \|_{\mathbb{H}^1_{(k)}} \le C^k
\eeq
where $C>0$ is a constant independent of $k.$ Then the integral Navier-Stokes hierarchy 
\be
u^{(k)} (t) = \mathcal{T}^{(k)} (t) u^{(k)}_0 + \int^t_0 \mathcal{T}^{(k)} ( t-s) W^{(k)} u^{(k+1)} (s) d s,\quad \forall k \ge 1,
\ee
has at most one mild solution $\mathfrak{U}(t) = (u^{(k)} (t))_{k \ge 1}$ in $[0,T)$ with $\mathfrak{U}(0) = \mathfrak{U}_0$ such that for every $k \ge 1,$ $ u^{(k)} \in L^{\8} ([0, T), \mathfrak{H}^1_{(k)})$ and satisfies the bound
\beq\label{eq:SolutionL2bound}
\| u^{(k)} \|_{L^\8([0,T), \mathbb{H}^1_{(k)})} \le C^k.
\eeq
\end{thm}

As in \cite{Chen2013}, we define
\be
\mathfrak{H}^1_{(\8)} = \Big \{ (u^{(k)})_{k \ge 1} \in \prod_{k \ge 1} \mathfrak{H}^1_{(k)} :\; \exists \lambda >0,\; \sum_{k \ge 1} \frac{1}{\lambda^k} \| u^{(k)} \|_{\mathbb{H}^1_{(k)}} < \8 \Big \}
\ee
equipped with
\be
\| (u^{(k)})_{k \ge 1} \|_{\mathfrak{H}^1_{(\8)}} = \inf \Big \{\lambda >0:\; \sum_{k \ge 1} \frac{1}{\lambda^k} \| u^{(k)} \|_{\mathbb{H}^1_{(k)}} \le 1 \Big \}.
\ee
Note that $\| \cdot \|_{\mathfrak{H}^1_{(\8)}}$ is not actually a norm but a $\rm (F)$-norm in $\mathfrak{H}^1_{(\8)}.$ Thus, $\mathfrak{H}^1_{(\8)}$ is a $F$-space (cf. \cite[Chapter II]{DS1964}).

For any $u \in \mathbb{H}^1 (\mathbb{R}^3),$ it is easy to check that $(u^{\otimes^k})_{k \ge 1} \in \mathfrak{H}^1_{(\8)}$ and
\be
\| (u^{\otimes^k})_{k \ge 1} \|_{\mathfrak{H}^1_{(\8)}} = 2 \| u \|_{\mathbb{H}^1}.
\ee
Namely, $\| \mathfrak{U} \|_{\mathfrak{H}^1_{(\8)}}$ is compatible with the Sobolev norm $\| u \|_{\mathbb{H}^1}$ for factorized hierarchies $\mathfrak{U} = (u^{\otimes^k})_{k \ge 1}$ with $u \in \mathbb{H}^1.$

\begin{rk}\label{rk:F-space}\rm
By using this $F$-space of Sobolev-type, Theorem \ref{th:unique} can be reformulated as follows: {\it The integral Navier-Stokes hierarchy \eqref{eq:NShierarchyIntEqua} has uniqueness of mild solutions in $\mathfrak{H}^1_{(\8)}.$}
\end{rk}

Combining Theorem \ref{th:unique} and Proposition \ref{prop:solutionequivalence} yields the following result.

\begin{cor}\label{cor:NES=NEH}
Let $u_0 \in \mathbb{H}^1 (\mathbb{R}^3)$ such that $\nabla \cdot u_0 =0.$ Let $u (t)$ be the unique (mild) solution in $C ([0, T^*), \mathbb{H}^1 (\mathbb{R}^3))$ for the Navier-Stokes equation \eqref{eq:NSEpressure-free} with the initial datum $u(0) = u_0,$ where $T^*$ is the maximal life-time of $u(t).$ Let $u^{(k)} (t) = u(t)^{\otimes^k}$ for every $k \ge 1.$ Then $\mathfrak{U} (t) = (u^{(k)} (t) )_{k \ge 1}$ is a unique weak solution in $\mathfrak{H}^1_{(\8)}$ for the Cauchy problem of the Navier-Stokes hierarchy \eqref{eq:NSEhierarchyW} on $[0, T^*)$ with the initial datum $\mathfrak{U} (0) = (u^{\otimes^k}_0)_{k \ge 1}.$
\end{cor}

\begin{rk}\label{rk:NSE-NSHEquivalence}\rm
This corollary shows that the initial problem for the Navier-Stokes hierarchy \eqref{eq:NSEhierarchyW} in $\mathfrak{H}^1_{(\8)}$ with a factorized divergence-free initial datum is equivalent to the Cauchy problem of the Navier-Stokes equation \eqref{eq:NSEpressure-free} in $\mathcal{H}^1.$
\end{rk}

The proof of Theorem \ref{th:unique} is based on Duhamel expansion for the solution to \eqref{eq:NShierarchyIntEqua}. Any solution $( u^{(k)} (t) )_{k \ge 1}$ to \eqref{eq:NShierarchyIntEqua} can be formally expanded in a Duhamel-type series, i.e., for any $k \ge 1,$
\beq\label{eq:DuhamelExpand}
\begin{split}
u^{(k)} (t) = \mathcal{T}^{(k)} (t) u^{(k)}_0 + \sum^{n-1}_{j=1} \int^t_0 d t_1 \int^{t_1}_0 & d t_2 \cdots  \int^{t_{j-1}}_0 d t_j \mathcal{T}^{(k)} ( t- t_1) W^{(k)}\cdots \\
& \times \mathcal{T}^{(k+j-1)} ( t_{j-1}-t_j) W^{(k+j-1)} \mathcal{T}^{(k+j)} ( t_j) u^{(k+j)}_0\\
+ \int^t_0 d t_1 \int^{t_1}_0 & d t_2 \cdots \int^{t_{n-1}}_0 d t_n \mathcal{T}^{(k)} ( t-t_1) W^{(k)} \cdots \\
& \times \mathcal{T}^{(k+n-1)} ( t_{n-1}-t_n ) W^{(k + n - 1)} u^{(k + n)} (t_n) \end{split}
\eeq
for every $n >1,$ with the convention $t_0 =t.$ Note that the terms in the summation contain only the initial data, which are said to be {\it fully expanded}, while the last error term involves the function at intermediate time $t_n.$

We want to estimate the terms on the right hand side of \eqref{eq:DuhamelExpand}, this will be mostly done in Fourier (momentum) space. To this end, we introduce, for a given $n >1$ and $1 \le j, m \le n$ with $j \not= m,$
\beq\label{eq:K+df}
K^+_{j, m} u^{(n)} (\vec{q}_n) = \big ( \mathrm{i} (q_j + q_m)^i u_{i_1,\ldots, i_{m-1}, i, i_{m+1}, \ldots, i_n} (\vec{q}_n) \big )
\eeq
and
\beq\label{eq:K-df}
K^-_{j, m} u^{(n)} (\vec{q}_n) = \Big ( - \mathrm{i} \frac{(q_j + q_m )^{i_j} (q_j + q_m)^{\el} (q_j + q_m )^i}{(q_j + q_m)^2} u_{i_1,\ldots,i_{j-1}, \el, i_{j+1}, \ldots, i_{m-1}, i, i_{m+1}, \ldots, i_n } (\vec{q}_n) \Big ).
\eeq
Recall that $a^i b_i = \sum^3_{i=1} a^i b_i$ and the convention notation $q^2 = |q|^2.$ Moreover, put
\beq\label{eq:Kdf}
K_{j,m} u^{(n)} (\vec{q}_n): = K^+_{j, m} u^{(n)} (\vec{q}_n) + K^-_{j, m} u^{(n)} (\vec{q}_n).
\eeq
Then, $K^\pm_{j,m}$ and $K_{j,m}$ are all linear operators from $\mathcal{S}_{(n)} (\mathbb{R}^3)$ into $\mathcal{S}_{(n-1)} (\mathbb{R}^3)$ for any $n >1.$ Clearly, if $\{j_1, m_1\} \cap \{j_2, m_2 \} = \emptyset$ then $K^\pm_{j_1,m_1}$ commutes with $K^\pm_{j_2,m_2},$ i.e.,
\beq\label{eq:Kcommute}
K^\pm_{j_1,m_1} K^\pm_{j_2,m_2} = K^\pm_{j_2,m_2} K^\pm_{j_1,m_1},\quad K_{j_1,m_1} K_{j_2,m_2} = K_{j_2,m_2} K_{j_1,m_1}.
\eeq

Now, in Fourier space, by \eqref{eq:W_j+} we have
\beq\label{eq:W+moment}
\begin{split}
W^+_{j, k+1} u^{(k+1)} (\vec{p}_k )= & - \mathrm{i} \int d q_{k+1} \big ( p^i_j u_{i_1,\ldots, i_k, i} (p_1, \ldots, p_{j-1}, p_j - q_{k+1}, p_{j+1}, \ldots, p_k, q_{k+1}) \big )\\
= & - \int d \vec{q}_{k+1} \bigg [ \prod^k_{\iota \neq j} \delta (p_\iota - q_\iota) \bigg ] \delta (p_j - q_j - q_{k+1}) K^+_{j, k+1} u^{(k+1)} (\vec{q}_{k+1}).
\end{split}\eeq
Similarly, by \eqref{eq:W_j-} we have
\beq\label{eq:W-moment}
\begin{split}
W^-_{j,k+1} u^{(k+1)} (\vec{p}_k )
= & \mathrm{i} \int d q_{k+1} \Big ( \frac{p^{i_j}_j p^{\el}_j p^i_j}{p_j^2} u_{i_1,\ldots,i_{j-1}, \el, i_{j+1}, \ldots, i_k, i} (p_1, \ldots, p_j - q_{k+1}, \ldots, p_k, q_{k+1}) \Big )\\
= & - \int d \vec{q}_{k+1} \bigg [ \prod^k_{\iota \neq j} \delta (p_\iota - q_\iota) \bigg ] \delta (p_j - q_j - q_{k+1}) K^-_{j, k+1} u^{(k+1)} (\vec{q}_{k+1}).
\end{split}\eeq
Therefore, by \eqref{eq:W_j} and \eqref{eq:Woperatordf}, $W^{(k)}$ acts in momentum space according to
\beq\label{eq:W-MomentumSpace}
\begin{split}
W^{(k)}  u^{(k+1)} (\vec{p}_k) = - \sum^k_{j=1} \int d \vec{q}_{k+1} \bigg [ \prod^k_{\iota \neq j} \delta (p_\iota - q_\iota) \bigg ] \delta (p_j - q_j - q_{k+1}) K_{j, k+1} u^{(k+1)} (\vec{q}_{k+1}).
\end{split}
\eeq
We will apply \eqref{eq:W-MomentumSpace} repeatedly and show that all integrals in \eqref{eq:DuhamelExpand} are absolutely convergent.

\section{Graphic expression for interaction operators}\label{GraphCollisionoperator}

In this section, we will use binary trees (see Section \ref{Binarytree} for details) to represent various terms in the Duhamel expansion \eqref{eq:DuhamelExpand}. Precisely, we will use a forest consisting of finite binary trees to indicate how the initial state evolves as the system undergoes a specific sequence of collision. First of all, we present the collision mapping from an initial state into the final state determined by such a forest.

\begin{defi}\label{df:Gamma}\rm
Given a fixed forest $\mathbb{T} \in \mathfrak{T}_{n, k},$  for every vertex $v \in V (\mathbb{T})$ we denote by $e^a_v$ the mother-edge of the vertex $v,$ by $e^b_v$ the marked daughter-edge, and by $e^c_v$ the unmarked daughter-edge (see Fig.\ref{fig:one-vertex}). Moreover, for each edge $e \in E(\mathbb{T})$ we associate a negative number $\gamma^e<0$ such that
\beq\label{eq:Gammanumbercondition}
\gamma^{e^a_v} = \gamma^{e^b_v} + \gamma^{e^c_v},
\eeq
for any vertex $v \in V(\mathbb{T}).$
\end{defi}

\begin{figure}[[htbp]
\begin{minipage}[t]{12cm}
\vspace{0pt}
 \centering
 \includegraphics[height=6cm,width=10cm]{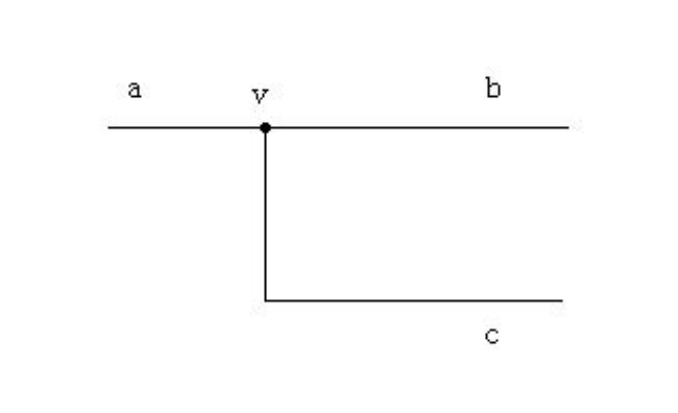}
\caption{\label{fig:one-vertex}\; One vertex with three edges}
\end{minipage}
\end{figure}

\begin{rk}\label{rk:Gamma}\rm
By definition, the values $\gamma^e$ associated with all leaves $e \in L(\mathbb{T})$ uniquely determine all others $\gamma^e$ for $e \in E(\mathbb{T}) \setminus L(\mathbb{T}).$
\end{rk}

\begin{defi}\label{df:granddaughter-edge}\rm
\begin{enumerate}[{\rm 1)}]

\item Given a fixed forest $\mathbb{T} \in \mathfrak{T}_{n,k},$ the {\it granddaughter-edge} $d(e)$ of an edge $e \in E(\mathbb{T})$ is defined as follows: If $e$ is a leaf, then $d(e) =e;$ otherwise, $d(e)$ is defined as the unique leaf $\bar{e}$ such that there is a route from $e$ to $\bar{e}$ on which all edges are marked daughter-edges.

\item Given a labelling $\pi_2$ on the leaves $L(\mathbb{T}),$ for a vertex $v \in V(\mathbb{T})$ we define an operator $K^{\pi_2}_v$ by
\beq\label{eq:K_vdf}\begin{split}
K^{\pi_2}_v & u^{(n+k)} (q_{e_v}, \vec{r}_{n+k}) = \big ( \mathrm{i} (q_{e^b_v} + q_{e^c_v})^i u_{i_1,\ldots, i_{m-1}, i, i_{m+1}, \ldots, i_{n+k}} (\vec{r}_{n+k}) \big )\\
& + \Big ( - \mathrm{i} \frac{(q_{e^b_v} + q_{e^c_v})^{i_j} (q_{e^b_v} + q_{e^c_v} )^{\el} (q_{e^b_v} + q_{e^c_v})^i}{(q_{e^b_v} + q_{e^c_v})^2} u_{i_1,\ldots,i_{j-1}, \el, i_{j+1}, \ldots, i_{m-1}, i, i_{m+1}, \ldots, i_{n+k}} (\vec{r}_{n+k}) \Big )
\end{split}\eeq
if $\pi_2 (d(e^b_v)) =j$ and $\pi_2 (d(e^c_v)) =m,$ where $1 \le j, m \le n+k$ with $j \not= m.$

\item For any two vertices $v, \bar{v} \in V(\mathbb{T}),$ the actions of $K^{\pi_2}_v$ and $K^{\pi_2}_{\bar{v}}$ follow the partial order $\prec,$ that is, if $v \prec \bar{v}$ then $K^{\pi_2}_{\bar{v}}$ first acts on $u^{(n+k)}$ and subsequently so does $K^{\pi_2}_v.$ In particular, $\prod_{v \in V(\mathbb{T})} K^{\pi_2}_v$ acts on $u^{(n+k)}$ according to the partial order $\prec$ on the vertices. Moreover, we write for any $\mathbb{T} \in \mathfrak{T}_{n,k},$
\beq\label{eq:K_Tdf}
K^{\pi_2}_{\mathbb{T}} u^{(n+k)} (\{q_{e_v}: v \in \mathbb{T} \}, \vec{r}_{n+k}) = \prod_{v \in V(\mathbb{T})} K^{\pi_2}_v u^{(n+k)} (q_{e_v}, \vec{r}_{n+k}).
\eeq
\end{enumerate}
\end{defi}

\begin{rk}\label{rk:Kaction}\rm
Note that if there is no order relation between $v$ and $\bar{v},$ then by definition $K^{\pi_2}_v$ commutes with $K^{\pi_2}_{\bar{v}}$ (e.g. \eqref{eq:Kcommute}). Thus, $K^{\pi_2}_{\mathbb{T}} u^{(n+k)} (\{q_{e_v}: v \in \mathbb{T} \}, \vec{r}_{n+k})$ is well defined.
\end{rk}

Next, we turn to the graphic representation for the fully expended terms in terms of forests defined above. For illustration, we first consider the simple case $n=1$ and $k=1.$ To this end, for a given $u^{(2)} \in \mathbb{L}^2_{(2)},$ put
\be
F^{(1)}_1 (t): = \int^t_0 d s\mathcal{T}^{(1)} (t-s) W^{(1)} \mathcal{T}^{(2)} (s) u^{(2)}.
\ee
In this case, $\mathfrak{T}_{n,k}$ with $n=k=1$ contains only a single element, i.e., the forest $\mathbb{T}_{1, 1}$ consists of the binary tree containing a vertex (cf. Fig.\ref{fig:one-vertex}). Note that in momentum space
\be
\mathcal{T}^{(2)} (s) u^{(2)} (\vec{p}_2) = e^{-s p^2_2} u^{(2)} (\vec{p}_2).
\ee
By \eqref{eq:W-MomentumSpace} we have
\be
\begin{split}
W^{(1)} \mathcal{T}^{(2)} ( s) u^{(2)} (p_1) = - \int d q_1 d q_2 e^{-s q^2_1 - s q^2_2} \delta (p_1 - q_1 - q_2) K_{1,2} u^{(2)} ( q_1, q_2).
\end{split}
\ee
Then we have
\be
\begin{split}
F^{(1)}_1 (t) (p_1) = & - \int d q_a d q_b d q_c \int^t_0 d s e^{-(t-s)q^2_a-s q^2_b - s q^2_c} \\
& \quad \times \delta (p_1- q_a) \delta (q_a - q_b - q_c) K_{b,c} u^{(2)} ( q_b, q_c),
\end{split}
\ee
where we have changed variables for corresponding to the variables in the binary tree $\mathbb{T}_{1,1}.$

Next, we consider the integral
\be
\mathrm{I} : = \int^t_0 d s e^{-(t-s) q_a^2 - s q_b^2 - s q_c^2}.
\ee
By Cauchy's integral formula, for any $s >0$ one has
\beq\label{eq:PropagatorIdent1}
e^{-s q^2} = - \int^{\8}_{- \8} \frac{ d \tau e^{-s (\gamma + \mathrm{i} \tau)} }{\gamma - q^2 + \mathrm{i} \tau},\quad \forall \gamma < 0.
\eeq
(Recall that $d \tau = d_{\mathrm{Leb}} \tau / 2 \pi.$) Hence we have
\be\begin{split}
\mathrm{I} = - \int^t_0 d s \int_{\mathbb{R}} \frac{d \tau^a d \tau^b d \tau^c e^{- (t-s) (\gamma^a + \mathrm{i} \tau^a) - s (\gamma^b + \mathrm{i} \tau^b) - s ( \gamma^c + \mathrm{i} \tau^c)} }{( \gamma^a - q_a^2 + \mathrm{i}\tau^a) ( \gamma^b - q_b^2 + \mathrm{i}\tau^b) ( \gamma^c - q^2_c + \mathrm{i} \tau^c)}.
\end{split}\ee
By Cauchy's theorem,
\beq\label{eq:PropagatorIdent2}
\int^{\8}_{- \8} \frac{d \tau e^{-s (\gamma + \mathrm{i} \tau)} }{\gamma - q^2 + \mathrm{i} \tau} = 0,
\eeq
if $s < 0$ and $\gamma < 0.$ Then the time integration in $I$ can be extended to $s \in \mathbb{R},$ and performing the $s$-integration, we have
\be
\mathrm{I} = - \int_{\mathbb{R}} \frac{d \tau^a d \tau^b d \tau^c e^{- t(\gamma^a + \mathrm{i} \tau^a)} \delta (\tau^a - \tau^b - \tau^c)}{( \gamma^a - q_a^2 + \mathrm{i}\tau^a) ( \gamma^b - q_b^2 + \mathrm{i}\tau^b) ( \gamma^c - q_c^2 + \mathrm{i} \tau^c)}
\ee
because $\gamma^a = \gamma^b + \gamma^c.$ This yields that
\be\begin{split}
F^{(1)}_1 (t) (p_1) = \int & \frac{d q_a d q_b d q_c d \tau^a d \tau^b d \tau^c}{( \gamma^a - q^2_a + \mathrm{i}\tau^a ) ( \gamma^b - q^2_b + \mathrm{i} \tau^b ) ( \gamma^c - q^2_c + \mathrm{i} \tau^c )} \delta (p_1 - q_a ) e^{- t(\gamma^a + \mathrm{i} \tau^a)} \\
& \quad \times \delta (\tau^a - \tau^b - \tau^c) \delta (q_a - q_b - q_c ) K_{b, c} u^{(2)} (q_b, q_c)\\
= \frac{1}{2} & \sum_{\pi_2 \in \Pi_2}\int d \vec{r}_2 \int \prod_{e \in E_2 (\mathbb{T}_{1,1})} d \tau^e d q_e K^{\pi_2}_{\mathbb{T}_{1,1}} u^{(2)} (\vec{r}_2 ) \prod_{e \in R_2 (\mathbb{T}_{1,1})} e^{- t (\gamma^e + \mathrm{i} \tau^e )} \delta (q_e - p_{\pi_1 (e)})\\
& \times \prod_{e \in L_2 (\mathbb{T}_{1,1})} \delta \big ( q_e - r_{\pi_2 (e)} \big ) \prod_{e \in E_2 (\mathbb{T}_{1,1})} \frac{1}{ \gamma^e - q^2_e + \mathrm{i} \tau^e}\\
& \times \prod_{v \in V(\mathbb{T}_{1,1})} \delta \big ( \tau^{e^a_v} - \tau^{e^b_v} - \tau^{e^c_v} \big ) \delta \big ( q_{e^a_v} - q_{e^b_v} - q_{e^c_v} \big ).\\
\end{split}\ee

This expression motivates us to give the definition of collision operators as follows.

\begin{defi}\label{def:collisionoperator}
Let $n \ge 0$ and $k \ge 1.$ For a given $\mathbb{T} \in \mathfrak{T}_{n,k},$ $t \ge 0,$ and a given family of strictly negative numbers $\Upsilon = \{\gamma^e\}_{e \in E(\mathbb{T})}$ such that $\gamma^{e^a_v} = \gamma^{e^b_v} + \gamma^{e^c_v}$ for all $v \in V(\mathbb{T}),$ we define the collision operator $C_{\mathbb{T}, t, \Upsilon}: \mathbb{L}^2_{(n+k)} \mapsto \mathbb{L}^2_{(k)}$ by
\beq\label{eq:collisionoerator}
\big ( C_{\mathbb{T}, t, \Upsilon} u^{(n+k)} \big ) (\vec{p}_k) = \int d \vec{r}_{n+k} \big ( G_{\mathbb{T}, t, \Upsilon} u^{(n+k)} \big ) (\vec{p}_k ; \vec{r}_{k+n})
\eeq
through its kernel
\beq\label{eq:collisionkernel}
\begin{split}
\big ( G_{\mathbb{T}, t, \Upsilon} u^{(n+k)} \big ) (\vec{p}_k ; \vec{r}_{k+n}) & = \frac{1}{(k+n)!} \sum_{\pi_2 \in \Pi_{k+n}} \prod_{e \in R_1 (\mathbb{T}) = L_1 (\mathbb{T})} e^{- t p^2_{\pi_1 (e)}} \delta (p_{\pi_1 (e)} - r_{\pi_2 (e)} )\\
\times \int \prod_{e \in E_2 (\mathbb{T})} d q_e d \tau^e K_{\mathbb{T}}^{\pi_2} & u^{(n+k)} (\vec{r}_{n+k} ) \prod_{e \in R_2 (\mathbb{T})} e^{- t (\gamma^e + \mathrm{i} \tau^e)} \delta (q_e - p_{\pi_1 (e)}) \prod_{e \in L_2 (\mathbb{T})} \delta (q_e - r_{\pi_2 (e)})\\
& \times \prod_{e \in E_2 (\mathbb{T})} \frac{1}{\gamma^e - q_e^2 + \mathrm{i} \tau^e}
\prod_{v \in V(\mathbb{T})} \delta \big ( \tau^{e^a_v} - \tau^{e^b_v} - \tau^{e^c_v} \big ) \delta \big (q_{e^a_v} - q_{e^b_v} - q_{e^c_v}\big ).
\end{split}
\eeq
\end{defi}

The collision operator $C_{\mathbb{T}, t, \Upsilon}$ will describe the terms of the summation in \eqref{eq:DuhamelExpand}.

\begin{rk}\label{rk:collisionoperator}\rm
Noticing that there are $|R_2| + |L_2| + |V| = 2k + 2 n - 2 |R_1|$ momentum delta-functions involving $q_e$-variables and $|R_1|$ delta-functions related to the roots in $R_1 (\mathbb{T}),$ but only $|E_2| = k + 2n - |R_1|$ momentum integration variables, we find that $G_{\mathbb{T}, t, \Upsilon}$ contains $k$ delta-functions among its $n + 2k$ variables, each of which corresponds to the momentum conservation in the corresponding one of the $k$ components of $\mathbb{T}.$ Also, we see that all the $q_e$ momenta are uniquely determined by the external momenta $\vec{p}_k$ and $\vec{r}_{n+k}.$ Hence, all the $d q_e$ integrations are well defined and correspond to substituting the appropriate linear combinations of the external momenta into $q_e.$
\end{rk}

\begin{prop}\label{prop:collisionoperator}
For every $\mathbb{T} \in \mathfrak{T}_{n,k}$ and a given family of negative numbers $\Upsilon = \{\gamma^e\}_{e \in E(\mathbb{T})},$ the collision kernel $G_{\mathbb{T}, t, \Upsilon}$ is well defined for all $t \ge 0.$ More precisely, all the $d \tau^e$ integrals in \eqref{eq:collisionkernel} are absolutely convergent.
\end{prop}

\begin{proof}
As remarked above, all the $q_e$-integrations of $G_{\mathbb{T}, t, \Upsilon}$ are well defined. It remains to prove the absolute convergence of all the $\tau^e$ integrals in \eqref{eq:collisionkernel}. Since the delta functions relate $\tau$-variables within the same trees, the integration can be done independently in each tree of $\mathbb{T}.$ Therefore, it suffices to consider the case of a binary tree $\mathrm{T} \in \mathfrak{T}_n.$ The proof is based on induction over $n.$

For $n=0$ there is no such integration. For $n=1,$ the $\tau^e$-integrations are of the form
\be
I := \int_{\mathbb{R}} \frac{d \tau^{e^a} d \tau^{e^b} d \tau^{e^c} \delta ( \tau^{e^a}- \tau^{e^b} - \tau^{e^c})}{(\gamma^{e^a} - q^2_{e^a} + \mathrm{i} \tau^{e^a}) (\gamma^{e^b} - q^2_{e^b} + \mathrm{i} \tau^{e^b}) (\gamma^{e^c} - q^2_{e^c} + \mathrm{i} \tau^{e^c})}
\ee
where $e^b$ and $e^c$ correspond respectively to the marked and unmarked daughter-edges of the mother-edge $e^a$ in $\mathbb{T}$ (cf. Fig.\ref{fig:one-vertex}). Recall that $\gamma^{e^a} = \gamma^{e^b} + \gamma^{e^c},$ and all the $\gamma$'s are strictly negative. Let $\gamma = \max \{ \gamma^{e^b}, \gamma^{e^c}\}.$ Note that
\be\begin{split}
I \le & \int_{\mathbb{R}} \frac{d \tau^{e^a} d \tau^{e^b}}{ | \gamma - q^2_{e^a} + \mathrm{i} \tau^{e^a} | | \gamma - q^2_{e^b} + \mathrm{i} \tau^{e^b} | | \gamma - q^2_{e^c} + \mathrm{i} (\tau^{e^b} - \tau^{e^a} ) |}\\
\le & \frac{1}{|\gamma|^3} \int_{\mathbb{R}} \frac{d \tau^{e^a} d \tau^{e^b}}{ \big | \mathrm{i} + \frac{\tau^{e^a}}{|\gamma|} \big | \big | \mathrm{i} + \frac{\tau^{e^b}}{|\gamma |} \big | \big | \mathrm{i} + \frac{1}{|\gamma|} (\tau^{e^b} - \tau^{e^a} ) \big |}\\
= & \frac{1}{|\gamma|} \int_{\mathbb{R}} \frac{d t d s}{ \langle t \rangle \langle  s  \rangle \langle s - t \rangle},
\end{split}\ee
then, by Lemma \ref{lem:Frequencyestimate} (twice) we have that the integration in the last line is finite and so
\beq\label{eq:tauintest}
I \lesssim \frac{1}{|\gamma|}.
\eeq

For general $n,$ we note that any binary tree with $n+1$ vertices can be built up from a binary tree $\mathrm{T}$ with $n$ vertexes by adjoining a new vertex to a leave. Indeed, we choose a maximal vertex $v$ of $\mathrm{T},$ i.e., there is no $v'$ such that $v \prec v'.$ We add a new vertex $v'$ by splitting one of leaves $e^b_v, e^c_v,$ denoted by $e_v$ which is $e^a_{v'},$ into two daughter-edges $e^b_{v'}$ and $e^c_{v'}$ of $v'.$ Then we create two new denominators, two new $\tau$-variables and one new delta function. The additional integration is
\be
\int_{\mathbb{R}} \frac{d \tau^{e^b_{v'}} d \tau^{e^c_{v'}} \delta ( \tau^{e^a_{v'}}- \tau^{e^b_{v'}} - \tau^{e^c_{v'}})}{(\gamma^{e^b_{v'}} - q^2_{e^b_{v'}} + \mathrm{i} \tau^{e^b_{v'}}) (\gamma^{e^c_{v'}} - q^2_{e^c_{v'}} + \mathrm{i} \tau^{e^c_{v'}})}
\ee
where $\gamma$'s are chosen such that $\gamma^{e_v} = \gamma^{e^a_{v'}} = \gamma^{e^b_{v'}} + \gamma^{e^c_{v'}}.$ As done above, by Lemma \ref{lem:Frequencyestimate} this integral is absolutely convergent uniformly in $\tau^{e_v}$ and for any choice of $\gamma^{e^b_{v'}}, \gamma^{e^c_{v'}}<0.$ After this integral is done, the tree has only $n$ vertices. Therefore, by induction, all the $d \tau^e$ integrals in \eqref{eq:collisionkernel} are absolutely convergent.
\end{proof}

\begin{prop}\label{prop:collisionoperatorindend}
For any given $t \ge 0$ and every $\mathbb{T} \in \mathfrak{T}_{n,k},$ the collision kernel $G_{\mathbb{T}, t, \Upsilon}$ is independent of the family of negative numbers $\Upsilon = \{\gamma^e\}_{e \in E(\mathbb{T})}.$ In particular, $C_{\mathbb{T}, t, \Upsilon}$ is independent of $\Upsilon,$ and $C_{\mathbb{T}, t, \Upsilon} =0$ when $t=0.$
\end{prop}

\begin{proof}
Since $\gamma^{e^a_v} = \gamma^{e^b_v} + \gamma^{e^c_v}$ for each vertex $v$ of $\mathbb{T},$ the only independent $\gamma$'s are the ones associated with the leaves of $\mathbb{T}.$ Given a fixed $\bar{e} \in L (\mathbb{T}),$ by using the estimate \eqref{eq:tauintest} in the proof of Proposition \ref{prop:collisionoperator}, $G_{\mathbb{T}, t, \Upsilon}$ has an analytic extension in the left-half plane $\{ z \in \mathbb{C}: \; \mathrm{Re} z < 0 \}$ as a function of $\gamma^{\bar{e}}.$ It suffices to show that $G_{\mathbb{T}, t, \Upsilon}$ is constant in a small neighborhood of a given value $\gamma^{\bar{e}}$ with $\mathrm{Re} \gamma^{\bar{e}} <0,$ while all the other $\gamma^e <0$ for $e \in L(\mathbb{T}) \backslash \{ \bar{e} \}$ are kept constant.

Indeed, if $\gamma^{\bar{e}}$ is replaced by $\gamma^{\bar{e}} + \beta,$ then for every $e \in E(\mathbb{T})$ on the route from $\bar{e}$ to the unique root connected to $\bar{e},$ we put $\bar{\tau}^e = \tau^e + \mathrm{i} \beta$ and keep $\bar{\tau}^e = \tau^e$ for all other $e \in E(\mathbb{T}).$ Here, we require that $|\mathrm{Re} \beta| < \min_{e \in E (\mathbb{T})} |\mathrm{Re} \gamma^e|$ in order to avoid deforming the $\tau^e$ integral contour through the pole at $\tau^e = \mathrm{i} (\gamma^e - q^2_e).$ Since
\be
\tau^{e^a_v} - \tau^{e^b_v} - \tau^{e^c_v} = \bar{\tau}^{e^a_v} - \bar{\tau}^{e^b_v} - \bar{\tau}^{e^c_v}
\ee
for all $v \in V(\mathbb{T}),$ it follows that the integral remains unchanged after the change of variable. This proves the independence of \eqref{eq:collisionkernel} from the family $\Upsilon = \{\gamma^e:\; e \in E(\mathbb{T}) \}.$

Finally, when $t =0,$  by \eqref{eq:tauintest} again, we take the limit $|\gamma| = \min_{e \in L (\mathbb{T})} |\gamma^e| \to \8$ and obtain that $C_{\mathbb{T}, 0, \Upsilon} =0.$
\end{proof}

\begin{rk}\label{rk:collisionoperatorGammafree}\rm
We will write $G_{\mathbb{T}, t, \Upsilon} = G_{\mathbb{T}, t}$ and, respectively, $C_{\mathbb{T}, t, \Upsilon} = C_{\mathbb{T}, t}$ for what follows.
\end{rk}

To describe the error term in \eqref{eq:DuhamelExpand} that involves the function $u^{(k+n)}$ at an intermediate time $t_n,$ we need to introduce a slight modification of the collision operator $C_{\mathbb{T}, t, \Upsilon}.$ For a given $\mathbb{T} \in \mathfrak{T}_{n,k},$ let $M(\mathbb{T})$ denote the set of maximal vertices of $\mathbb{T},$ i.e., $\bar{v} \in M (\mathbb{T})$ if and only if $\bar{v} \in V(\mathbb{T})$ and there is no $v \in V(\mathbb{T})$ such that $\bar{v} \prec v.$ Let $D_{\bar{v}} = \{e^b_{\bar{v}}, e^c_{\bar{v}}\}$ denote the set of daughter-edges of $\bar{v} \in M(\mathbb{T}).$

\begin{defi}\label{def:erroroperator}
Given $n,k \ge 1$ and $t \ge 0,$ for a fixed $\mathbb{T} \in \mathfrak{T}_{n,k}$ we define the error operator $R_{\mathbb{T}, t}$ by
\beq\label{eq:Errortermoperator}
\big ( R_{\mathbb{T}, t} u^{(k+n)} \big ) (\vec{p}_k) = \int d \vec{r}_{k+n} \big ( Q_{\mathbb{T}, t} u^{(k+n)} \big ) (\vec{p}_k, \vec{r}_{k+n})
\eeq
through its kernel
\beq\label{eq:erroroperatorkernel}\begin{split}
\big ( Q_{\mathbb{T}, t} u^{(k+n)} \big ) (\vec{p}_k, & \vec{r}_{k+n})
: = \frac{1}{(k+n)!} \sum_{\pi_2 \in \Pi_{k+n}} \sum_{\bar{v} \in M(\mathbb{T})} \prod_{e \in R_1 (\mathbb{T}) = L_1 (\mathbb{T})} e^{- t p^2_{\pi_1 (e)}} \delta (p_{\pi_1 (e)} - r_{\pi_2 (e)} )\\
\times \int \prod_{e \in E_2 (\mathbb{T})} & d q_e \prod_{ \substack{ e \in E_2 (\mathbb{T})\\ e \notin D_{\bar{v}} } } \frac{d \tau^e}{\gamma^e - q_e^2 + \mathrm{i} \tau^e} K_{\mathbb{T}}^{\pi_2} u^{(n+k)} (\vec{r}_{n+k} ) \prod_{ e \in R_2 (\mathbb{T})} e^{- t(\gamma^e + \mathrm{i} \tau^e)} \delta (q_e - p_{\pi_1 (e)})\\
& \times \prod_{e \in L_2 (\mathbb{T})} \delta (q_e - r_{\pi_2 (e)}) \prod_{ \substack{ v \in V(\mathbb{T})\\ v \not= \bar{v} } } \delta \big ( \tau^{e^a_v} - \tau^{e^b_v} - \tau^{e^c_v} \big ) \prod_{ v \in V(\mathbb{T}) } \delta \big (q_{e^a_v} - q_{e^b_v} - q_{e^c_v} \big )
\end{split}\eeq
where the family of strictly negative numbers $\Upsilon = \{\gamma^e:\; e \in E(\mathbb{T}) \}$ is chosen as in \eqref{eq:Gammanumbercondition}.
\end{defi}

Namely, $R_{\mathbb{T}, t}$ is defined so that there are no propagators associated with the daughter-edges of each $\bar{v} \in M (\mathbb{T}).$ Note that although $\gamma^e$ does not appear in \eqref{eq:erroroperatorkernel} for $e \in D_{\bar{v}},$ the value of the $\gamma$ associated with the mother-edge of $\bar{v}$ depends on them.

\begin{rk}\label{rk:erroroperatorGammafree}\rm
\begin{enumerate}[{\rm 1)}]

\item As similar to the collision kernel $G_{\mathbb{T}, t},$ $Q_{\mathbb{T}, t}$ is well defined, that is, all the $d \tau$-integrals in \eqref{eq:erroroperatorkernel} are absolutely convergent. This can been done as in the proof of Proposition \ref{prop:collisionoperator}.

\item It can be proved as in Proposition \ref{prop:collisionoperatorindend} that the error kernel $Q_{\mathbb{T}, t}$ is independent of the choice of $\Upsilon = \{\gamma^e:\; e \in E(\mathbb{T}) \},$ and so does the error operator $R_{\mathbb{T}, t}.$ This explain the reason that we do not include the notation $\Upsilon$ in the definitions of $Q_{\mathbb{T}, t}$ and $R_{\mathbb{T}, t}.$
\end{enumerate}\end{rk}

\section{Graphic representation for the Navier-Stokes hierarchy}\label{GraphicNShierachy}

In this section, we will show that
\beq\label{eq:NSEhierarchyTreeExpanssion}
\begin{split}
u^{(k)} (t) = e^{t \triangle^{(k)}} u^{(k)}_0 + \sum^n_{j=1} \sum_{\mathbb{T} \in \mathfrak{T}_{j,k}} C_{\mathbb{T}, t} u^{(k+j)}_0 - \sum_{\mathbb{T} \in \mathfrak{T}_{n+1,k}} \int^t_0 d s R_{\mathbb{T}, t-s} u^{(k+n+1)} (s)
\end{split}
\eeq
for any $k \ge 1$ and for every $n \ge 1.$ This will be done respectively for the fully expended terms and remainder terms in the Duhamel expansion \eqref{eq:DuhamelExpand}.

First of all, we note that for any $k \ge 1,$ by the definition of the collision operator \eqref{eq:collisionoerator},
\beq\label{eq:Duhamelfulltermtreen=0}
\mathcal{T}^{(k)}_0 (t) u^{(k)} = \sum_{\mathbb{T} \in \mathfrak{T}_{0, k}} C_{\mathbb{T}, t} u^{(k)}
\eeq
for every $u^{(k)} \in \mathfrak{L}^2_{(k)},$ where the summation on the right hand side is only for the (unique) forest consisting of $k$ trivial trees. The first aim of this section is to extend the expression \eqref{eq:Duhamelfulltermtreen=0} to the fully expended terms in \eqref{eq:DuhamelExpand}.

\begin{thm}\label{th:Duhamelfulltree}
Let $k \ge 1$ and $n \ge 1.$ For any given $u^{(k+n)} \in \mathbb{L}^2_{(k+n)},$ we have
\beq\label{eq:Duhamelfulltermtree}\begin{split}
\sum_{\mathbb{T} \in \mathfrak{T}_{n,k}} C_{\mathbb{T}, t} u^{(k+n)}= \int^t_0 d t_1 & \int^{t_1}_0 d t_2 \cdots  \int^{t_{n-1}}_0 d t_n \mathcal{T}^{(k)} ( t-t_1) W^{(k)}\cdots \\
& \times \mathcal{T}^{(k+n-1)} ( t_{n-1}-t_n) W^{(k+n-1)} \mathcal{T}^{(k+n)} ( t_n) u^{(k+n)}
\end{split}\eeq
for all $t \ge 0.$
\end{thm}

\begin{proof} Assume $k \ge 1$ and $n \ge 1.$ Fix $T>0.$ For a given $u^{(k+n)} \in \mathfrak{L}^2_{(k+n)},$ put
\be\begin{split}
F^{(k)}_{n, t}: = \int^t_0 & d t_1 \int^{t_1}_0 d t_2 \cdots  \int^{t_{n-1}}_0 d t_n \mathcal{T}^{(k)} ( t-t_1) W^{(k)}\cdots \\
& \quad \times \mathcal{T}^{(k+n-1)} ( t_{n-1}-t_n) W^{(k+n-1)} \mathcal{T}^{(k+n)} ( t_n) u^{(k+n)}
\end{split}\ee
for $t \in [0, T].$ Clearly, $F^{(k)}_{n, 0} =0.$ Also, by Proposition \ref{prop:collisionoperatorindend} one has
\be
\sum_{\mathbb{T} \in \mathfrak{T}_{n,k}} C_{\mathbb{T}, 0} u^{(k+n)} =0.
\ee
Thus, \eqref{eq:Duhamelfulltermtree} holds true at $t=0.$

For $t >0,$ we compute the derivative of $F^{(k)}_{n, t}$ with respect to $t$ as follows
\be\begin{split}
\partial_t F^{(k)}_{n, t} =  \int^t_0 d t_2 \cdots  \int^{t_{n-1}}_0 d t_n  W^{(k)}\mathcal{T}^{(k+1)} ( t-t_2) & \cdots \mathcal{T}^{(k+n-1)} ( t_{n-1}-t_n) W^{(k+n-1)} \mathcal{T}^{(k+n)} ( t_n) u^{(k+n)}\\
+ \sum^k_{j=1} \triangle_j \int^t_0 d t_1 & \int^{t_1}_0 d t_2 \cdots  \int^{t_{n-1}}_0 d t_n \mathcal{T}^{(k)} ( t-t_1) W^{(k)}\cdots \\
& \times \mathcal{T}^{(k+n-1)} ( t_{n-1}-t_n) W^{(k+n-1)} \mathcal{T}^{(k+n)} ( t_n) u^{(k+n)}.
\end{split}\ee
Then in momentum space, one has
\beq\label{eq:Fullytermequa}
\partial_t F^{(k)}_{n, t} (\vec{p}_k) = - \vec{p}^2_k  F^{(k)}_{n, t} (\vec{p}_k) + W^{(k)} F^{(k+1)}_{n-1, t} (\vec{p}_k)
\eeq
for all $t \in [0, T].$ On the other hand, let
\be
\Xi^{(k)}_{n, t}: = \sum_{\mathbb{T} \in \mathfrak{T}_{n,k}} C_{\mathbb{T}, t} u^{(k+n)}.
\ee
In the sequel, we will show that $\Xi^{(k)}_{n, t}$ also satisfies \eqref{eq:Fullytermequa}, i.e.,
\beq\label{eq:Fullytermequatree}
\partial_t \Xi^{(k)}_{n, t} (\vec{p}_k) = - \vec{p}^2_k  \Xi^{(k)}_{n, t} (\vec{p}_k) + W^{(k)} \Xi^{(k+1)}_{n-1, t} (\vec{p}_k)
\eeq
for all $t \in [0, T].$ By induction over $n,$ this shows that $\Xi^{(k)}_{n, t} = F^{(k)}_{n, t}$ for all $k \ge 1, n \ge 1,$ and for all $t \in [ 0, T],$ as required.

First of all, by \eqref{eq:collisionkernel} the derivative of $\Xi^{(k)}_{n,t}$ with respect to $t$ can be computed as follows: (Note that the integral with respect to $\tau_e$ is not absolutely convergence after the differentiation, so the following calculation is formal, but we will indicate how to make this rigorous later.)
\beq\label{eq:FullytermequatreeB}\begin{split}
\partial_t \Xi^{(k)}_{n, t}(\vec{p}_k) & = - \sum_{\mathbb{T} \in \mathfrak{T}_{n,k}} \int d \vec{r}_{n+k} \frac{1}{(k+n)!} \sum_{\pi_2 \in \Pi_{k+n}} \prod_{e \in R_1 (\mathbb{T}) = L_1 (\mathbb{T})} e^{- t p^2_{\pi_1 (e)}} \delta (p_{\pi_1 (e)} - r_{\pi_2 (e)} )\\
\times \int & \prod_{e \in E_2 (\mathbb{T})} d q_e d \tau^e K_{\mathbb{T}}^{\pi_2} u^{(n+k)} (\vec{r}_{n+k} ) \prod_{e \in R_2 (\mathbb{T})} e^{- t (\gamma^e + \mathrm{i} \tau^e)} \delta (q_e - p_{\pi_1 (e)})\prod_{e \in L_2 (\mathbb{T})} \delta (q_e - r_{\pi_2 (e)}) \\
& \times \bigg [ \sum_{e \in R_1 (\mathbb{T})}p^2_{\pi_1 (e)} + \sum_{e \in R_2 (\mathbb{T})}(\gamma^e + \mathrm{i} \tau^e) \bigg ]\\
& \times \prod_{e \in E_2 (\mathbb{T})} \frac{1}{\gamma^e - q_e^2 + \mathrm{i} \tau^e}\prod_{v \in V(\mathbb{T})} \delta \big ( \tau^{e^a_v} - \tau^{e^b_v} - \tau^{e^c_v} \big ) \delta \big (q_{e^a_v} - q_{e^b_v} - q_{e^c_v}\big ).
\end{split}
\eeq
Note that
\be
\sum_{e \in R_2 (\mathbb{T})}(\gamma^e + \mathrm{i} \tau^e) = \sum_{e \in R_2 (\mathbb{T})}(\gamma^e - q^2_e + \mathrm{i} \tau^e) + \sum_{e \in R_2 (\mathbb{T})} q^2_e.
\ee
Since the delta-functions $\prod_{e \in R_2 (\mathbb{T})} \delta (q_e - p_{\pi_1 (e)})$ are involved in the integration, we have
\be
\sum_{e \in R_1 (\mathbb{T})}p^2_{\pi_1 (e)} + \sum_{e \in R_2 (\mathbb{T})} q^2_e = \sum_{e \in R(\mathbb{T})} p^2_{\pi_1 (e)} = \vec{p}^2_k.
\ee
Combing these two equations yields
\be
\partial_t \Xi^{(k)}_{n, t} (\vec{p}_k) = - \vec{p}^2_k \Xi^{(k)}_{n, t} (\vec{p}_k) + B (\vec{p}_k)
\ee
where
\beq\label{eq:FullytermequatreeB1}\begin{split}
B(\vec{p}_k) & = - \sum_{\mathbb{T} \in \mathfrak{T}_{n,k}} \int d \vec{r}_{n+k} \frac{1}{(k+n)!} \sum_{\pi_2 \in \Pi_{k+n}} \prod_{e \in R_1 (\mathbb{T}) = L_1 (\mathbb{T})} e^{- t p^2_{\pi_1 (e)}} \delta (p_{\pi_1 (e)} - r_{\pi_2 (e)} )\\
\times & \int \prod_{e \in E_2 (\mathbb{T})} d q_e d \tau^e K_{\mathbb{T}}^{\pi_2} u^{(n+k)} (\vec{r}_{n+k} ) \prod_{e \in R_2 (\mathbb{T})} e^{- t (\gamma^e + \mathrm{i} \tau^e)} \delta (q_e - p_{\pi_1 (e)}) \prod_{e \in L_2 (\mathbb{T})} \delta (q_e - r_{\pi_2 (e)})\\
\times & \sum_{e \in R_2 (\mathbb{T})}(\gamma^e - q^2_e + \mathrm{i} \tau^e) \prod_{e \in E_2 (\mathbb{T})} \frac{1}{\gamma^e - q_e^2 + \mathrm{i} \tau^e} \prod_{v \in V(\mathbb{T})} \delta \big ( \tau^{e^a_v} - \tau^{e^b_v} - \tau^{e^c_v} \big ) \delta \big (q_{e^a_v} - q_{e^b_v} - q_{e^c_v}\big ).
\end{split}
\eeq
In order to prove \eqref{eq:Fullytermequatree}, it then suffices to show that $B  (\vec{p}_k) = W^{(k)} \Xi^{(k+1)}_{n-1, t} (\vec{p}_k).$

To this end, for a given $\bar{e} \in R_2 (\mathbb{T}),$ let $\bar{v} = v (\bar{e}) \in V (\mathbb{T})$ be the only vertex such that $\bar{e} \in \bar{v}$ (there is such vertex by the definition of $R_2 (\mathbb{T})$). Then, \eqref{eq:FullytermequatreeB1} can be rewritten as (note that $L_2(\mathbb{T}) \cap R_2 (\mathbb{T}) = \emptyset$):
\beq\label{eq:FullytermequatreeB2}\begin{split}
B & (\vec{p}_k) = - \sum_{\mathbb{T} \in \mathfrak{T}_{n,k}} \sum_{\bar{e} \in R_2 (\mathbb{T})} \int d \vec{r}_{n+k} \frac{1}{(k+n)!} \sum_{\pi_2 \in \Pi_{k+n}} \prod_{e \in R_1 (\mathbb{T}) = L_1 (\mathbb{T})} e^{- t p^2_{\pi_1 (e)}} \delta (p_{\pi_1 (e)} - r_{\pi_2 (e)} )\\
& \times \int \prod_{e \in \bar{v}} d q_e d \tau^e e^{- t (\gamma^{\bar{e}} + \mathrm{i} \tau^{\bar{e}})} \delta (q_{\bar{e}} - p_{\pi_1 (\bar{e})}) \delta \big ( \tau^{e^a_{\bar{v}}} - \tau^{e^b_{\bar{v}}} - \tau^{e^c_{\bar{v}}} \big ) \delta \big (q_{e^a_{\bar{v}}} - q_{e^b_{\bar{v}}} - q_{e^c_{\bar{v}}} \big ) K^{\pi_2}_{\bar{v}}\\
& \times \int \prod_{\substack{e \in E_2 (\mathbb{T})\\ e \notin \bar{v}}} d q_e d \tau^e \prod_{ \substack{v \in V(\mathbb{T}) \\ v \not= \bar{v} }} K^{\pi_2}_v u^{(n+k)} (\vec{r}_{n+k} ) \prod_{ \substack{ e \in R_2 (\mathbb{T})\\ e \not= \bar{e} } } e^{- t (\gamma^e + \mathrm{i} \tau^e)} \delta (q_e - p_{\pi_1 (e)})\\
& \times \prod_{e \in L_2 (\mathbb{T})} \delta (q_e - r_{\pi_2 (e)}) \prod_{\substack{ e \in E_2 (\mathbb{T})\\ e \not= \bar{e}}} \frac{1}{\gamma^e - q_e^2 + \mathrm{i} \tau^e} \prod_{ \substack{ v \in V(\mathbb{T})\\ v \not= \bar{v} }} \delta \big ( \tau^{e^a_v} - \tau^{e^b_v} - \tau^{e^c_v} \big )  \delta \big (q_{e^a_v} - q_{e^b_v} - q_{e^c_v}\big ).
\end{split}
\eeq
Since $\gamma^{e^a_v} = \gamma^{e^b_v} + \gamma^{e^c_v},$ performing the $d \tau^{\bar{e}}$ integration, we have
\be
\int d \tau^{\bar{e}} e^{- t (\gamma^{\bar{e}} + \mathrm{i} \tau^{\bar{e}})} \delta \big ( \tau^{e^a_{\bar{v}}} - \tau^{e^b_{\bar{v}}} - \tau^{e^c_{\bar{v}}} \big ) = e^{- t \sum_{e \in \bar{v}, e \neq \bar{e}}(\gamma^e + \mathrm{i} \tau^e)}.
\ee
Combing this term with the factor $e^{- t\sum_{e \in R_2 (\mathbb{T}) \backslash \{ \bar{e} \}} (\gamma^e + \mathrm{i} \tau^e)},$ we obtain a factor
\be
e^{- t\sum_{e \in R} (\gamma^e + \mathrm{i} \tau^e)}
\ee
with $R = \{ e^b_{\bar{v}}, e^c_{\bar{v}} \} \cup R_2 (\mathbb{T}) \backslash \{ \bar{e} \}.$

For this $\mathbb{T}$ and $\bar{e},$ we construct a new forest $\tilde{\mathbb{T}} = \tilde{\mathbb{T}} (\mathbb{T}, \bar{e}) \in \mathfrak{T}_{n-1, k+1}$ with $R \cup R_1 (\mathbb{T})$ as being the set of $k+1$ roots of it as follows:
\begin{enumerate}

\item Remove the vertex $\bar{v}$ together with the (root) edge $\bar{e}$ (recall that $\bar{v}$ is the unique vertex to which $\bar{e}$ is adjacent).

\item Add the two daughter-edges adjacent to $\bar{v}$ to the set of roots such that the marked daughter-edge inherits the label of $\bar{e},$ while the unmarked daughter-edge becomes the $(k+1)$-th root of the new forest $\tilde{\mathbb{T}}.$

\item All the other roots keep their labels.

\end{enumerate}
The vertices and edges of $\tilde{\mathbb{T}}$ are respectively $V(\tilde{\mathbb{T}}) = V(\mathbb{T}) \backslash \{\bar{v} \}$ and $E(\tilde{\mathbb{T}}) = E (\mathbb{T}) \backslash \{ \bar{e} \},$ but the leaves of $\tilde{\mathbb{T}}$ is identical to that of $\mathbb{T},$ i.e., $L (\tilde{\mathbb{T}}) = L(\mathbb{T}).$ For illustrating this construction, see Fig.\ref{fig:foresttrans(4,3)to(3,4)} for an example of a forest $\mathbb{T} \in \mathfrak{T}_{4,3}$ together with the root of the first tree mapping to a forest $\tilde{\mathbb{T}}$ in $\mathfrak{T}_{3,4}.$
\begin{figure}[htb]
\begin{minipage}[t]{12cm}
\vspace{0pt}
 \centering
 \includegraphics[height=7cm,width=10cm]{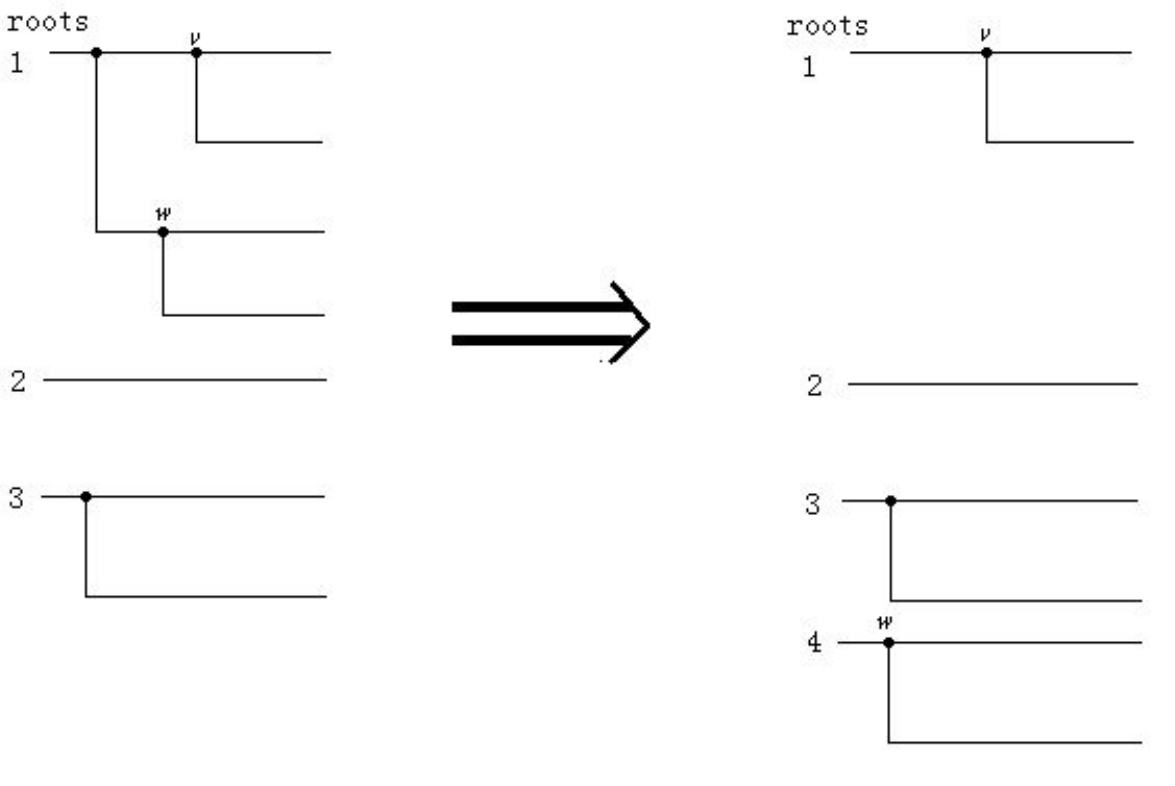}
\caption{\label{fig:foresttrans(4,3)to(3,4)} $\mathbb{T}$ together with $\bar{e}$ mapping to $\tilde{\mathbb{T}}$}
\end{minipage}
\end{figure}

Notice that the map constructed above from $\mathbb{T} \in \mathfrak{T}_{n,k}$ together with $\bar{e} \in R_2(\mathbb{T})$ to $\tilde{\mathbb{T}} \in \mathfrak{T}_{n-1, k+1}$ is surjective but not injective. This is because that for every $\tilde{\mathbb{T}} \in \mathfrak{T}_{n-1, k+1},$ its last $k+1$-th root can be attached to any of the first $k$ roots, and so there are $k$ possible choices of $\mathbb{T} \in \mathfrak{T}_{n,k}$ together with $\bar{e} \in R_2(\mathbb{T})$ mapping to $\tilde{\mathbb{T}}.$ Furthermore, this fact implies that the sum over $\mathbb{T} \in \mathfrak{T}_{n,k}$ and $\bar{e} \in R_2 (\mathbb{T})$ in \eqref{eq:FullytermequatreeB2} can be replaced by a sum over $\tilde{\mathbb{T}} \in \mathfrak{T}_{n-1, k+1}$ and a sum over the first $k$ roots of $\tilde{\mathbb{T}}.$ Therefore, using these notations we can rewrite \eqref{eq:FullytermequatreeB2} as
\be\begin{split}
B (\vec{p}_k) & = - \sum^k_{j=1} \int d \vec{\tilde{p}}_{k+1} \bigg [ \prod^k_{\iota \not= j} \delta (p_\iota - \tilde{p}_\iota) \bigg ] \delta \big (p_j - \tilde{p}_j - \tilde{p}_{k+1} \big ) K_{j, k+1}\\
\times & \sum_{\tilde{\mathbb{T}} \in \mathfrak{T}_{n-1,k+1}} \int d \vec{r}_{n+k} \frac{1}{(k+n)!} \sum_{\pi_2 \in \Pi_{k+n}} \prod_{e \in R_1 (\tilde{\mathbb{T}}) = L_1 (\tilde{\mathbb{T}})} e^{- t \tilde{p}^2_{\pi_1 (e)}} \delta (p_{\pi_1 (e)} - r_{\pi_2 (e)} )\\
& \times \int \prod_{e \in E_2 (\tilde{\mathbb{T}})} d \tau^e d q_e K_{\tilde{\mathbb{T}}}^{\pi_2} u^{(n+k)} (\vec{r}_{n+k} ) \prod_{ e \in R_2 (\tilde{\mathbb{T}})} e^{- t (\gamma^e + \mathrm{i} \tau^e)} \delta (q_e - \tilde{p}_{\pi_1 (e)}) \prod_{e \in L_2 (\tilde{\mathbb{T}})} \delta (q_e - r_{\pi_2 (e)})\\
& \times \prod_{e \in E_2 (\tilde{\mathbb{T}})} \frac{1}{\gamma^e - q_e^2 + \mathrm{i} \tau^e} \prod_{ v \in V(\tilde{\mathbb{T}}) } \delta \big ( \tau^{e^a_v} - \tau^{e^b_v} - \tau^{e^c_v} \big ) \delta \big (q_{e^a_v} - q_{e^b_v} - q_{e^c_v}\big )
=  W^{(k)} \Xi^{(k+1)}_{n-1, t} (\vec{p}_k)
\end{split}
\ee
where we have rewritten the integration over $d q_e$ with $e \in \bar{v}$ in the second line of \eqref{eq:FullytermequatreeB2}, so that it corresponds to the action of the operator $W^{(k)}$ as defined in \eqref{eq:W-MomentumSpace}. This proves \eqref{eq:Fullytermequatree}.

Finally, we turn to showing how to make the calculation in \eqref{eq:FullytermequatreeB} rigorous. To this end, we introduce a regularizing factor $\exp \{ - \epsilon \sum_{e \in E_2 (\mathbb{T})} |\tau^e | \}$ in the definition of $G_{\mathbb{T}, t}$ in \eqref{eq:collisionkernel} for any $\epsilon >0,$ denoted by $G_{\mathbb{T}, t; \epsilon}$ this new kernel, and obtain the corresponding operator $\Xi^{(k)}_{n, t; \epsilon}.$ Then the proceeding calculation for $\Xi^{(k)}_{n, t; \epsilon}$ in place of $\Xi^{(k)}_{n, t}$ can be done rigorously, and we have
\be
\Xi^{(k)}_{n, t} (\vec{p}_k) = \lim_{\epsilon \to 0} \int^t_0 d s \big [ - \vec{p}^2_k  \Xi^{(k)}_{n, s; \epsilon} (\vec{p}_k) + W^{(k)} \Xi^{(k+1)}_{n-1, s; \epsilon} (\vec{p}_k) \big ].
\ee
By Proposition \ref{prop:collisionoperator}, the integrations over $\tau$-variables in $\Xi^{(k)}_{n, t; \epsilon} $ are all absolutely convergent uniformly for every $\epsilon>0$ and for all $t \in [0,T]$ with any fixed $T>0.$ Thus, taking the limit $\epsilon \to 0$ on the right hand side of the above equation into the integral, we obtain
\be\label{eq:FullytermequatreeInt}
\Xi^{(k)}_{n, t} (\vec{p}_k) = \int^t_0 d s \big [ - \vec{p}^2_k  \Xi^{(k)}_{n, s} (\vec{p}_k) + W^{(k)} \Xi^{(k+1)}_{n-1, s} (\vec{p}_k) \big ]
\ee
which is equivalent to the equation \eqref{eq:Fullytermequatree}, since $\Xi^{(k)}_{n, 0} (\vec{p}_k) =0$ as shown above.
\end{proof}

Next, we consider the error terms in \eqref{eq:DuhamelExpand}.

\begin{thm}\label{th:Duhamelerrortree}
Let $k \ge 1$ and $n \ge 1.$ For any given $T>0,$ if $u^{(k+n)} \in L^2 ([0, T], \mathfrak{H}^1_{(k+n)})$ then
\beq\label{eq:Duhamelerrortree}
\begin{split}
- \sum_{\mathbb{T} \in \mathfrak{T}_{n,k}} \int^t_0 d s R_{\mathbb{T}, t-s} u^{(k+n)} (s) = \int^t_0 d t_1 & \int^{t_1}_0 d t_2 \cdots \int^{t_{n-1}}_0 d t_n \mathcal{T}^{(k)} ( t-t_1) W^{(k)} \cdots \\
& \times \mathcal{T}^{(k+n-1)} ( t_{n-1}-t_n ) W^{(k + n - 1)} u^{(k + n)} (t_n)
\end{split}\eeq
for all $ t \in [0, T].$
\end{thm}

\begin{proof}
Assume that $k \ge 1$ and $n \ge 1.$ Fix $T>0.$ For any $u^{(k+n)} \in L^2 ([0, T], \mathfrak{H}^1_{(k+n)})$ we put
\be\begin{split}
Q^{(k)}_n (t) u^{(k+n)} : = & \int^t_0 d t_1 \int^{t_1}_0 d t_2 \cdots \int^{t_{n-1}}_0 d t_n \mathcal{T}^{(k)} ( t-t_1) W^{(k)} \cdots \\
& \quad \times \mathcal{T}^{(k+n-1)} ( t_{n-1}-t_n ) W^{(k + n - 1)} u^{(k + n)} (t_n)
\end{split}
\ee
for all $t \in [0, T].$ By Fubini's theorem, we have
\be\begin{split}
Q^{(k)}_n (t) u^{(k+n)} = & \int^t_0 d s \bigg [ \int^{t-s}_0 d t_1 \int^{t_1}_0 d t_2 \cdots \int^{t_{n-2}}_0 d t_{n-1} \mathcal{T}^{(k)} ( t-s-t_1) W^{(k)} \cdots \\
& \quad \times \mathcal{T}^{(k+n-2)} ( t_{n-2}-t_{n-1}) W^{(k+n-2)} \mathcal{T}^{(k+n-1)} ( t_{n-1}) \bigg ] W^{(k + n - 1)} u^{(k + n)} (s).
\end{split}
\ee
By \eqref{eq:Duhamelfulltermtree}, we have
\beq\label{eq:Errortermtree}
Q^{(k)}_n (t) u^{(k+n)} = \int^t_0 d s \sum_{\tilde{\mathbb{T}} \in \mathfrak{T}_{n-1, k}} C_{\tilde{\mathbb{T}}, t-s} W^{(k + n - 1)} u^{(k + n)} (s).
\eeq
Furthermore, by \eqref{eq:collisionkernel} and \eqref{eq:W-MomentumSpace} we have, in momentum space,
\beq\label{eq:Errortermequatree}\begin{split}
Q^{(k)}_n & (t) u^{(k+n)} (\vec{p}_k)\\
= - & \int^t_0 d s \sum_{\tilde{\mathbb{T}} \in \mathfrak{T}_{n-1,k}} \int \frac{d \vec{r}_{k + n-1}}{(k+n-1)!} \sum_{\pi_2 \in \Pi_{k+n-1}} \prod_{e \in R_1 (\tilde{\mathbb{T}}) = L_1 (\tilde{\mathbb{T}})} e^{- (t-s) p^2_{\pi_1 (e)}} \delta (p_{\pi_1 (e)} - r_{\pi_2 (e)} )\\
& \times \int \prod_{e \in E_2 (\tilde{\mathbb{T}})} d q_e d \tau^e \prod_{ e \in R_2 (\tilde{\mathbb{T}})} e^{- (t-s) (\gamma^e + \mathrm{i} \tau^e)} \delta (q_e - p_{\pi_1 (e)}) \prod_{e \in L_2 (\tilde{\mathbb{T}})} \delta (q_e - r_{\pi_2 (e)})\\
& \times \prod_{e \in E_2 (\tilde{\mathbb{T}})} \frac{1}{\gamma^e - q_e^2 + \mathrm{i} \tau^e}\prod_{ v \in V(\tilde{\mathbb{T}}) } \delta \big ( \tau^{e^a_v} - \tau^{e^b_v} - \tau^{e^c_v} \big ) \delta \big (q_{e^a_v} - q_{e^b_v} - q_{e^c_v}\big )\\
& \times K^{\pi_2}_{\tilde{\mathbb{T}}} \bigg [ \sum^{k+ n -1}_{j=1} \int d \vec{\tilde{q}}_{k+1} \bigg [ \prod^{k+n-1}_{\iota \not= j} \delta (r_\iota - \tilde{q}_\iota) \bigg ] \delta \big (r_j - \tilde{q}_j - \tilde{q}_{k+n} \big ) K_{j, k+n} u^{(n+k)} (s, \vec{\tilde{q}}_{n+k}) \bigg ]
\end{split}
\eeq
where the sum over $j$ in the last line corresponds to the action of the operator $W^{(k+n-1)}$ on $u^{(n+k)} (s).$ Evidently, $j$ labels the leaves $L (\tilde{\mathbb{T}}).$ Thus, choosing $j=1, \ldots, k+n-1$ corresponds to fixing one of the leaves of $\tilde{\mathbb{T}}.$

For a fixed $\tilde{\mathbb{T}} \in \mathfrak{T}_{n-1, k}$ and $\bar{e} \in L(\tilde{\mathbb{T}}),$ we construct a new forest $\mathbb{T} \in \mathfrak{T}_{n,k}$ by splitting the edge $\bar{e}$ with a new vertex and attaching a new leaf to this vertex as its unmarked daughter-edge (e.g. Fig.\ref{fig:foresttrans(4,3)to(5,3)}).
\begin{figure}[htb]
\begin{minipage}[t]{12cm}
\vspace{0pt}
 \centering
 \includegraphics[height=7cm,width=10cm]{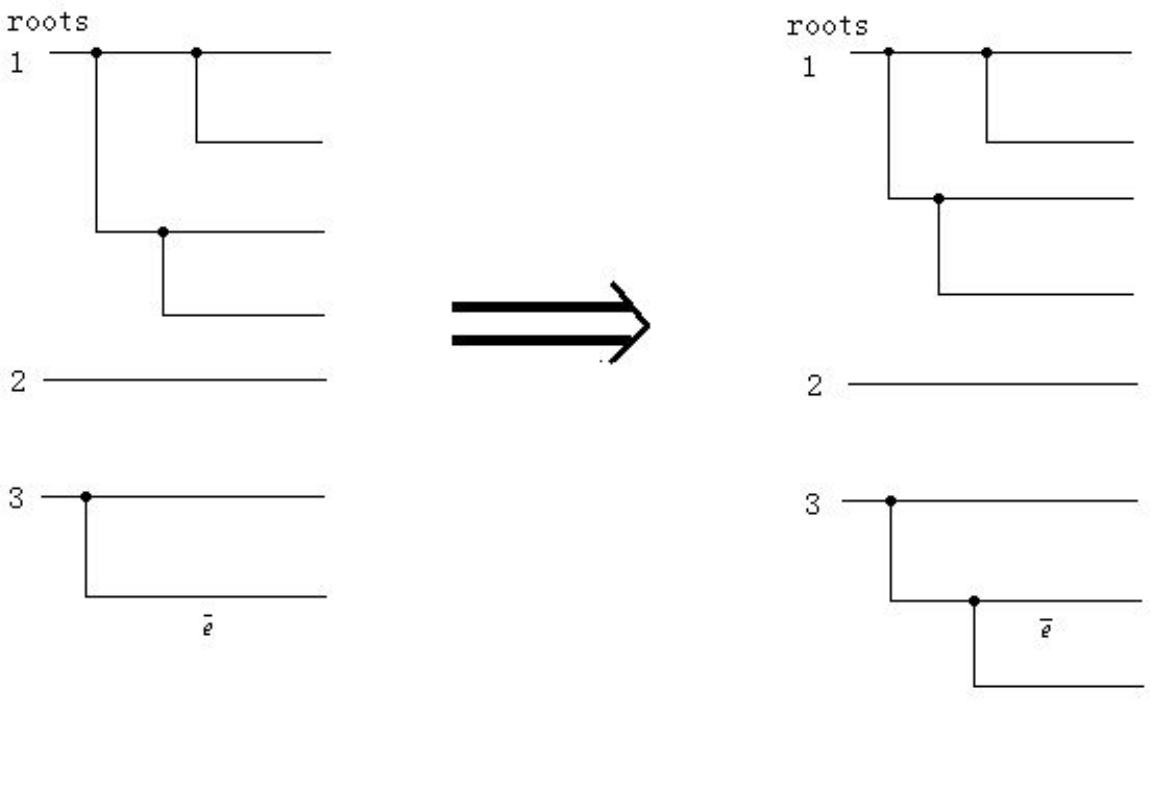}
\caption{\label{fig:foresttrans(4,3)to(5,3)} $(\tilde{\mathbb{T}}, \bar{e})$ mapping to $\mathbb{T}$}
\end{minipage}
\end{figure}
We notice that the map constructed from $\tilde{\mathbb{T}} \in \mathfrak{T}_{n-1,k}$ together with $\bar{e} \in L (\mathbb{T})$ to $\mathbb{T} \in \mathfrak{T}_{n, k}$ is surjective but not injective. Indeed, for a given $\mathbb{T} \in \mathfrak{T}_{n, k},$ by removing a $\bar{v} \in M(\mathbb{T})$ and deleting the daughter-edges of $\bar{v},$ we obtain a $\tilde{\mathbb{T}} \in \mathfrak{T}_{n-1, k}$ from which we can obtain $\mathbb{T}.$ (Recall that $M(\mathbb{T})$ denotes the set of maximal vertices of $\mathbb{T},$ i.e., $\bar{v} \in M (\mathbb{T})$ if and only if $\bar{v} \in V(\mathbb{T})$ and there is no $v \in V(\mathbb{T})$ such that $\bar{v} \prec v.$) This also yields that the sum over $\tilde{\mathbb{T}} \in \mathfrak{T}_{n-1,k}$ and $j \in \{1, 2, \ldots, k+n-1 \}$ in \eqref{eq:Errortermequatree} can be replaced by a sum over $\mathbb{T} \in \mathfrak{T}_{n, k}$ and a sum over $\bar{v} \in M(\mathbb{T}).$

Next, we want to rewrite \eqref{eq:Errortermequatree} in terms of forests $\mathbb{T} \in \mathfrak{T}_{n, k}.$ To this end, note that
\begin{enumerate}

\item From \eqref{eq:Errortermequatree} the two daughter-edges of $\bar{v}$ will have the integrations for $q$-variables as all the other edges, but they will not have any $\tau$-variable, any $\gamma$-variable, or any propagator.

\item The labelling $\pi_2$ of the leaves of $\tilde{\mathbb{T}}$ induces a labelling with $\{1, 2, \ldots, k+n-1 \}$ of the leaves of $\mathbb{T},$ except for the unmarked daughter-edge of the chosen $\bar{v} \in M (\mathbb{T})$ which is always labelled by the number $k+n.$

\item Since $u^{(n+k)} (\vec{r}_{k+n})$ is symmetric with respect to permutation on $\vec{r}_{n+k},$ we can restore a full symmetry of the leaf-variables $\vec{r}_{k+n};$ for this, we need to replace the sum over $\pi_2 \in \Pi_{k+n-1}$ by a sum over $\pi_2 \in \Pi_{k+n}$ and replace the factor $(k+n-1)!$ by $(k+n)!.$

\end{enumerate}
Thus, we can rewrite \eqref{eq:Errortermequatree} as
\be\begin{split}
Q^{(k)}_n & (t) u^{(k+n)} (\vec{p}_k)\\
= & - \int^t_0 d s \sum_{\mathbb{T} \in \mathfrak{T}_{n,k}} \sum_{\pi_2 \in \Pi_{k+n}} \sum_{\bar{v} \in M(\mathbb{T})} \int \frac{d \vec{r}_{n+k}}{(k+n)!} \prod_{e \in R_1 (\mathbb{T}) = L_1 (\mathbb{T})} e^{- (t-s) \tilde{p}^2_{\pi_1 (e)}} \delta (p_{\pi_1 (e)} - r_{\pi_2 (e)} )\\
& \times \int \prod_{e \in E_2 (\mathbb{T})} d q_e \prod_{ \substack{ e \in E_2 (\mathbb{T})\\ e \notin D_{\bar{v}}}} \frac{d \tau^e}{\gamma^e - q_e^2 + \mathrm{i} \tau^e} K^{\pi_2}_{\mathbb{T}} u^{(n+k)} (s, \vec{r}_{n+k}) \prod_{ e \in R_2 (\mathbb{T})} e^{- (t-s) (\gamma^e + \mathrm{i} \tau^e)} \delta (q_e - p_{\pi_1 (e)})\\
& \times \prod_{e \in L_2 (\mathbb{T})} \delta (q_e - r_{\pi_2 (e)}) \prod_{\substack{ v \in V(\mathbb{T})\\ v \not= \bar{v} } } \delta \big ( \tau^{e^a_v} - \tau^{e^b_v} - \tau^{e^c_v} \big ) \prod_{ v \in V(\mathbb{T}) } \delta \big (q_{e^a_v} - q_{e^b_v} - q_{e^c_v}\big )\\
=  & - \sum_{\mathbb{T} \in \mathfrak{T}_{n,k}} \int^t_0 d s R_{\mathbb{T}, t-s} u^{(k+n)} (s)
\end{split}
\ee
where $D_{\bar{v}}$ denotes the set of daughter-edges of $\bar{v}.$ The proof is complete.
\end{proof}

In conclusion, combining Theorems \ref{th:Duhamelfulltree} and \ref{th:Duhamelerrortree} yields the expression \eqref{eq:NSEhierarchyTreeExpanssion}.

\section{A prior space-time estimates}\label{SpacetimeEstimate}

In this section, we first present {\it a prior} space-time estimates for interaction operators $C_{\mathbb{T}, t}$ and $R_{\mathbb{T}, t},$ and then give the proof of Theorem \ref{th:unique}.

\subsection{Space-time estimates for collision operators}

We have {\it a prior} space-time estimates for $C_{\mathbb{T}, t}$ as follows.

\begin{thm}\label{th:Collisionspace-timebound}
Fix $\alpha > \frac{1}{2}$ and $k \ge 1.$ Let $M>0.$ Then there exists a constant $C>0$ depending only on $\alpha, k$ such that for any $v^{(k)} \in \mathcal{S}_{(k)} (\mathbb{R}^3)$ satisfying
\beq\label{eq:Testfunctbound}
\sup_{1 \le i_1, \ldots, i_k \le 3} \sup_{p_1, \ldots, p_k \in \mathbb{R}^3} \langle p_1 \rangle^3 \cdots \langle p_k \rangle^3 | v^{(k)}_{i_1, \ldots, i_k} (\vec{p}_k) | \le M,
\eeq
for any $n \ge 0,$ and for all $\mathbb{T} \in \mathfrak{T}_{n,k},$ we have
\beq\label{eq:Collisionspace-timebound}
\big | \big \langle v^{(k)}, C_{\mathbb{T}, t} u^{(n+k)} \big \rangle_{\mathbb{L}^2_{(k)}} \big | \le M C^n \big \| u^{(n+k)} \big \|_{\mathbb{H}^\alpha_{(n+k)}}
\eeq
for all $u^{(n+k)} \in \mathfrak{H}^\alpha_{(n+k)}$ and any $0< t \le 1.$
\end{thm}

\begin{proof}
Fix $k \ge 1$ and $0 < t \le 1.$ Let $M>0.$ Assume that $v^{(k)} \in \mathcal{S}_{(k)} (\mathbb{R}^3)$ satisfying \eqref{eq:Testfunctbound}. Let $\mathbb{T} \in \mathfrak{T}_{n,k}$ and suppose $u^{(n+k)} \in \mathfrak{L}^2_{(n+k)}.$ We have, by the definition of $C_{\mathbb{T}, t} u^{(n+k)}$ in \eqref{eq:collisionoerator},
\be\begin{split}
\big \langle v^{(k)}, C_{\mathbb{T}, t} u^{(n+k)} \big \rangle_{\mathbb{L}^2_{(k)}} = \frac{1}{(k+n)!} & \sum_{\pi_2 \in \Pi_{k+n}} \int d \vec{p}_k d \vec{r}_{n+k} \prod_{e \in R_1 (\mathbb{T}) = L_1 (\mathbb{T})} e^{- t p^2_{\pi_1 (e)}} \delta (p_{\pi_1 (e)} - r_{\pi_2 (e)} )\\
\times & \int \prod_{e \in E_2 (\mathbb{T})} d \tau^e d q_e \bigg [ \sum_{1 \le i_1, \ldots, i_k \le 3} \overline{v^{(k)}_{i_1, \ldots, i_k} (\vec{p}_k)} K_{\mathbb{T}}^{\pi_2} u^{(n+k)}_{i_1, \ldots, i_k} (\vec{r}_{n+k} ) \bigg ]\\
& \times \prod_{e \in R_2 (\mathbb{T})} e^{- t (\gamma^e + \mathrm{i} \tau^e)} \delta (q_e - p_{\pi_1 (e)})\prod_{e \in L_2 (\mathbb{T})} \delta (q_e - r_{\pi_2 (e)})\\
& \times \prod_{e \in E_2 (\mathbb{T})} \frac{1}{\gamma^e - q_e^2 + \mathrm{i} \tau^e} \prod_{v \in V(\mathbb{T})} \delta \big ( \tau^{e^a_v} - \tau^{e^b_v} - \tau^{e^c_v} \big ) \delta \big (q_{e^a_v} - q_{e^b_v} - q_{e^c_v}\big ).
\end{split}\ee
For simplicity, we will write $E = E (\mathbb{T}), R_2 = R_2 (\mathbb{T}),$ $L_1 = L_1 (\mathbb{T}),$ etc. Since $u^{(n+k)} (\vec{r}_{n+k})$ is symmetry with respect to the permutation on $\vec{r}_{n+k},$ the integral on the right-hand side of the above equation has the same value for every $\pi_2 \in \Pi_{n+k},$ and hence, instead of averaging over $\pi_2,$ we fix one $\pi_2$ so that $\pi_2 (e) \ge |R_1| +1$ for all $e \in L_2.$ Then, using all the $\delta$-functions and integrating over the variables $\vec{p}_k$ and $\vec{r}_{n+k},$ one has
\be
\begin{split}
\big | \big \langle v^{(k)}, C_{\mathbb{T}, t} & u^{(n+k)} \big \rangle_{\mathbb{L}^2_{(k)}} \big | \le
M e^{- t \sum_{e \in R_2} \gamma^e } \int \prod_{e \in E} d q_e \prod_{e \in E_2} d \tau^e \prod_{e \in E_2} \frac{1}{| \gamma^e - q_e^2 + \mathrm{i} \tau^e |} \prod_{e \in R} \frac{1}{ \langle q_e \rangle^3}\\
\times & \prod_{v \in V} \delta \big ( \tau^{e^a_v} - \tau^{e^b_v} - \tau^{e^c_v} \big ) \delta \big (q_{e^a_v} - q_{e^b_v} - q_{e^c_v} \big )
\sum_{1 \le i_1, \ldots, i_k \le 3} | K_{\mathbb{T}}^{\pi_2} u^{(n+k)}_{i_1, \ldots, i_k} (q_e: \; e \in L ) |
\end{split}
\ee
where we have used the assumption \eqref{eq:Testfunctbound}, and the permutation symmetry of $u^{(n+k)},$ namely, $u^{(n+k)}$ depends only on the set of the variables $q_e$ associated with the leaves of $\mathbb{T},$ but not on the order of those variables.

Choosing $\gamma^e = - \frac{1}{t}$ for all $e \in L_2,$ we have that $\gamma^e \le - \frac{1}{t}$ for every $e \in E_2,$ and
\be
\sum_{e \in R_2} \gamma^e = - (n + |R_2|)\frac{1}{t}.
\ee
Moreover, by the definition of $K_{\mathbb{T}}^{\pi_2} u^{(n+k)} $ (see \eqref{eq:K_vdf} and \eqref{eq:K_Tdf}), we have
\beq\label{eq:KoperatorEst}
| K_{\mathbb{T}}^{\pi_2} u^{(n+k)}_{i_1, \ldots, i_k} (q_e) | \le 6^n \prod_{v \in V}  | q_{e^b_v} + q_{e^c_v} | \sum_{1 \le j_{|R_1| +1}, \ldots, j_{n+k} \le 3} \big | u^{(n+k)}_{i_1, \ldots, i_{|R_1|}, j_{|R_1|+1}, \ldots, j_{n+k}} (q_e) \big |.
\eeq
Then for a fixed $\alpha > \frac{1}{2},$
\beq\label{eq:Collisionspace-timebound01}
\begin{split}
\big | \big \langle v^{(k)}, C_{\mathbb{T}, t} u^{(n+k)} \big \rangle_{\mathbb{L}^2_{(k)}}  \big | \le
M C^n \int & \prod_{e \in E} d q_e \prod_{e \in E_2} d \tau^e \prod_{e \in E_2} \frac{1}{| (\frac{1}{t} + q_e^2) \mathrm{i} + \tau^e |} \prod_{e \in R} \frac{1}{ \langle q_e \rangle^3}\\
\times \prod_{v \in V} \delta \big ( \tau^{e^a_v} - \tau^{e^b_v} - \tau^{e^c_v} \big ) \delta \big (q_{e^a_v} & - q_{e^b_v} - q_{e^c_v}\big ) | q_{e^b_v} + q_{e^c_v} | \sum_{ 1 \le i_1, \ldots, i_{n+k} \le 3} \big | u^{(n+k)}_{i_1, \ldots, i_{n+k}} (q_e: \; e \in L) \big |\\
\le C^n \big \| u^{(k+n))} \big \|_{\mathbb{H}^\alpha_{(k+n)}} \bigg ( \int & \prod_{e \in L_2} \frac{d q_e}{(1+|q_e|^2)^\alpha} \bigg [ \int \prod_{e \in E \backslash L} d q_e \prod_{e \in E_2} d \tau^e \prod_{e \in R_2} \frac{1}{ \langle q_e \rangle^3}\\
\times \prod_{e \in E_2} \frac{1}{| (1 + q_e^2) \mathrm{i} + \tau^e |} & \prod_{v \in V} \delta \big ( \tau^{e^a_v} - \tau^{e^b_v} - \tau^{e^c_v} \big ) \delta \big (q_{e^a_v} - q_{e^b_v} - q_{e^c_v}\big ) | q_{e^b_v} + q_{e^c_v} | \bigg ]^2 \bigg )^\frac{1}{2}
\end{split}
\eeq
where $C>0$ in the second line is a constant depending only on $\alpha, k.$ Here, we have used the conditions $0 < t \le 1,$ the Cauchy-Schwarz inequality and the integral $\int \prod_{e \in L_1 = R_1} \frac{1}{ ( 1 +|q_e|^2)^\alpha \langle q_e \rangle^6} d q_e < \8.$

Next, we estimate the following integral (note that $E \backslash L = E_2 \backslash L_2$)
\beq\label{eq:Collisionitegralbound}\begin{split}
I = & \int \prod_{e \in E_2 \backslash L_2} d q_e \prod_{e \in E_2} d \tau^e \prod_{e \in E_2} \frac{1}{| (1 + q_e^2) \mathrm{i} + \tau^e |} \prod_{e \in R_2} \frac{1}{ \langle q_e \rangle^3}\\
& \times \prod_{v \in V} \delta \big ( \tau^{e^a_v} - \tau^{e^b_v} - \tau^{e^c_v} \big ) \delta \big (q_{e^a_v} - q_{e^b_v} - q_{e^c_v}\big ) | q_{e^b_v} + q_{e^c_v} |.\\
\end{split}
\eeq
For bounding the integral \eqref{eq:Collisionitegralbound}, we first successively integrate over all $\tau$-variables and then over all momenta except for the momenta of the leaves.

First of all, we claim that
\beq\label{eq:Intovertauvariables}
\begin{split}
\int & \prod_{e \in E_2} \frac{d \tau^e}{| (1 + q_e^2) \mathrm{i} + \tau^e |} \prod_{v \in V} \delta \big ( \tau^{e^a_v} - \tau^{e^b_v} - \tau^{e^c_v} \big )\\
& \lesssim_\varepsilon \prod_{v:\; e^a_v \notin R_2} \frac{1}{(1 + q^2_{e^b_v} + q^2_{e^c_v})^{1-\varepsilon}}
\prod_{v:\; e^a_v \in R_2} \frac{1}{(1 + q^2_{e^a_v} + q_{e^b_v}^2 + q_{e^c_v}^2 )^{1-\varepsilon}}.
\end{split}\eeq
where $\varepsilon$ is a small constant which will be specified later.

Note that the delta functions relate variables within the same trees, the integration then can be done independently in each tree of $\mathbb{T}.$ The order of integration is prescribed according to the converse order of vertices of the trees (see Appendix \ref{Binarytree} below), that is, the $\tau$-variables of a vertex $v$ will be integrated only when those of all vertices $v'$ with $v \prec v'$ have already been integrated out.

Now, we choose a vertex $v \in V$ and suppose that the $\tau$-integrations over all $v' \succ v$ have been performed. We will perform the integration over the $\tau$-variables associated with the daughter-edges of the vertex $v.$ We need to distinguish two cases according to whether the mother-edge of $v,$ $e^a_v$ (with the notation of Fig.\ref{fig:one-vertex}), is the root or not.

{$\bullet$ $e^a_v$ is not a root:}\; 1)\; Both $e^b_v$ and $e^c_v$ are leaves. By Lemma \ref{lem:FrequencyestimateM} one has
\be\begin{split}
\int \frac{d \tau^b d \tau^c \delta (\tau^a - \tau^b - \tau^c)}{| (1 + q_b^2) \mathrm{i} + \tau^b | | (1 + q_c^2) \mathrm{i} + \tau^c |} & = \int \frac{d \tau^b}{| (1 + q_b^2) \mathrm{i} + \tau^b | | (1 + q_c^2) \mathrm{i} + \tau^a - \tau^b |}\\
& \lesssim_\varepsilon \frac{1}{|\tau^a + (1 + q^2_b + q^2_c) \mathrm{i}|^{1 - \varepsilon}}
\end{split}\ee
for all $\tau^a \in \mathbb{R}$ and all $q_b, q_c \in \mathbb{R}^3.$

2) One of $e^b_v$ and $e^c_v$ is a leaf. Assuming that there exists $v'$ such that $v' \succ v$ with $e^a_{v'} = e^c_v$ and $e^b_v$ is a leaf, we have
\be\begin{split}
\int & \frac{d \tau^b d \tau^c \delta (\tau^a - \tau^b - \tau^c)}{| (1 + q_b^2) \mathrm{i} + \tau^b | | (1 + q_c^2) \mathrm{i} + \tau^c | |(1 + q^2_{b'} + q^2_{c'}) \mathrm{i} + \tau^c |^{1 - \varepsilon}} \\
& \le \frac{1}{(1 + q^2_{b'} + q^2_{c'})^{1-\varepsilon}} \int \frac{d \tau^c }{| (1 + q_b^2) \mathrm{i} + \tau^a - \tau^c | | (1 + q_c^2) \mathrm{i} + \tau^c |}\\
& \le \frac{1}{(1 + q^2_{b'} + q^2_{c'})^{1-\varepsilon}} \frac{1}{|\tau^a + (1 + q_b^2 + q_c^2 ) \mathrm{i}|^{1-\varepsilon}}.
\end{split}\ee
Similarly, the same inequality holds when $e^a_{v'} = e^b_v$ and $e^c_v$ is a leaf.

3) Neither $e^b_v$ nor $e^c_v$ is a leaf, i.e.,there are $v', v'' \succ v$ such that $e^b_v = e^a_{v'}$ and $e^c_v = e^a_{v''}.$ Then we have
\be\begin{split}
\int & \frac{d \tau^b d \tau^c \delta (\tau^a - \tau^b - \tau^c)}{| (1 + q_b^2) \mathrm{i} + \tau^b | | (1 + q_c^2) \mathrm{i} + \tau^c | |(1 + q^2_{b'} + q^2_{c'}) \mathrm{i} + \tau^b |^{1 - \varepsilon} |(1 + q^2_{b''} + q^2_{c''}) \mathrm{i} + \tau^c |^{1 - \varepsilon}} \\
& \le \frac{1}{(1 + q^2_{b'} + q^2_{c'})^{1-\varepsilon} (1 + q^2_{b''} + q^2_{c''})^{1-\varepsilon}} \int \frac{d \tau^c }{| (1 + q_b^2) \mathrm{i} + \tau^a - \tau^c | | (1 + q_c^2 ) \mathrm{i} + \tau^c | } \\
& \le \frac{1}{(1 + q^2_{b'} + q^2_{c'})^{1-\varepsilon} (1 + q^2_{b''} + q^2_{c''})^{1-\varepsilon}} \frac{1}{|\tau^a + (1 + q_b^2 + q_c^2 ) \mathrm{i}|^{1-\varepsilon}}.
\end{split}\ee

{$\bullet$ $e^a_v$ is a root:}\; In this case, we will integrate over all the $\tau$-variables associated with the edges of the vertex including the mother-edge. We have three different cases yet.

1) Both $e^b_v$ and $e^c_v$ are leaves. In this case, we have
\be\begin{split}
\int & \frac{d \tau^a d \tau^b d \tau^c \delta (\tau^a - \tau^b - \tau^c)}{| (1 + q_a^2) \mathrm{i} + \tau^a | | (1 + q_b^2) \mathrm{i} + \tau^b | | (1 + q_c^2) \mathrm{i} + \tau^c |}\\
& = \int \frac{d \tau^a d \tau^b}{| (1 + q_a^2) \mathrm{i} + \tau^a | | (1 + q_b^2) \mathrm{i} + \tau^b | | (1 + q_c^2) \mathrm{i} + \tau^a - \tau^b |}\\
& \lesssim_\varepsilon \frac{1}{(1 + q^2_a + q^2_b + q^2_c)^{1 - \varepsilon}}
\end{split}\ee
where we have used Lemma \ref{lem:FrequencyestimateM} twice.

2) One of $e^b_v$ and $e^c_v$ is a leaf. Suppose that there exists $v'$ such that $v' \succ v$ with $e^a_{v'} = e^c_v$ and $e^b_v$ is a leaf, we have
\be\begin{split}
\int & \frac{d \tau^a d \tau^b d \tau^c \delta (\tau^a - \tau^b - \tau^c)}{| (1 + q_a^2) \mathrm{i} + \tau^a | | (1 + q_b^2) \mathrm{i} + \tau^b | | (1 + q_c^2) \mathrm{i} + \tau^c | |(1 + q^2_{b'} + q^2_{c'}) \mathrm{i} + \tau^c |^{1 - \varepsilon}} \\
& \le \frac{1}{(1 + q^2_{b'} + q^2_{c'})^{1-\varepsilon}} \int \frac{d \tau^a d \tau^c }{| (1 + q_a^2) \mathrm{i} + \tau^a | | (1 + q_b^2) \mathrm{i} + \tau^a - \tau^c | | (1 + q_c^2) \mathrm{i} + \tau^c |}\\
& \lesssim_\varepsilon \frac{1}{(1 + q^2_{b'} + q^2_{c'})^{1-\varepsilon}} \frac{1}{(1 + q^2_a + q_b^2 + q_c^2 )^{1-\varepsilon}}.
\end{split}\ee
Similarly, the same inequality holds when $e^a_{v'} = e^b_v$ and $e^c_v$ is a leaf.

3) Neither $e^b_v$ nor $e^c_v$ is a leaf, i.e., there are $v', v'' \succ v$ such that $e^b_v = e^a_{v'}$ and $e^c_v = e^a_{v''}.$ Then we have
\be\begin{split}
\int & \frac{d \tau^a d \tau^b d \tau^c \delta (\tau^a - \tau^b - \tau^c)}{| (1 + q_a^2) \mathrm{i} + \tau^a | | (1 + q_b^2) \mathrm{i} + \tau^b | | (1 + q_c^2) \mathrm{i} + \tau^c | |(1 + q^2_{b'} + q^2_{c'}) \mathrm{i} + \tau^b |^{1 - \varepsilon} |(1 + q^2_{b''} + q^2_{c''}) \mathrm{i} + \tau^c |^{1 - \varepsilon}} \\
& \le \frac{1}{(1 + q^2_{b'} + q^2_{c'})^{1-\varepsilon} (1 + q^2_{b''} + q^2_{c''})^{1-\varepsilon}} \int \frac{d \tau^a d \tau^c }{| (1 + q_a^2) \mathrm{i} + \tau^a | | (1 + q_b^2) \mathrm{i} + \tau^a - \tau^c | | (1 + q_c^2 ) \mathrm{i} + \tau^c | } \\
& \lesssim_\varepsilon \frac{1}{(1 + q^2_{b'} + q^2_{c'})^{1-\varepsilon} (1 + q^2_{b''} + q^2_{c''})^{1-\varepsilon}} \frac{1}{(1 + q^2_a + q_b^2 + q_c^2 )^{1-\varepsilon}}.
\end{split}\ee

In summary, we have proven \eqref{eq:Intovertauvariables}.

Then, by \eqref{eq:Intovertauvariables} we have
\be\label{eq:CollisionitegralboundM}\begin{split}
I \lesssim_\varepsilon \int & \prod_{e \in E_2 \backslash L_2} d q_e \prod_{v \in V} \delta \big (q_{e^a_v} - q_{e^b_v} - q_{e^c_v}\big ) \prod_{e \in R_2} \frac{1}{ \langle q_e \rangle^3}\\
& \times \prod_{v:\; e^a_v \notin R_2} \frac{| q_{e^b_v} + q_{e^c_v} |}{(1 + q^2_{e^b_v} + q^2_{e^c_v})^{1-\varepsilon}}
\prod_{v:\; e^a_v \in R_2} \frac{| q_{e^b_v} + q_{e^c_v} |}{(1 + q^2_{e^a_v} + q_{e^b_v}^2 + q_{e^c_v}^2 )^{1-\varepsilon}}.
\end{split}
\ee
Thus we have
\be\begin{split}
\big | \big \langle v^{(k)}, C_{\mathbb{T}, t} u^{(n+k)} \big \rangle_{\mathbb{L}^2_{(k)}}  \big |
\le M C^n & \big \| u^{(k+n))} \big \|_{\mathbb{H}^\alpha_{(k+n)}} \bigg ( \int \prod_{e \in E_2} d q_e \prod_{v \in V} \delta \big (q_{e^a_v} - q_{e^b_v} - q_{e^c_v}\big ) \prod_{e \in L_2} \frac{1}{(1 + |q_e |^2)^\alpha} \\
\times \prod_{e \in R_2} & \frac{1}{ \langle q_e \rangle^6} \prod_{v:\; e^a_v \notin R_2} \frac{| q_{e^b_v} + q_{e^c_v} |^2}{(1 + q^2_{e^b_v} + q^2_{e^c_v})^{2 (1-\varepsilon)}} \prod_{v:\; e^a_v \in R_2} \frac{| q_{e^b_v} + q_{e^c_v} |^2}{(1 + q^2_{e^a_v} + q_{e^b_v}^2 + q_{e^c_v}^2 )^{2 (1-\varepsilon)}} \bigg )^\frac{1}{2}
\end{split}\ee
where $C>0$ is a constant depending only on $\alpha, k,$ and $\varepsilon.$

It remains to estimate the momenta integrations
\be\begin{split}
\Xi = & \int \prod_{e \in E_2} d q_e \prod_{v \in V} \delta \big (q_{e^a_v} - q_{e^b_v} - q_{e^c_v}\big ) \prod_{e \in L_2} \frac{1}{(1 + |q_e |^2)^\alpha} \prod_{e \in R_2} \frac{1}{ \langle q_e \rangle^6}\\
& \times \prod_{v:\; e^a_v \notin R_2} \frac{| q_{e^b_v} + q_{e^c_v} |^2}{(1 + q^2_{e^b_v} + q^2_{e^c_v})^{2 (1-\varepsilon)}} \prod_{v:\; e^a_v \in R_2} \frac{| q_{e^b_v} + q_{e^c_v} |^2}{(1 + q^2_{e^a_v} + q_{e^b_v}^2 + q_{e^c_v}^2 )^{2 (1-\varepsilon)}}.
\end{split}\ee
Since the delta functions relate variables within the same trees, we only need to consider the integration over all $q$-variables associated with a tree $\mathrm{T},$ that is
\beq\label{eq:Integralonetreebound}\begin{split}
\Xi (\mathrm{T}) = \int \prod_{e \in E_2 (\mathrm{T})} d q_e \prod_{v \in V (\mathrm{T})} \delta \big (q_{e^a_v} - q_{e^b_v} - q_{e^c_v}\big ) & \frac{1}{ \langle q_{R(\mathrm{T})} \rangle^6} \frac{| q_{e^b_v} + q_{e^c_v} |^2}{(1 + q^2_{e^a_v = R(\mathrm{T})} + q_{e^b_v}^2 + q_{e^c_v}^2 )^{2 (1-\varepsilon)}}\\
\times \prod_{e \in L_2 (\mathrm{T})} \frac{1}{(1 + |q_e |^2)^\alpha} & \prod_{v \in V (\mathrm{T}):\; e^a_v \not= R(\mathrm{T})} \frac{| q_{e^b_v} + q_{e^c_v} |^2}{(1 + q^2_{e^b_v} + q^2_{e^c_v})^{2 (1-\varepsilon)}}.
\end{split}\eeq
Again, we begin the integration with the $q$-variables of the leaves and proceed toward the root, and a vertex $v$ will be integrated only when all vertices $v'$ with $v \prec v'$ have already been integrated out.

Indeed, for a maximal vertex $v$ whose two daughter edges must be leaves, if $e^a_v \not= R$ then the associated integration is
\be\begin{split}
\Xi_v = & \int d q_{e^b_v} d q_{e^c_v} \delta \big (q_{e^a_v} - q_{e^b_v} - q_{e^c_v}\big ) \frac{1}{(1 + |q_{e^b_v} |^2)^\alpha} \frac{1}{(1 + |q_{e^c_v} |^2)^\alpha} \frac{| q_{e^b_v} + q_{e^c_v} |^2}{(1 + q^2_{e^b_v} + q^2_{e^c_v})^{2 (1-\varepsilon)}}\\
= & \int d q_{e^b_v} \frac{1}{(1 + |q_{e^b_v} |^2)^\alpha} \frac{1}{(1 + |q_{e^a_v} - q_{e^b_v} |^2)^\alpha} \frac{| q_{e^a_v}|^2}{(1 + q^2_{e^b_v} + (q_{e^a_v} - q_{e^b_v})^2 )^{2 (1-\varepsilon)}}.
\end{split}\ee
Taking $\varepsilon = \frac{1}{4} \min \{ 1, \alpha - \frac{1}{2} \} >0$ (noticing that $\alpha > \frac{1}{2}$), we claim that there exists a constant $C_\alpha >0$ depending only on $\alpha$ such that
\beq\label{eq:Integralonevertexbound}
\Xi_v \le C_\alpha \frac{1}{(1 + |q_{e^a_v} |^2)^\alpha}.
\eeq
For proving this inequality, noticing that $q_{e^a_v}^2 \lesssim q^2_{e^b_v} + (q_{e^a_v} - q_{e^b_v})^2,$ we have
\be
\frac{q_{e^a_v}^2}{(1 + q^2_{e^b_v} + (q_{e^a_v} - q_{e^b_v})^2 )^{2 (1-\varepsilon)}} \lesssim \frac{1}{[ q^2_{e^b_v} + (q_{e^a_v} - q_{e^b_v})^2 ]^{1- 2 \varepsilon}}.
\ee
Thus
\be\begin{split}
\Xi_v  \lesssim & \frac{1}{(1 + |q_{e^a_v} |^2)^\alpha} \int d q_{e^b_v} \frac{(1 + q^2_{e^b_v} + (q_{e^a_v} - q_{e^b_v})^2 )^\alpha}{(1 + |q_{e^b_v} |^2)^\alpha (1 + |q_{e^a_v} - q_{e^b_v} |^2)^\alpha} \frac{1}{(q^2_{e^b_v} + (q_{e^a_v} - q_{e^b_v})^2 )^{1-2\varepsilon}}\\
\le & \frac{C_\alpha}{(1 + |q_{e^a_v} |^2)^\alpha} \int d q_{e^b_v} \frac{1}{(1 + |q_{e^b_v} |^2)^\alpha |q_{e^b_v} |^{2(1-2\varepsilon)}}
\le \frac{C_\alpha}{(1 + |q_{e^a_v} |^2)^\alpha}
\end{split}\ee
provided $\alpha > \frac{1}{2}$ and $\varepsilon = \frac{1}{4} \min \{ 1, \alpha - \frac{1}{2} \} >0.$ This completes the proof of the inequality \eqref{eq:Integralonevertexbound}.

Subsequently, every vertex $v$ for which all vertices $v'$ with $v \prec v'$ have already been integrated out is associated with the integration of the form $\Xi_v$ as above when $e^a_v \not= R.$ Thus, we obtain
\be\begin{split}
\Xi (\mathrm{T}) \le C_\alpha^{|V(\mathrm{T})|-1} \int & d q_{e^a_v = R(\mathrm{T})} d q_{e^b_v} d q_{e^c_v} \delta \big (q_{e^a_v} - q_{e^b_v} - q_{e^c_v}\big ) \frac{1}{ \langle q_{R(\mathrm{T})} \rangle^6}\\
& \times \frac{| q_{e^b_v} + q_{e^c_v} |^2}{(1 + q^2_{e^a_v = R(\mathrm{T})} + q_{e^b_v}^2 + q_{e^c_v}^2 )^{2 (1-\varepsilon)}}
\frac{1}{(1 + |q_{e^b_v} |^2)^\alpha} \frac{1}{(1 + |q_{e^c_v} |^2)^\alpha} \le C_\alpha^{|V(\mathrm{T})|}.
\end{split}\ee
Therefore, we conclude that
\be
\Xi = \prod_{\mathrm{T} \in \mathbb{T}} \Xi (\mathrm{T}) \le C_\alpha^n.
\ee
This completes the proof.
\end{proof}

\subsection{Space-time estimates for error operators}

In this subsection, we prove a space-time estimate for the error term $R_{\mathbb{T}, t}.$

\begin{thm}\label{th:Collisionspace-timebounderror}
Fix $k \ge 1.$ Let $M>0.$ Then there exists a constant $C>0$ depending only on $k$ such that for any $v^{(k)} \in \mathcal{S}_{(k)} (\mathbb{R}^3)$ satisfying \eqref{eq:Testfunctbound} with the bound $M,$ for any $n \ge 0$ and any $\mathbb{T} \in \mathfrak{T}_{n,k},$ we have
\beq\label{eq:Collisionspace-timebounderror}
\big | \big \langle v^{(k)}, R_{\mathbb{T}, t} u^{(n+k)} \big \rangle_{\mathbb{L}^2_{(k)}} \big | \le M C^n t^{\frac{n}{2} -1}\big \| u^{(n+k)} \big \|_{\mathbb{H}^1_{(n+k)}}
\eeq
for all $u^{(n+k)} \in \mathfrak{H}^1_{(n+k)}$ and any $0< t \le 1.$
\end{thm}

\begin{proof}
By \eqref{eq:Errortermoperator}, we have
\be\begin{split}
\big \langle v^{(k)}, R_{\mathbb{T}, t} u^{(n+k)} \big \rangle_{\mathbb{L}^2_{(k)}} = \frac{1}{(k+n)!} \sum_{\pi_2 \in \Pi_{k+n}} \int d \vec{p}_k d \vec{r}_{n+k}&
\sum_{\bar{v} \in M(\mathbb{T})} \prod_{e \in R_1 (\mathbb{T}) = L_1 (\mathbb{T})} e^{- t p^2_{\pi_1 (e)}} \delta (p_{\pi_1 (e)} - r_{\pi_2 (e)} )\\
\times \int \prod_{e \in E_2 (\mathbb{T})} d q_e \prod_{ \substack{ e \in E_2 (\mathbb{T})\\ e \notin D_{\bar{v}} } } \frac{d \tau^e}{\gamma^e - q_e^2 + \mathrm{i} \tau^e} & \bigg [ \sum_{1 \le i_1, \ldots, i_k \le 3} \overline{v^{(k)}_{i_1, \ldots, i_k} (\vec{p}_k)} K_{\mathbb{T}}^{\pi_2} u^{(n+k)}_{i_1, \ldots, i_k} (\vec{r}_{n+k} ) \bigg ]\\
\times \prod_{ e \in R_2 (\mathbb{T})} & e^{- t(\gamma^e + \mathrm{i} \tau^e)} \delta (q_e - p_{\pi_1 (e)}) \prod_{e \in L_2 (\mathbb{T})} \delta (q_e - r_{\pi_2 (e)})\\ \times \prod_{ \substack{ v \in V(\mathbb{T})\\ v \not= \bar{v} } } & \delta \big ( \tau^{e^a_v} - \tau^{e^b_v} - \tau^{e^c_v} \big ) \prod_{ v \in V(\mathbb{T}) } \delta \big (q_{e^a_v} - q_{e^b_v} - q_{e^c_v} \big ).
\end{split}\ee
Recall that $M (\mathbb{T})$ is the set of maximal elements in $V ( \mathbb{T})$ and $D_v$ denotes the set of daughter-edges for a vertex $v \in M ( \mathbb{T}).$ In the integration of the right hand side of the above equation, if there exists a $\bar{v} \in M (\mathbb{T})$ such that $e^a_{\bar{v}} \in R (\mathbb{T}),$ then there is only one denominator containing $\tau^e$ in the integral for this tree, and so the associated $\tau^e$-integral would not be absolutely convergent. In this case, using \eqref{eq:PropagatorIdent1}, we perform the integration over the $\tau^{e^a_{\bar{v}}}$ and then obtain a factor $e^{- t q^2_{e^a_{\bar{v}}}}.$ Denote by $E_2 (\mathbb{T}, \bar{v}) = E_2 (\mathbb{T}) \backslash \{ e^a_{\bar{v}}\}$ if $e^a_{\bar{v}} \in R (\mathbb{T}),$ and otherwise $E_2 (\mathbb{T}, \bar{v}) = E_2 (\mathbb{T}).$ After performing the integration over $\tau^e$ associated to $e \notin E_2 (\mathbb{T}, \bar{v}),$ we take the absolute value of the above integrand.

As done in the proof of Theorem \ref{th:Collisionspace-timebound}, instead of averaging over $\pi_2,$ we fix one $\pi_2$ so that $\pi_2 (e) \ge |R_1| +1$ for all $e \in L_2.$ Then by \eqref{eq:KoperatorEst}, we obtain that (noticing that $\gamma^e = - \frac{1}{t}$ for all $e \in L_2$)
\be
\begin{split}
\big | \big \langle v^{(k)}, R_{\mathbb{T}, t} u^{(n+k)} \big \rangle_{\mathbb{L}^2_{(k)}}  \big | \le
M C^n \sum_{\bar{v} \in M(\mathbb{T})} \int & \prod_{e \in E} d q_e \prod_{ \substack{ e \in E_2 (\mathbb{T}, \bar{v})\\ e \notin D_{\bar{v}} }} \frac{d \tau^e}{| (\frac{1}{t} + q_e^2) \mathrm{i} + \tau^e |} \prod_{e \in R} \frac{1}{ \langle q_e \rangle^3}\\
\times \prod_{ \substack{ v \in V(\mathbb{T})\\ v \not= \bar{v} } } \delta \big ( \tau^{e^a_v} - \tau^{e^b_v} - \tau^{e^c_v} \big ) \prod_{v \in V} \delta \big ( q_{e^a_v} - & q_{e^b_v} - q_{e^c_v}\big ) | q_{e^b_v} + q_{e^c_v} | \sum_{ 1 \le i_1, \ldots, i_{n+k} \le 3} \big | u^{(n+k)}_{i_1, \ldots, i_{n+k}} (q_e: \; e \in L) \big |\\
\le M C^n \big \| u^{(k+n))} \big \|_{\mathbb{H}^1_{(k+n)}} \sum_{\bar{v} \in M(\mathbb{T})} \bigg ( \int \prod_{e \in L_2} & \frac{d q_e}{ \langle q_e \rangle^2 } \bigg [ \int \prod_{e \in E_2 \backslash L_2} d q_e \prod_{ \substack{ e \in E_2 (\mathbb{T}, \bar{v})\\ e \notin D_{\bar{v}} }} \frac{d \tau^e}{| (\frac{1}{t} + q_e^2) \mathrm{i} + \tau^e |}\\
\times \prod_{e \in R_2} \frac{1}{ \langle q_e \rangle^3} \prod_{ \substack{ v \in V(\mathbb{T})\\ v \not= \bar{v} } } & \delta \big ( \tau^{e^a_v} - \tau^{e^b_v} - \tau^{e^c_v} \big ) \prod_{v \in V} \delta \big (q_{e^a_v} - q_{e^b_v} - q_{e^c_v}\big ) | q_{e^b_v} + q_{e^c_v} | \bigg ]^2 \bigg )^\frac{1}{2}
\end{split}
\ee
where $C>0$ depends only on $k.$

Making variable substitution as $q_e \longrightarrow t^{- \frac{1}{2}} q_e$ for all $e \in E_2 \backslash L_2,$ and $\tau^e \longrightarrow t^{-1} \tau^e$ for all $e \in E_2 (\mathbb{T}, \bar{v}) \backslash D_{\bar{v}},$ we have
\be\begin{split}
\frac{d \tau^e}{| (1/t + q_e^2) \mathrm{i} + \tau^e |} & \longrightarrow \frac{d \tau^e}{| (1 + q_e^2) \mathrm{i} + \tau^e |},\\
\delta \big ( \tau^{e^a_v} - \tau^{e^b_v} - \tau^{e^c_v} \big ) & \longrightarrow t \delta \big ( \tau^{e^a_v} - \tau^{e^b_v} - \tau^{e^c_v} \big ),\\
\delta \big (q_{e^a_v} - q_{e^b_v} - q_{e^c_v}\big ) & \longrightarrow t^{\frac{3}{2}} \delta \big (q_{e^a_v} - q_{e^b_v} - q_{e^c_v}\big ).
\end{split}\ee
Then
\be
\begin{split}
\big | \big \langle v^{(k)}, R_{\mathbb{T}, t} u^{(n+k)} & \big \rangle_{\mathbb{L}^2_{(k)}}  \big | \le
M C^n t^{\frac{n}{2} -1} \big \| u^{(k+n))} \big \|_{\mathbb{H}^1_{(k+n)}}\\
\times \sum_{\bar{v} \in M(\mathbb{T})} \bigg ( \int & \prod_{e \in L_2} \frac{d q_e}{\langle q_e \rangle^2 } \bigg [ \int \prod_{e \in E_2 \backslash L_2} d q_e \prod_{ \substack{ e \in E_2 (\mathbb{T}, \bar{v})\\ e \notin D_{\bar{v}} }} \frac{d \tau^e}{| (1 + q_e^2) \mathrm{i} + \tau^e |}  \prod_{e \in R_2} \frac{1}{ \langle t^{- \frac{1}{2}} q_e \rangle^3}\\
& \times \prod_{ \substack{ v \in V(\mathbb{T})\\ v \not= \bar{v} } } \delta \big ( \tau^{e^a_v} - \tau^{e^b_v} - \tau^{e^c_v} \big ) \prod_{v \in V} \delta \big (q_{e^a_v} - q_{e^b_v} - q_{e^c_v}\big ) | q_{e^b_v} + q_{e^c_v} | \bigg ]^2 \bigg )^\frac{1}{2}.
\end{split}
\ee

We first estimate $I_{\bar{v}},$ where
\be\begin{split}
I_{\bar{v}} = & \int \prod_{e \in E_2 \backslash L_2} d q_e \prod_{ \substack{ e \in E_2 (\mathbb{T}, \bar{v})\\ e \notin D_{\bar{v}} }} \frac{d \tau^e}{| (1 + q_e^2) \mathrm{i} + \tau^e |}  \prod_{e \in R_2} \frac{1}{ \langle t^{- \frac{1}{2}} q_e \rangle^3}\\
\times & \prod_{ \substack{ v \in V(\mathbb{T})\\ v \not= \bar{v} } } \delta \big ( \tau^{e^a_v} - \tau^{e^b_v} - \tau^{e^c_v} \big ) \prod_{v \in V} \delta \big (q_{e^a_v} - q_{e^b_v} - q_{e^c_v}\big ) | q_{e^b_v} + q_{e^c_v} |.
\end{split}\ee
Given $0< \varepsilon <1$ which will be fixed later, as done in the proof of Theorem \ref{th:Collisionspace-timebound}, we perform integrations over all $\tau$-variables and then obtain
\be\label{eq:CollisionitegralboundM}\begin{split}
I_{\bar{v}} \lesssim_\varepsilon \int & \prod_{e \in E_2 \backslash L_2} d q_e \prod_{v \in V} \delta \big (q_{e^a_v} - q_{e^b_v} - q_{e^c_v}\big ) \prod_{e \in R_2} \frac{1}{ \langle t^{- \frac{1}{2}} q_e \rangle^3}\\
& \times | q_{e^b_{\bar{v}}} + q_{e^c_{\bar{v}}} |\prod_{\substack{v:\; e^a_v \notin R_2\\ v \not= \bar{v}}} \frac{| q_{e^b_v} + q_{e^c_v} |}{(1 + q^2_{e^b_v} + q^2_{e^c_v})^{1-\varepsilon}}
\prod_{ \substack{ v:\; e^a_v \in R_2\\ v \not= \bar{v}} } \frac{| q_{e^b_v} + q_{e^c_v} |}{(1 + q^2_{e^a_v} + q_{e^b_v}^2 + q_{e^c_v}^2 )^{1-\varepsilon}}.
\end{split}
\ee
Thus,
\be
\begin{split}
\big | \big \langle v^{(k)}, R_{\mathbb{T}, t} u^{(n+k)} \big \rangle_{\mathbb{L}^2_{(k)}} & \big | \le
M C^n t^{\frac{n}{2} -1} \big \| u^{(k+n))} \big \|_{\mathbb{H}^1_{(k+n)}}\\
\times \sum_{\bar{v} \in M(\mathbb{T})} \bigg ( \int \prod_{e \in E_2} & d q_e \prod_{v \in V} \delta \big (q_{e^a_v} - q_{e^b_v} - q_{e^c_v} \big ) \prod_{e \in L_2} \frac{1}{ \langle q_e \rangle^2} \prod_{e \in R_2} \frac{1}{ \langle t^{- \frac{1}{2}} q_e \rangle^6}\\
\times | q_{e^b_{\bar{v}}} + q_{e^c_{\bar{v}}} |^2 & \prod_{\substack{v:\; e^a_v \notin R_2\\ v \not= \bar{v}}} \frac{| q_{e^b_v} + q_{e^c_v} |^2}{(1 + q^2_{e^b_v} + q^2_{e^c_v})^{ 2(1-\varepsilon) }} \prod_{ \substack{ v:\; e^a_v \in R_2\\ v \not= \bar{v}} } \frac{| q_{e^b_v} + q_{e^c_v} |^2}{(1 + q^2_{e^a_v} + q_{e^b_v}^2 + q_{e^c_v}^2 )^{2 (1-\varepsilon)}} \bigg )^\frac{1}{2}.
\end{split}
\ee

Next, we estimate the integration on the right hand side of the above inequality. Fix $\bar{v} \in M(\mathbb{T})$ and denote the integration by $\Xi (\mathbb{T}, \bar{v}),$ i.e.,
\be\begin{split}
\Xi (\mathbb{T}, \bar{v}) = & \int \prod_{e \in E_2} d q_e \prod_{v \in V} \delta \big (q_{e^a_v} - q_{e^b_v} - q_{e^c_v} \big ) \prod_{e \in L_2} \frac{1}{ \langle q_e \rangle^2} \prod_{e \in R_2} \frac{1}{ \langle t^{- \frac{1}{2}} q_e \rangle^6}\\
& \times | q_{e^b_{\bar{v}}} + q_{e^c_{\bar{v}}} |^2 \prod_{\substack{v:\; e^a_v \notin R_2\\ v \not= \bar{v}}} \frac{| q_{e^b_v} + q_{e^c_v} |^2}{(1 + q^2_{e^b_v} + q^2_{e^c_v})^{ 2(1-\varepsilon) }} \prod_{ \substack{ v:\; e^a_v \in R_2\\ v \not= \bar{v}} } \frac{| q_{e^b_v} + q_{e^c_v} |^2}{(1 + q^2_{e^a_v} + q_{e^b_v}^2 + q_{e^c_v}^2 )^{2 (1-\varepsilon)}}.
\end{split}\ee
Since the delta functions relate variables within the same trees, we may consider separately the integrations over all $q$-variables associated with each tree. Note that the integration associated with a tree $\mathrm{T}$ not including $\bar{v}$ is the same as $\Xi (\mathrm{T})$ in \eqref{eq:Integralonetreebound} in the proof of Theorem \ref{th:Collisionspace-timebound} and so we have $\Xi (\mathrm{T}) \le C^{|V(\mathrm{T})|},$ where $C$ is a absolute constant. Thus, it remains to estimate the integration over the tree $\mathrm{T}$ containing $\bar{v},$ which we denote by $\Xi (\mathrm{T}, \bar{v}).$

Note that
\beq\label{eq:Maxvertexint}
\int d q_{e^b_{\bar{v}}} d q_{e^c_{\bar{v}}} \frac{\delta ( q_{e^a_{\bar{v}}} - q_{e^b_{\bar{v}}} - q_{e^c_{\bar{v}}} ) | q_{e^b_{\bar{v}}} + q_{e^c_{\bar{v}}} |^2}{ \langle q_{e^b_{\bar{v}}} \rangle^2 \langle q_{e^c_{\bar{v}}} \rangle^2} = |q_{e^a_{\bar{v}}}|^2 \int \frac{ d q_{e^b_{\bar{v}}}}{ \langle q_{e^b_{\bar{v}}} \rangle^2 \langle q_{e^a_{\bar{v}}} - q_{e^b_{\bar{v}}} \rangle^2} \lesssim |q_{e^a_{\bar{v}}}|.
\eeq
If $e^a_{\bar{v}} = R (\mathrm{T}),$ then
\be
\Xi (\mathrm{T}, \bar{v}) = \int d q_{e^a_{\bar{v}}} d q_{e^b_{\bar{v}}} d q_{e^c_{\bar{v}}} \frac{\delta ( q_{e^a_{\bar{v}}} - q_{e^b_{\bar{v}}} - q_{e^c_{\bar{v}}} ) | q_{e^b_{\bar{v}}} + q_{e^c_{\bar{v}}} |^2}{\langle t^{-\frac{1}{2}} q_{e^a_{\bar{v}}} \rangle^6 \langle q_{e^b_{\bar{v}}} \rangle^2 \langle q_{e^c_{\bar{v}}} \rangle^2} \lesssim \int \frac{1}{\langle q_{e^a_{\bar{v}}} \rangle^5} d q_{e^a_{\bar{v}}} < \8
\ee
where we have used the fact that $\langle t^{-\frac{1}{2}} q_{e^a_{\bar{v}}} \rangle \ge \langle q_{e^a_{\bar{v}}} \rangle$ (because $0< t \le 1$). It remains to deal with the case $e^a_{\bar{v}} \not= R (\mathrm{T}),$ where
\be\begin{split}
\Xi (\mathrm{T}, \bar{v}) = \int \prod_{e \in E_2 (\mathrm{T})} d q_e \prod_{v \in V (\mathrm{T})} \delta \big (q_{e^a_v} - q_{e^b_v} - q_{e^c_v}\big ) & \frac{1}{ \langle t^{-\frac{1}{2}}q_{R(\mathrm{T})} \rangle^6} \frac{| q_{e^b_v} + q_{e^c_v} |^2}{(1 + q^2_{e^a_v = R(\mathrm{T})} + q_{e^b_v}^2 + q_{e^c_v}^2 )^{2 (1-\varepsilon)}}\\
\times | q_{e^b_{\bar{v}}} + q_{e^c_{\bar{v}}} |^2 & \prod_{e \in L_2 (\mathrm{T})} \frac{1}{ \langle q_e \rangle^2} \prod_{\substack{ v \in V (\mathrm{T}) \backslash \{ \bar{v} \}\\ e^a_v \not= R(\mathrm{T})}} \frac{| q_{e^b_v} + q_{e^c_v} |^2}{(1 + q^2_{e^b_v} + q^2_{e^c_v})^{2 (1-\varepsilon)}}.
\end{split}\ee

At first, for a maximal vertex $v \in V (\mathrm{T}) \backslash \{ \bar{v} \}$ whose two daughter edges must be leaves, the associated integration is
\be\begin{split}
\Xi_v = & \int d q_{e^b_v} d q_{e^c_v} \delta \big (q_{e^a_v} - q_{e^b_v} - q_{e^c_v}\big ) \frac{1}{ \langle q_{e^b_v} \rangle^2} \frac{1}{\langle q_{e^c_v} \rangle^2} \frac{| q_{e^b_v} + q_{e^c_v} |^2}{(1 + q^2_{e^b_v} + q^2_{e^c_v})^{2 (1-\varepsilon)}}\\
= & \int d q_{e^b_v} \frac{1}{\langle q_{e^b_v} \rangle^2} \frac{1}{\langle q_{e^a_v} - q_{e^b_v} \rangle^2} \frac{| q_{e^a_v}|^2}{(1 + q^2_{e^b_v} + (q_{e^a_v} - q_{e^b_v})^2 )^{2 (1-\varepsilon)}}.
\end{split}\ee
Taking $0 <\varepsilon \le \frac{3}{4},$ we have
\be
\Xi_v \lesssim \frac{1}{\langle q_{e^a_v} \rangle^2}.
\ee
Subsequently, every vertex $v$ with $v \nprec \bar{v}$ for which all vertices $v'$ with $v \prec v'$ have already been integrated out is associated with the integration of the form $\Xi_v$ as above. On the other hand, by \eqref{eq:Maxvertexint}, the integration associated with $\bar{v}$ is
\be\begin{split}
\Xi_{\bar{v}} = & \int d q_{e^b_{\bar{v}}} d q_{e^c_{\bar{v}}} \delta \big ( q_{e^a_{\bar{v}}} - q_{e^b_{\bar{v}}} - q_{e^c_{\bar{v}}}\big ) \frac{1}{ \langle q_{e^b_{\bar{v}}} \rangle^2} \frac{1}{\langle q_{e^c_{\bar{v}}} \rangle^2} | q_{e^b_{\bar{v}}} + q_{e^c_{\bar{v}}} |^2 \lesssim \langle q_{e^a_{\bar{v}}} \rangle.
\end{split}\ee

Secondly, suppose that $v$ is a vertex with one of the leaves being $e^a_{\bar{v}},$ for instance $e^b_v = e^a_{\bar{v}}.$ If $e^a_v \not= R(\mathrm{T}),$ then the integration is
\be\begin{split}
\Xi_v = & \int d q_{e^b_v} d q_{e^c_v} \delta \big (q_{e^a_v} - q_{e^b_v} - q_{e^c_v}\big ) \frac{\langle q_{e^b_v} \rangle}{\langle q_{e^c_v} \rangle^2} \frac{| q_{e^b_v} + q_{e^c_v} |^2}{(1 + q^2_{e^b_v} + q^2_{e^c_v})^{2 (1-\varepsilon)}}\\
\le & | q_{e^a_v}|^2 \int d q_{e^b_v} \frac{1}{\langle q_{e^a_v} - q_{e^b_v} \rangle^2} \frac{1}{\langle q_{e^b_v} \rangle^{3 - 4 \varepsilon}} \lesssim | q_{e^a_v}|
\end{split}\ee
provided $0 < \varepsilon \le \frac{1}{4}.$ Subsequently, for each $v \prec \bar{v}$ with $e^a_v \not= R(\mathrm{T}),$ the associated integration is the same as this $\Xi_v$ and so $\lesssim | q_{e^a_v}|.$

Finally, we will arrive the vertex $v \prec \bar{v}$ with $e^a_v = R(\mathrm{T}).$ Then the integration is
\be\begin{split}
\Xi_v = & \int d q_{e^a_v} d q_{e^b_v} d q_{e^c_v} \delta \big (q_{e^a_v} - q_{e^b_v} - q_{e^c_v}\big ) \frac{1}{ \langle t^{-\frac{1}{2}}q_{R(\mathrm{T})} \rangle^6} \frac{\langle q_{e^b_v} \rangle}{\langle q_{e^c_v} \rangle^2} \frac{| q_{e^b_v} + q_{e^c_v} |^2}{(1 + q^2_{e^a_v = R(\mathrm{T})} + q^2_{e^b_v} + q^2_{e^c_v})^{2 (1-\varepsilon)}}\\
\le & \int d q_{e^a_v} d q_{e^b_v} \frac{1}{\langle q_{R(\mathrm{T})} \rangle^4} \frac{\langle q_{e^b_v} \rangle}{\langle q_{e^a_v} - q_{e^b_v} \rangle^2} \frac{1}{(1 + q^2_{e^a_v = R(\mathrm{T})} + q^2_{e^b_v} + (q_{e^a_v} - q_{e^b_v})^2 )^{2 (1-\varepsilon)}} \le C
\end{split}\ee
for some absolute constant $C< \8$ uniformly for all $0 < \varepsilon \le \frac{1}{4}.$

Thus, if $0 < \varepsilon \le \frac{1}{4}$ we have
\be
\Xi (\mathrm{T}, \bar{v}) \le C^{|V (\mathrm{T})|}
\ee
and hence
\be
\Xi (\mathbb{T}, \bar{v}) = \Xi (\mathrm{T}, \bar{v}) \prod_{\bar{v} \notin \mathrm{T}} \Xi (\mathrm{T}) \le C^n
\ee
because $V ( \mathbb{T}) =n.$ Note that $|M (\mathbb{T})| \le n.$ Therefore, we conclude \eqref{eq:Collisionspace-timebounderror}.
\end{proof}

\subsection{Proof for uniqueness}

We are now ready to prove Theorem \ref{th:unique}. Let $T>0.$ For a given $C>0,$ suppose that $\mathfrak{U}_1(t) = (u^{(k)}_1 (t))_{k \ge 1}$ and $\mathfrak{U}_2 (t) = (u^{(k)}_2 (t))_{k \ge 1}$ are two mild solutions in $[0,T]$ to the hierarchy \eqref{eq:NShierarchyIntEqua} such that $\mathfrak{U}_1(0) = \mathfrak{U}_2 (0),$ and for every $i=1,2,$  $u^{(k)}_i \in L^{\8} ([0, T], \mathfrak{H}^1_{(k)})$ satisfying the bound
\be
\| u^{(k)}_i \|_{L^\8([0,T], \mathbb{H}^1_{(k)})} \le C^k
\ee
for all $k \ge 1.$ We need to prove that $\mathfrak{U}_1(t) = \mathfrak{U}_2(t)$ for every $t \in [0, T].$ In fact, it suffices to prove that for each fixed $k \ge 1,$ $u^{(k)}_1 (t) = u^{(k)}_2 (t)$ for all $t \in [0,T].$

To this end, for a given $k \ge 1,$ we can expand $u^{(k)}_i (t)$ ($i=1,2$) in a Duhamel expansion as in \eqref{eq:DuhamelExpand}. Then by Theorems \ref{th:Duhamelfulltree} and \ref{th:Duhamelerrortree}, we have
\beq\label{eq:DuhamelExpandGraph}
\begin{split}
u^{(k)}_i (t) = \mathcal{T}^{(k)} (t) u^{(k)}_i (0) + \sum^{n-1}_{m=1} \sum_{\mathbb{T} \in \mathfrak{T}_{m,k}} C_{\mathbb{T}, t} u^{(k+m)}_i (0) - \sum_{\mathbb{T} \in \mathfrak{T}_{n,k}} \int^t_0 d s R_{\mathbb{T}, t-s} u^{(k+n)}_i (s)
\end{split}
\eeq
for all $n >1,$ where $i=1,2.$ For any $v^{(k)} \in \mathcal{S}_{(k)} (\mathbb{R}^3),$ one has
\be
\begin{split}
\langle v^{(k)}, u^{(k)}_i (t) \rangle_{\mathbb{L}^2_{(k)}} = \langle v^{(k)}, \mathcal{T}^{(k)} (t) u^{(k)}_i (0) \rangle_{\mathbb{L}^2_{(k)}} + & \sum^{n-1}_{m=1} \sum_{\mathbb{T} \in \mathfrak{T}_{m,k}} \langle v^{(k)}, C_{\mathbb{T}, t} u^{(k+m)}_i (0) \rangle_{\mathbb{L}^2_{(k)}}\\
& - \sum_{\mathbb{T} \in \mathfrak{T}_{n,k}} \int^t_0 d s \langle v^{(k)}, R_{\mathbb{T}, t-s} u^{(k+n)}_i (s) \rangle_{\mathbb{L}^2_{(k)}}
\end{split}
\ee
for $i=1,2.$ By Theorem \ref{th:Collisionspace-timebound}, the terms in the sum over $m$ are all finite. Since $\mathfrak{U}_1(0) = \mathfrak{U}_2 (0),$ when taking the difference between $\langle v^{(k)}, u^{(k)}_1 (t) \rangle_{\mathbb{L}^2_{(k)}}$ and $\langle v^{(k)}, u^{(k)}_2 (t) \rangle_{\mathbb{L}^2_{(k)}},$ the free evolution terms $\mathcal{T}^{(k)} (t) u^{(k)}_i (0)$ and all the terms in the sum over $m$ disappear, and so we have
\be
\langle v^{(k)}, u^{(k)}_1 (t) - u^{(k)}_2 (t) \rangle_{\mathbb{L}^2_{(k)}} = - \sum_{\mathbb{T} \in \mathfrak{T}_{n,k}} \int^t_0 d s \big \langle v^{(k)}, R_{\mathbb{T}, t-s} [ u^{(k+n)}_1 (s) - u^{(k+n)}_2 (s) ] \big \rangle_{\mathbb{L}^2_{(k)}}
\ee
for any $n > 1.$ Then for $0< t \le 1$ and $n >1,$ by Theorem \ref{th:Collisionspace-timebounderror} and the estimation $|\mathfrak{T}_{n,k} | \le 2^{3 n +k}$ (cf. \eqref{eq:ForestNo}) we have
\be\begin{split}
\big | \langle & v^{(k)},u^{(k)}_1 (t) - u^{(k)}_2 (t) \rangle_{\mathbb{L}^2_{(k)}} \big |\\
& \le M C^n \int^t_0 d s (t -s)^{\frac{n}{2}-1} \big ( \| u^{(k+n)}_1 (s) \|_{\mathbb{H}^1_{(n + k)}} + \| u^{(k+n)}_2 (s) \|_{\mathbb{H}^1_{(n + k)}} \big ) \le M C^n t^\frac{n}{2}
\end{split}\ee
where we have used the assumption that $\| u^{(n+k)}_i \|_{L^\8([0,T], \mathbb{H}^1_{(n+k)})} \le C^{n+k}$ for $i=1,2.$ Hence, taking $t = \min \big \{ \frac{1}{2}, \frac{1}{(2 C)^2} \big \},$ we have
\be
\big | \langle v^{(k)}, u^{(k)}_1 (t) - u^{(k)}_2 (t) \rangle_{\mathbb{L}^2_{(k)}} \big | \le M \frac{1}{2^n}.
\ee
Since $n >1$ is arbitrary,
\be
\langle v^{(k)}, u^{(k)}_1 (t) - u^{(k)}_2 (t) \rangle_{\mathbb{L}^2_{(k)}} =0
\ee
for any $v^{(k)} \in \mathcal{S}_{(k)} (\mathbb{R}^3),$ provided $t \le \min \big \{ \frac{1}{2}, \frac{1}{(2 C)^2} \big \}.$ Thus, we conclude that
$u^{(k)}_1 (t) = u^{(k)}_2 (t)$ for all $t \le \min \big \{ \frac{1}{2}, \frac{1}{(2 C)^2} \big \}.$ By iteration, we can prove that $u^{(k)}_1 (t) = u^{(k)}_2 (t)$ for all $t \in [0,T].$ This completes the proof of Theorem \ref{th:unique}.

\section{A solution formula for the incompressible Navier-Stokes equation}\label{SolutionFormulaNSE}

By \eqref{eq:NSEhierarchyTreeExpanssion}, we have a formal formula for the solution to the Navier-Stokes hierarchy \eqref{eq:NSEhierarchyW} with an initial data $\mathfrak{U}_0 = (u^{(k)}_0)_{k \ge 1}$ as follows
\beq\label{eq:NSEhierarchySolutionFormula}
\begin{split}
u^{(k)} (t) = e^{t \triangle^{(k)}} u^{(k)}_0 + \sum^\8_{n=1} \sum_{\mathbb{T} \in \mathfrak{T}_{n,k}} C_{\mathbb{T}, t} u^{(n+k)}_0
\end{split}
\eeq
for every $k \ge 1,$ provided the remainder terms converge to zero as $n \to \8.$

The following is to prove such a formula for the Navier-Stokes equation \eqref{eq:NSEpressure-free} in $\mathbb{H}^1 (\mathbb{R}^3).$

\begin{thm}\label{th:NSEsolutionformulaH1}
Let $u_0 \in \mathbb{H}^1 (\mathbb{R}^3)$ with $\bigtriangledown \cdot u_0 =0.$ Let $u$ be the unique weak solution in $C ([0, T^*), \mathbb{H}^1 (\mathbb{R}^3))$ for the initial problem of the Navier-Stokes equation
\be
\left \{\begin{split}
& \partial_t u = \triangle u - W (u \otimes u), \\
& \bigtriangledown \cdot u = 0, \end{split} \right.
\ee
with the initial datum $u (0) = u_0,$ where $T^*$ is the maximal life-time of $u(t).$ Then there exits $0 <t^* < T^*$ such that
\beq\label{eq:NSESolutionFormula}
u (t) = e^{t \triangle} u_0 + \sum^\8_{n=1} \sum_{\mathbb{T} \in \mathfrak{T}_{n,1}} C_{\mathbb{T}, t} u^{\otimes^{n+1}}_0
\eeq
in the sense of distributions for every $0 < t < t^*.$
\end{thm}

\begin{proof}
Fix $0< T <T^*.$ Then $u (t) \in C ([0, T], \mathbb{H}^1 (\mathbb{R}^3))$ so that
\be
C_T = \sup_{t \in [0,T]} \| u(t) \|_{\mathbb{H}^1} < \8.
\ee
Put $u^{(k)} (t) = u (t)^{\otimes^k}$ for every $k \ge 1,$ then by Corollary \ref{cor:NES=NEH}, $(u^{(k)} (t) )_{k \ge 1}$ is the unique mild solution for the Navier-Stokes hierarchy \eqref{eq:NShierarchyIntEqua} with the initial datum $u^{(k)} (0) = u^{\otimes^k}_0,$ such that for every $k \ge 1,$ $ u^{(k)} \in L^{\8} ([0, T], \mathfrak{H}^1_{(k)})$ and satisfies the bound
\be
\sup_{0 \le t \le T} \| u^{(k)} (t) \|_{\mathbb{H}^1} \le C^k_T.
\ee
On the other hand, it follows from \eqref{eq:NSEhierarchyTreeExpanssion} that
for $n >1,$
\be
u (t) = e^{t \triangle} u_0 + \sum^{n-1}_{j=1} \sum_{\mathbb{T} \in \mathfrak{T}_{j,1}} C_{\mathbb{T}, t} u^{\otimes^{j+1}}_0 - \sum_{\mathbb{T} \in \mathfrak{T}_{n,1}} \int^t_0 d s R_{\mathbb{T}, t-s} u^{(n+1)} (s)
\ee
for all $0 \le t \le T.$ As shown in the proof of Theorem \ref{th:Collisionspace-timebound}, every term $C_{\mathbb{T}, t} u^{\otimes^{j+1}}_0$ is well defined in the sense of distribution. It remains to prove that
\be
\lim_{n \to \8} \sum_{\mathbb{T} \in \mathfrak{T}_{n,1}} \int^t_0 d s A_{\mathbb{T}, t-s} u^{(n+1)} (s) =0
\ee
in the sense of distributions for $0 < t < t^*,$ where $t^*$ will be fixed later.

By Theorem \ref{th:Collisionspace-timebounderror}, for any $\phi \in \mathcal{S} (\mathbb{R}^3)$ we have
\be
| \langle \phi, A_{\mathbb{T}, t-s} u^{(n+1)} (s) \rangle | \le M_\phi C^n_0 (t-s)^{\frac{n}{2} -1} \| u^{(n+1)} (s) \|_{\mathbb{H}^1_{(n+1)}}
\ee
for $t, s \in (0, T]$ so that $0 < t -s \le 1,$ where $M_\phi$ is a positive constant depending only on $\phi,$ and $C_0>0$ is an absolute constant. Then
\be\begin{split}
\Big | \Big \langle \phi, \sum_{\mathbb{T} \in \mathfrak{T}_{n,1}} \int^t_0 d s A_{\mathbb{T}, t-s} u^{(n+1)} (s) \Big \rangle \Big |
& \le M_\phi C^n_0 2^{3n+1} \int^t_0 d s (t-s)^{\frac{n}{2} -1} \| u^{(n+1)} (s) \|_{\mathbb{H}^1_{(n+1)}} \\
& \le M_\phi C^n_0 2^{3n+1} C^{n+1}_T t^\frac{n}{2}.
\end{split}\ee
Thus, choosing $t^* = \min \{ 1, \; ( 8 C_0 C_T)^{- 2} \},$ we have
\be
\lim_{n \to \8} \Big \langle \phi, \sum_{\mathbb{T} \in \mathfrak{T}_{n,1}} \int^t_0 d s A_{\mathbb{T}, t-s} u^{(n+1)} (s) \Big \rangle =0
\ee
for all $0< t < t^*.$ This proves \eqref{eq:NSESolutionFormula}.
\end{proof}

\begin{rk}\label{rk:NSESolutionFormula}\rm
\begin{enumerate}[{\rm (1)}]

\item Note that each $C_{\mathbb{T}, t}$ with $\mathbb{T} \in \mathfrak{T}_{n,1}$ is a multi-parameter integral operator with an explicit kernel \eqref{eq:collisionkernel} in momentum space, which describes a kind of processes of two-body interaction of $n+1$ ``particles". Thus, the formula \eqref{eq:NSESolutionFormula} may be regarded as an explicit expression of solution to the homogeneous, incompressible Navier-Stokes equation \eqref{eq:NSE} in $\mathbb{R}^3,$ and should be useful for computing this solution.

\item A natural question is whether $t^*$ can be taken as being $T^*,$ that is, whether the formula \eqref{eq:NSESolutionFormula} holds true for all $0 < t < T^*$ ? For checking this problem, it seems to need new ideas beyond the argument in the above proof.
\end{enumerate}
\end{rk}

Furthermore, we can prove stronger convergence of the series in \eqref{eq:NSESolutionFormula} if the initial data have higher regularity.

\begin{thm}\label{th:NSEsolutionformulaStrongConverg}
Let $\beta < - \frac{3}{2}$ and let $\alpha > \frac{3}{2}.$ Let $u_0 \in \mathbb{H}^\alpha (\mathbb{R}^3)$ with $\bigtriangledown \cdot u_0 =0,$ and let $u$ be the unique weak solution in $C ([0, T^*), \mathbb{H}^\alpha (\mathbb{R}^3))$ for the initial problem of the Navier-Stokes equation \eqref{eq:NSEpressure-free} with the initial datum $u (0) = u_0,$ where $T^*$ is the maximal life-time of $u(t).$ Then there exits $0 <t^* < T^*$ such that \eqref{eq:NSESolutionFormula} holds in $\mathbb{H}^\beta (\mathbb{R}^3)$ for all $0 < t < t^*.$
\end{thm}

The proof is based on {\it a prior} space-time estimates for the interaction operators in multi-parameter Sobolev spaces as follows, which are of their own interest.

\begin{prop}\label{prop:Collisionspace-timeSobolevbound}
Fix $\beta \in \mathbb{R}$ and $\alpha > \max \{ \frac{1}{2}, \beta + \frac{3}{2} \}.$ Then there exists a constant $C>0,$ depending only on $\alpha$ and $\beta,$ such that for every $k \ge 1,$ $n \ge 0,$ and any $\mathbb{T} \in \mathfrak{T}_{n,k},$ one has
\beq\label{eq:Collisionspace-spacetimeSobolevbound}
\big \| C_{\mathbb{T}, t} u^{(n+k)} \big \|_{\mathbb{H}^\beta_{(k)}} \le C^{k +n} t^{\delta n} \big \| u^{(n+k)} \big \|_{\mathbb{H}^\alpha_{(n+k)}}
\eeq
for all $u^{(n+k)} \in \mathfrak{H}^\alpha_{(n+k)}$ and any $t >0,$ where $\delta = \frac{1}{4}\min \{ 1, \alpha - \frac{1}{2} \}.$
\end{prop}

\begin{prop}\label{prop:Erroroperatorspace-timeSobolevbound}
Fix $\beta < - \frac{3}{2}$ and $\alpha > \frac{3}{2}.$ Then there exists a constant $C>0,$ depending only on $\alpha$ and $\beta,$ such that for every $k \ge 1,$ $n \ge 1,$ and any $\mathbb{T} \in \mathfrak{T}_{n,k},$ one has
\beq\label{eq:Erroroperatorspace-spacetimeSobolevbound}
\big \| R_{\mathbb{T}, t} u^{(n+k)} \big \|_{\mathbb{H}^\beta_{(k)}} \le C^{k +n} t^{\delta n} \big \| u^{(n+k)} \big \|_{\mathbb{H}^\alpha_{(n+k)}}
\eeq
for all $u^{(n+k)} \in \mathfrak{H}^\alpha_{(n+k)}$ and any $0< t \le 1,$ where $\delta = \frac{1}{4} \min \{ 1, \alpha - \frac{3}{2} \}.$
\end{prop}

\begin{proof}[\bf Proof of Theorem \ref{th:NSEsolutionformulaStrongConverg}]
We can proceed the same argument in the proof of Theorem \ref{th:NSEsolutionformulaH1}, with the help of Propositions \ref{prop:Collisionspace-timeSobolevbound} and \ref{prop:Erroroperatorspace-timeSobolevbound} in the case $k=1.$ The details are omitted.
\end{proof}

We next turn to the proofs of Propositions \ref{prop:Collisionspace-timeSobolevbound} and \ref{prop:Erroroperatorspace-timeSobolevbound}, which follow the argument in the ones of Theorems \ref{th:Collisionspace-timebound} and \ref{th:Collisionspace-timebounderror}. First of all, we prove \eqref{eq:Collisionspace-spacetimeSobolevbound} on {\it a prior} space-time estimates of collision operators in Sobolev spaces.

\begin{proof}[\bf Proof of Proposition \ref{prop:Collisionspace-timeSobolevbound}]
Fix $\beta \in \mathbb{R}$ and $\alpha > \max \{\frac{1}{2}, \beta + \frac{3}{2}\}.$ Let $ k \ge 1$ and $t >0.$ Let $\mathbb{T} \in \mathfrak{T}_{n,k}$ and suppose $u^{(n+k)} \in \mathfrak{L}^2_{(n+k)}.$ We have, by the definition of $C_{\mathbb{T}, t} u^{(n+k)}$ in \eqref{eq:collisionoerator},
\be\begin{split}
\big \| C_{\mathbb{T}, t} u^{(n+k)} \big \|^2_{\mathbb{H}^\beta_{(k)}} = \int & d \vec{p}_k \langle p_1 \rangle^{2 \beta} \cdots \langle p_k \rangle^{2 \beta} \\
\times \sum_{1 \le i_1, \ldots, i_k \le 3} & \bigg | \frac{1}{(k+n)!} \sum_{\pi_2 \in \Pi_{k+n}} \int d \vec{r}_{n+k} \prod_{e \in R_1 (\mathbb{T}) = L_1 (\mathbb{T})} e^{- t p^2_{\pi_1 (e)}} \delta (p_{\pi_1 (e)} - r_{\pi_2 (e)} )\\
\times \int \prod_{e \in E_2 (\mathbb{T})} d \tau^e d q_e & K_{\mathbb{T}}^{\pi_2} u^{(n+k)}_{i_1, \ldots, i_k} (\vec{r}_{n+k} ) \prod_{e \in R_2 (\mathbb{T})} e^{- t (\gamma^e + \mathrm{i} \tau^e)} \delta (q_e - p_{\pi_1 (e)})\prod_{e \in L_2 (\mathbb{T})} \delta (q_e - r_{\pi_2 (e)})\\
& \times \prod_{e \in E_2 (\mathbb{T})} \frac{1}{\gamma^e - q_e^2 + \mathrm{i} \tau^e} \prod_{v \in V(\mathbb{T})} \delta \big ( \tau^{e^a_v} - \tau^{e^b_v} - \tau^{e^c_v} \big ) \delta \big (q_{e^a_v} - q_{e^b_v} - q_{e^c_v}\big ) \bigg |^2.
\end{split}\ee
For simplicity, we will write $E = E (\mathbb{T}), R_2 = R_2 (\mathbb{T}),$ $L_1 = L_1 (\mathbb{T}),$ etc. Since $u^{(n+k)} (\vec{r}_{n+k})$ is symmetry with respect to the permutation on $\vec{r}_{n+k},$ the integral with respect to $\vec{r}_{n+k}$ of the above equation has the same value for every $\pi_2 \in \Pi_{n+k},$ and hence, instead of averaging over $\pi_2,$ we fix one $\pi_2$ so that $\pi_2 (e) \ge |R_1| +1$ for all $e \in L_2.$ Then, using all the $\delta$-functions and integrating over the variables $\vec{r}_{n+k}$ with $p_j$ being substituted by $q_e$ if $e \in R$ such that $\pi_1 (e) = j$ for $1 \le j \le k,$ one has
\be\begin{split}
\big \| C_{\mathbb{T}, t} u^{(n+k)} \big \|^2_{\mathbb{H}^\alpha_{(k)}} \le e^{- 2 t \sum_{e \in R_2} \gamma^e }& \sum_{1 \le i_1, \ldots, i_k \le 3} \int \prod_{e \in R_1} \langle q_e \rangle^{2 \beta} d q_e \bigg [
\int \prod_{e \in E_2} d q_e d \tau^e | K_{\mathbb{T}}^{\pi_2} u^{(n+k)}_{i_1, \ldots, i_k} (q_e: \; e \in L ) |\\
\times \prod_{e \in R_2} & \langle q_e \rangle^\beta \prod_{e \in E_2} \frac{1}{| \gamma^e - q_e^2 + \mathrm{i} \tau^e |} \prod_{v \in V} \delta \big ( \tau^{e^a_v} - \tau^{e^b_v} - \tau^{e^c_v} \big ) \delta \big (q_{e^a_v} - q_{e^b_v} - q_{e^c_v} \big )
\bigg ]^2
\end{split}\ee
where we have used the permutation symmetry of $u^{(n+k)},$ namely, $u^{(n+k)}$ depends only on the set of the variables $q_e$ associated with the leaves of $\mathbb{T},$ but not on the order of those variables.

We choose $\gamma^e = - \frac{1}{t}$ for all $e \in L_2.$ This yields that $\gamma^e \le - \frac{1}{t}$ for every $e \in E_2,$ and $\sum_{e \in R_2} \gamma^e = - (n + |R_2|)/t.$ Moreover, by \eqref{eq:KoperatorEst} we have
\be
| K_{\mathbb{T}}^{\pi_2} u^{(n+k)}_{i_1, \ldots, i_k} (q_e) | \le 6^n \prod_{v \in V}  | q_{e^b_v} + q_{e^c_v} | \sum_{1 \le j_{|R_1| +1}, \ldots, j_{n+k} \le 3} \big | u^{(n+k)}_{i_1, \ldots, i_{|R_1|}, j_{|R_1|+1}, \ldots, j_{n+k}} (q_e) \big |.
\ee
Then by the Cauchy-Schwarz inequality,
\be\begin{split}
\big \| C_{\mathbb{T}, t} u^{(n+k)} \big \|^2_{\mathbb{H}^\alpha_{(k)}} \le C^{k+n} \sum_{1 \le i_1, \ldots, i_{|R_1|} \le 3} & \int \prod_{e \in R_1} \langle q_e \rangle^{2 \beta} d q_e \bigg [ \int \prod_{e \in E_2 } d q_e d \tau^e \prod_{e \in R_2} \langle q_e \rangle^\beta \prod_{e \in E_2} \frac{1}{| (\frac{1}{t} + q_e^2) \mathrm{i} + \tau^e |}\\
\times \prod_{v \in V} \delta \big ( \tau^{e^a_v} - \tau^{e^b_v} - \tau^{e^c_v} \big ) \delta \big (q_{e^a_v} - & q_{e^b_v} - q_{e^c_v} \big )
| q_{e^b_v} + q_{e^c_v} | \sum_{ 1 \le i_{|R_1|+1}, \ldots, i_{n+k} \le 3} \big | u^{(n+k)}_{i_1, \ldots, i_{n+k}} (q_e: \; e \in L) \big | \bigg ]^2\\
\le C^{k +n} \sum_{1 \le i_1, \ldots, i_{n+k} \le 3} & \int \prod_{e \in L} \langle q_e \rangle^{2 \alpha} d q_e \big | u^{(n+k)}_{i_1, \ldots, i_{n+k}} (q_e: \; e \in L) \big |^2\\
\times \int \prod_{e \in L_2} d q_e \bigg [ & \prod_{e \in L_2} \frac{1}{\langle q_e \rangle^\alpha} \int \prod_{e \in E_2 \backslash L_2} d q_e \prod_{e \in E_2 } d \tau^e \prod_{e \in R_2} \langle q_e \rangle^\beta \prod_{e \in E_2} \frac{1}{| (\frac{1}{t} + q_e^2) \mathrm{i} + \tau^e |} \\
& \times \prod_{v \in V} \delta \big ( \tau^{e^a_v} - \tau^{e^b_v} - \tau^{e^c_v} \big )\delta \big (q_{e^a_v} - q_{e^b_v} - q_{e^c_v} \big )
| q_{e^b_v} + q_{e^c_v} | \bigg ]^2\\
=: C^{k+n} \big \| u^{(k+n))} \big \|^2_{\mathbb{H}^\alpha_{(k+n)}} & \times I_t ,
\end{split}\ee
where we have used the fact $\alpha > \beta$ and $C>0$ is an absolute constant, and
\be\label{eq:Collisionspace-timeintegralSobolevbound01}
\begin{split}
I_t = \int \prod_{e \in L_2} & \frac{d q_e}{(1+|q_e|^2)^\alpha} \bigg [ \int \prod_{e \in E_2 \backslash L_2} d q_e \prod_{e \in E_2} d \tau^e \prod_{e \in R_2} \langle q_e \rangle^\beta \prod_{e \in E_2} \frac{1}{| (\frac{1}{t} + q_e^2) \mathrm{i} + \tau^e |} \\
& \times \prod_{v \in V} \delta \big ( \tau^{e^a_v} - \tau^{e^b_v} - \tau^{e^c_v} \big ) \delta \big (q_{e^a_v} - q_{e^b_v} - q_{e^c_v}\big ) | q_{e^b_v} + q_{e^c_v} | \bigg ]^2.
\end{split}
\ee

Making variable substitution as $q_e \longrightarrow t^{- \frac{1}{2}} q_e$ and $\tau^e \longrightarrow t^{-1} \tau^e$ for all $e \in E_2,$ we have
\be\begin{split}
\frac{d \tau^e}{| (1/t + q_e^2) \mathrm{i} + \tau^e |} & \longrightarrow \frac{d \tau^e}{| (1 + q_e^2) \mathrm{i} + \tau^e |},\\
\delta \big ( \tau^{e^a_v} - \tau^{e^b_v} - \tau^{e^c_v} \big ) & \longrightarrow t \delta \big ( \tau^{e^a_v} - \tau^{e^b_v} - \tau^{e^c_v} \big ),\\
\delta \big (q_{e^a_v} - q_{e^b_v} - q_{e^c_v}\big ) & \longrightarrow t^{\frac{3}{2}} \delta \big (q_{e^a_v} - q_{e^b_v} - q_{e^c_v}\big ),
\end{split}\ee
and hence,
\be
\begin{split}
I_t = & t^{(\alpha - \frac{1}{2})n + (\alpha-\beta -\frac{3}{2}) |R_2|} \int \prod_{e \in L_2} \frac{d q_e}{(t +|q_e|^2)^\alpha} \bigg [ \int \prod_{e \in E_2 \backslash L_2} d q_e \prod_{e \in E_2} d \tau^e \prod_{e \in R_2} ( t + |q_e|^2 )^\frac{\beta}{2}\\
& \times \prod_{e \in E_2} \frac{1}{| (1 + q_e^2) \mathrm{i} + \tau^e |} \prod_{v \in V} \delta \big ( \tau^{e^a_v} - \tau^{e^b_v} - \tau^{e^c_v} \big ) \delta \big (q_{e^a_v} - q_{e^b_v} - q_{e^c_v}\big ) | q_{e^b_v} + q_{e^c_v} | \bigg ]^2.
\end{split}
\ee
By \eqref{eq:Intovertauvariables} in the proof of Theorem \ref{th:Collisionspace-timebound}, for a small enough $0 < \varepsilon < 1$ which will be specified later, we have
\be\begin{split}
\int & \prod_{e \in E_2} \frac{d \tau^e}{| (1 + q_e^2) \mathrm{i} + \tau^e |} \prod_{v \in V} \delta \big ( \tau^{e^a_v} - \tau^{e^b_v} - \tau^{e^c_v} \big )\\
& \lesssim_\varepsilon \prod_{v:\; e^a_v \notin R_2} \frac{1}{(1 + q^2_{e^b_v} + q^2_{e^c_v})^{1-\varepsilon}}
\prod_{v:\; e^a_v \in R_2} \frac{1}{(1 + q^2_{e^a_v} + q_{e^b_v}^2 + q_{e^c_v}^2 )^{1-\varepsilon}}.
\end{split}\ee
Thus, after all integrations over $\tau$-variables have been done, we have
\be\label{eq:CollisionitegralboundM}\begin{split}
I_t \lesssim_\varepsilon t^{(\alpha - \frac{1}{2})n + (\alpha - \beta - \frac{3}{2}) |R_2|} & \int \prod_{e \in E_2} d q_e \prod_{v \in V} \delta \big (q_{e^a_v} - q_{e^b_v} - q_{e^c_v}\big ) \prod_{e \in R_2} ( t + |q_e|^2 )^\beta \prod_{e \in L_2} \frac{1}{(t +|q_e|^2)^\alpha}\\
& \times \prod_{v:\; e^a_v \notin R_2} \frac{| q_{e^b_v} + q_{e^c_v} |^2}{(1 + q^2_{e^b_v} + q^2_{e^c_v})^{2 (1-\varepsilon)}}
\prod_{v:\; e^a_v \in R_2} \frac{| q_{e^b_v} + q_{e^c_v} |^2}{(1 + q^2_{e^a_v} + q_{e^b_v}^2 + q_{e^c_v}^2 )^{2 (1-\varepsilon)}}\\
= t^{2 n \varepsilon } & \int \prod_{e \in E_2} d q_e \prod_{v \in V} \delta \big (q_{e^a_v} - q_{e^b_v} - q_{e^c_v}\big ) \prod_{e \in R_2} (1 + |q_e|^2 )^\beta \prod_{e \in L_2} \frac{1}{(1 +|q_e|^2)^\alpha}\\
& \times \prod_{v:\; e^a_v \notin R_2} \frac{| q_{e^b_v} + q_{e^c_v} |^2}{(\frac{1}{t} + q^2_{e^b_v} + q^2_{e^c_v})^{2 (1-\varepsilon)}}
\prod_{v:\; e^a_v \in R_2} \frac{| q_{e^b_v} + q_{e^c_v} |^2}{(\frac{1}{t} + q^2_{e^a_v} + q_{e^b_v}^2 + q_{e^c_v}^2 )^{2 (1-\varepsilon)}}
\end{split}
\ee
where we have made variable substitution as $q_e \longrightarrow t^\frac{1}{2} q_e$ for all $e \in E_2$ in the last expression.

It remains to estimate the momenta integrations
\be\begin{split}
\Xi = & \int \prod_{e \in E_2} d q_e \prod_{v \in V} \delta \big (q_{e^a_v} - q_{e^b_v} - q_{e^c_v}\big ) \prod_{e \in R_2} ( 1 + |q_e|^2 )^\beta \prod_{e \in L_2} \frac{1}{(1 + |q_e |^2)^\alpha}\\
& \times \prod_{v:\; e^a_v \notin R_2} \frac{| q_{e^b_v} + q_{e^c_v} |^2}{(\frac{1}{t} + q^2_{e^b_v} + q^2_{e^c_v})^{2 (1-\varepsilon)}} \prod_{v:\; e^a_v \in R_2} \frac{| q_{e^b_v} + q_{e^c_v} |^2}{(\frac{1}{t} + q^2_{e^a_v} + q_{e^b_v}^2 + q_{e^c_v}^2 )^{2 (1-\varepsilon)}}.
\end{split}\ee
Since the delta functions relate variables within the same trees, we only need to consider the integration over all $q$-variables associated with a tree $\mathrm{T},$ that is
\beq\label{eq:IntegralonetreeSobolevSpacebound}\begin{split}
\Xi (\mathrm{T}) = \int \prod_{e \in E_2 (\mathrm{T})} d q_e \prod_{v \in V (\mathrm{T})} \delta \big (q_{e^a_v} - q_{e^b_v} - q_{e^c_v}\big ) & \frac{(1 + |q_{e^a_v = R(\mathrm{T})} |^2 )^\beta | q_{e^b_v} + q_{e^c_v} |^2}{(\frac{1}{t} + q^2_{e^a_v = R(\mathrm{T})} + q_{e^b_v}^2 + q_{e^c_v}^2 )^{2 (1-\varepsilon)}}\\
\times \prod_{e \in L_2 (\mathrm{T})} \frac{1}{(1 + |q_e |^2)^\alpha} & \prod_{v \in V (\mathrm{T}):\; e^a_v \not= R(\mathrm{T})} \frac{| q_{e^b_v} + q_{e^c_v} |^2}{(\frac{1}{t} + q^2_{e^b_v} + q^2_{e^c_v})^{2 (1-\varepsilon)}}.
\end{split}\eeq
Again, we begin the integration with the $q$-variables of the leaves and proceed toward the root, and a vertex $v$ will be integrated only when all vertices $v'$ with $v \prec v'$ have already been integrated out.

Indeed, for a maximal vertex $v$ whose two daughter edges must be leaves, if $e^a_v \not= R$ then the associated integration is
\be\begin{split}
\Xi_v = & \int d q_{e^b_v} d q_{e^c_v} \delta \big (q_{e^a_v} - q_{e^b_v} - q_{e^c_v}\big ) \frac{1}{(1 + |q_{e^b_v} |^2)^\alpha} \frac{1}{(1 + |q_{e^c_v} |^2)^\alpha} \frac{| q_{e^b_v} + q_{e^c_v} |^2}{(\frac{1}{t} + q^2_{e^b_v} + q^2_{e^c_v})^{2 (1-\varepsilon)}}\\
= & \int d q_{e^b_v} \frac{1}{(1 + |q_{e^b_v} |^2)^\alpha} \frac{1}{(1 + |q_{e^a_v} - q_{e^b_v} |^2)^\alpha} \frac{q_{e^a_v}^2}{(\frac{1}{t} + q^2_{e^b_v} + (q_{e^a_v} - q_{e^b_v})^2 )^{2 (1-\varepsilon)}}.
\end{split}\ee
As shown in \eqref{eq:Integralonevertexbound}, taking $\varepsilon = \frac{1}{4} \min \{ 1, \alpha - \frac{1}{2} \}$ (noticing that $\alpha > \frac{1}{2}$), we have a constant $C_\alpha >0$ depending only on $\alpha$ such that
\be
\Xi_v \le C_\alpha \frac{1}{(1 + |q_{e^a_v} |^2)^\alpha}.
\ee
Subsequently, every vertex $v$ for which all vertices $v'$ with $v \prec v'$ have already been integrated out is associated with the integration of the form $\Xi_v$ as above when $e^a_v \not= R.$ Thus, we obtain
\be\begin{split}
\Xi (\mathrm{T}) \le C_\alpha^{|V(\mathrm{T})|-1} & \int d q_{e^a_v = R(\mathrm{T})} d q_{e^b_v} d q_{e^c_v} \delta \big (q_{e^a_v} - q_{e^b_v} - q_{e^c_v}\big ) (1 + |q_{e^a_v} |^2 )^\beta\\
& \times \frac{1}{(1 + |q_{e^b_v} |^2)^\alpha} \frac{1}{(1 + |q_{e^c_v} |^2)^\alpha} \frac{| q_{e^b_v} + q_{e^c_v} |^2}{(\frac{1}{t} + q^2_{e^a_v = R(\mathrm{T})} + q_{e^b_v}^2 + q_{e^c_v}^2 )^{2 (1-\varepsilon)}}\\
\le C_\alpha^{|V(\mathrm{T})|-1} & \int d q_{e^a_v = R(\mathrm{T})} d q_{e^b_v} (1 + |q_{e^a_v} |^2 )^\beta \frac{1}{(1 + |q_{e^b_v} |^2)^\alpha} \frac{1}{(1 + |q_{e^a_v} - q_{e^b_v} |^2)^\alpha}\\
& \times \frac{| q_{e^a_v}|^2}{( q^2_{e^a_v = R(\mathrm{T})} + q_{e^b_v}^2 + (q_{e^a_v} - q_{e^b_v})^2 )^{2 (1-\varepsilon)}}\\
\le C_\alpha^{|V(\mathrm{T})|} \int & d q_{e^a_v} (1 + |q_{e^a_v} |^2 )^{\beta - \alpha}
\le C_{\beta, \alpha}^{|V(\mathrm{T})|}
\end{split}\ee
provided $\alpha > \max \{ \frac{1}{2}, \beta + \frac{3}{2} \}.$ Therefore, we conclude that
\be
\Xi = \prod_{\mathrm{T} \in \mathbb{T}} \Xi (\mathrm{T}) \le C_{\beta, \alpha}^n.
\ee
This completes the proof.
\end{proof}

The proof of \eqref{eq:Erroroperatorspace-spacetimeSobolevbound} is essentially the same as in that of Proposition \ref{prop:Collisionspace-timeSobolevbound}, but for the sake of completeness, we include the details.

\begin{proof}[\bf Proof of Proposition \ref{prop:Erroroperatorspace-timeSobolevbound}]
Fix $\beta < - \frac{3}{2}$ and $\alpha > \frac{3}{2}.$ Let $k \ge 1$ and $n \ge 1.$ For a fixed $\mathbb{T} \in \mathfrak{T}_{n,k},$ by \eqref{eq:Errortermoperator} we have
\be\begin{split}
\big \| R_{\mathbb{T}, t} u^{(n+k)} \big \|^2_{\mathbb{H}^\beta_{(k)}}& = \int d \vec{p}_k \langle p_1 \rangle^{2 \beta} \cdots \langle p_k \rangle^{2 \beta} \\
\times \sum_{1 \le i_1, \ldots, i_k \le 3} & \Big | \frac{1}{(k+n)!} \sum_{\pi_2 \in \Pi_{k+n}} \int d \vec{r}_{n+k}
\sum_{\bar{v} \in M(\mathbb{T})} \prod_{e \in R_1 (\mathbb{T}) = L_1 (\mathbb{T})} e^{- t p^2_{\pi_1 (e)}} \delta (p_{\pi_1 (e)} - r_{\pi_2 (e)} )\\
\times \int \prod_{e \in E_2 (\mathbb{T})} & d q_e \prod_{ \substack{ e \in E_2 (\mathbb{T})\\ e \notin D_{\bar{v}} } } \frac{d \tau^e}{\gamma^e - q_e^2 + \mathrm{i} \tau^e} K_{\mathbb{T}}^{\pi_2} u^{(n+k)}_{i_1, \ldots, i_k} (\vec{r}_{n+k} ) \prod_{ e \in R_2 (\mathbb{T})} e^{- t(\gamma^e + \mathrm{i} \tau^e)} \delta (q_e - p_{\pi_1 (e)})\\
\times & \prod_{e \in L_2 (\mathbb{T})} \delta (q_e - r_{\pi_2 (e)}) \prod_{ \substack{ v \in V(\mathbb{T})\\ v \not= \bar{v} } } \delta \big ( \tau^{e^a_v} - \tau^{e^b_v} - \tau^{e^c_v} \big ) \prod_{ v \in V(\mathbb{T}) } \delta \big (q_{e^a_v} - q_{e^b_v} - q_{e^c_v} \big ) \Big |^2.
\end{split}\ee
Recall that $M (\mathbb{T})$ is the set of maximal elements in $V ( \mathbb{T})$ and $D_v$ the set of daughter-edges for a vertex $v \in M ( \mathbb{T}).$ In the integration of the right hand side of the above equation, if there exists a $\bar{v} \in M (\mathbb{T})$ such that $e^a_{\bar{v}} \in R (\mathbb{T}),$ then there is only one denominator containing $\tau^e$ in the integral for this tree and the associated $\tau^e$-integral would not be absolutely convergent. In this case, using \eqref{eq:PropagatorIdent1}, we perform the integration over the $\tau^{e^a_{\bar{v}}}$ and obtain a factor $e^{- t q^2_{e^a_{\bar{v}}}}.$ Recall that $E_2 (\mathbb{T}, \bar{v}) = E_2 (\mathbb{T}) \backslash \{ e^a_{\bar{v}}\}$ if $e^a_{\bar{v}} \in R (\mathbb{T}),$ and otherwise $E_2 (\mathbb{T}, \bar{v}) = E_2 (\mathbb{T}).$

As done in the proof of Theorem \ref{th:Collisionspace-timebound}, instead of averaging over all $\pi_2 \in \Pi_{n+k},$ we only need to compute the integral on the right-hand side of the above equation with a permutation $\pi_2$ so that $\pi_2 (e) \ge |R_1| +1$ for all $e \in L_2.$ After performing the integration over $\tau^e$ associated to $e \notin E_2 (\mathbb{T}, \bar{v}),$ we take the absolute value of the above integrand. Subsequently, using all the $\delta$-functions and integrating over the variables $\vec{r}_{n+k}$ with $p_j$ being substituted by $q_e$ if $e \in R$ such that $\pi_1 (e) = j$ for $1 \le j \le k,$ we obtain that
\be\begin{split}
\big \| R_{\mathbb{T}, t} u^{(n+k)} \big \|^2_{\mathbb{H}^\beta_{(k)}}= e^{- 2 t \sum_{e \in R_2} \gamma^e } & \sum_{1 \le i_1, \ldots, i_k \le 3} \int \prod_{e \in R_1} \langle q_e \rangle^{2 \beta} d q_e \Big [ \sum_{\bar{v} \in M(\mathbb{T})} \int \prod_{e \in E_2 (\mathbb{T})} d q_e \\
\times & \prod_{ \substack{ e \in E_2 (\mathbb{T}, \bar{v})\\ e \notin D_{\bar{v}} } } \frac{d \tau^e}{|\gamma^e - q_e^2 + \mathrm{i} \tau^e |} | K_{\mathbb{T}}^{\pi_2} u^{(n+k)}_{i_1, \ldots, i_k} ( q_e: \; e \in L )| \prod_{ e \in R_2 (\mathbb{T})} \langle q_e \rangle^\beta\\
& \times \prod_{ \substack{ v \in V(\mathbb{T})\\ v \not= \bar{v} } } \delta \big ( \tau^{e^a_v} - \tau^{e^b_v} - \tau^{e^c_v} \big ) \prod_{ v \in V(\mathbb{T}) } \delta \big (q_{e^a_v} - q_{e^b_v} - q_{e^c_v} \big ) \Big ]^2.
\end{split}\ee
Choosing $\gamma^e = - \frac{1}{t}$ for all $e \in L_2$ and by \eqref{eq:KoperatorEst}, we have (cf. \eqref{eq:Collisionspace-timebound01})
\be
\begin{split}
\big \| R_{\mathbb{T}, t} u^{(n+k)} \big \|^2_{\mathbb{H}^\beta_{(k)}} \le
C^{k+n} \sum_{1 \le i_1, \ldots, i_{|R_1|} \le 3} & \int \prod_{e \in R_1} \langle q_e \rangle^{2 \beta} d q_e \bigg [ \sum_{\bar{v} \in M(\mathbb{T})} \int \prod_{e \in E} d q_e\\
\times \prod_{ \substack{ e \in E_2 (\mathbb{T}, \bar{v})\\ e \notin D_{\bar{v}} }} \frac{d \tau^e}{| (\frac{1}{t} + q_e^2) \mathrm{i} + \tau^e |} & \prod_{e \in R_2} \langle q_e \rangle^\beta \sum_{ 1 \le i_{|R_1| +1}, \ldots, i_{n+k} \le 3} | u^{(n+k)}_{i_1, \ldots, i_{n+k}} (q_e: \; e \in L) |\\
\times \prod_{ \substack{ v \in V(\mathbb{T})\\ v \not= \bar{v} } } \delta ( & \tau^{e^a_v} - \tau^{e^b_v} - \tau^{e^c_v} ) \prod_{v \in V} \delta ( q_{e^a_v} - q_{e^b_v} - q_{e^c_v} ) | q_{e^b_v} + q_{e^c_v} | \bigg ]^2\\
\end{split}
\ee
Then by the Cauchy-Schwarz inequality and the fact that $\alpha > \beta,$ we have
\be
\begin{split}
\big \| R_{\mathbb{T}, t} u^{(n+k)} \big \|^2_{\mathbb{H}^\beta_{(k)}}
\le C^{n+k} \big \| u^{(k+n))} \big \|^2_{\mathbb{H}^\alpha_{(k+n)}} \int \prod_{e \in L_2} & d q_e \bigg [ \sum_{\bar{v} \in M(\mathbb{T})} \int \prod_{e \in E_2 \backslash L_2} d q_e \prod_{ \substack{ e \in E_2 (\mathbb{T}, \bar{v})\\ e \notin D_{\bar{v}} }} \frac{d \tau^e}{| (\frac{1}{t} + q_e^2) \mathrm{i} + \tau^e |}\\
\times \prod_{e \in L_2} \frac{1}{\langle q_e \rangle^\alpha} \prod_{e \in R_2} \langle q_e \rangle^\beta \prod_{ \substack{ v \in V(\mathbb{T})\\ v \not= \bar{v} } } & \delta \big ( \tau^{e^a_v} - \tau^{e^b_v} - \tau^{e^c_v} \big ) \prod_{v \in V} \delta \big (q_{e^a_v} - q_{e^b_v} - q_{e^c_v}\big ) | q_{e^b_v} + q_{e^c_v} | \bigg ]^2\\
\le C^{n+k} \big \| u^{(k+n))} \big \|^2_{\mathbb{H}^\alpha_{(k+n)}} \sum_{\bar{v} \in M(\mathbb{T})} & \int \prod_{e \in L_2} d q_e \bigg [ \int \prod_{e \in E_2 \backslash L_2} d q_e \prod_{ \substack{ e \in E_2 (\mathbb{T}, \bar{v})\\ e \notin D_{\bar{v}} }} \frac{d \tau^e}{| (\frac{1}{t} + q_e^2) \mathrm{i} + \tau^e |}\\
\times \prod_{e \in L_2} \frac{1}{\langle q_e \rangle^\alpha} \prod_{e \in R_2} \langle q_e \rangle^\beta \prod_{ \substack{ v \in V(\mathbb{T})\\ v \not= \bar{v} } } & \delta \big ( \tau^{e^a_v} - \tau^{e^b_v} - \tau^{e^c_v} \big ) \prod_{v \in V} \delta \big (q_{e^a_v} - q_{e^b_v} - q_{e^c_v}\big ) | q_{e^b_v} + q_{e^c_v} | \bigg ]^2
\end{split}
\ee
where $C>0$ is an absolute constant which may vary in different lines, in the second inequality we have used the Minkowski inequality and that $|M(\mathbb{T})| \le n.$ Putting
\be
\begin{split}
 J(t) = \int \prod_{e \in L_2} & d q_e \bigg [ \int \prod_{e \in E_2 \backslash L_2} d q_e \prod_{ \substack{ e \in E_2 (\mathbb{T}, \bar{v})\\ e \notin D_{\bar{v}} }} \frac{d \tau^e}{| (\frac{1}{t} + q_e^2) \mathrm{i} + \tau^e |} \prod_{e \in L_2} \frac{1}{\langle q_e \rangle^\alpha}\\
\times \prod_{e \in R_2} & \langle q_e \rangle^\beta \prod_{ \substack{ v \in V(\mathbb{T})\\ v \not= \bar{v} } } \delta \big ( \tau^{e^a_v} - \tau^{e^b_v} - \tau^{e^c_v} \big ) \prod_{v \in V} \delta \big (q_{e^a_v} - q_{e^b_v} - q_{e^c_v}\big ) | q_{e^b_v} + q_{e^c_v} | \bigg ]^2
\end{split}
\ee
we need to estimate $J(t).$

To this end, making variable substitution as $q_e \longrightarrow t^{- \frac{1}{2}} q_e$ for all $e \in E_2 \backslash L_2,$ and $\tau^e \longrightarrow t^{-1} \tau^e$ for all $e \in E_2 (\mathbb{T}, \bar{v}) \backslash D_{\bar{v}},$ we have
\be\begin{split}
\frac{d \tau^e}{| (1/t + q_e^2) \mathrm{i} + \tau^e |} & \longrightarrow \frac{d \tau^e}{| (1 + q_e^2) \mathrm{i} + \tau^e |},\\
\delta \big ( \tau^{e^a_v} - \tau^{e^b_v} - \tau^{e^c_v} \big ) & \longrightarrow t \delta \big ( \tau^{e^a_v} - \tau^{e^b_v} - \tau^{e^c_v} \big ),\\
\delta \big (q_{e^a_v} - q_{e^b_v} - q_{e^c_v}\big ) & \longrightarrow t^{\frac{3}{2}} \delta \big (q_{e^a_v} - q_{e^b_v} - q_{e^c_v}\big ).
\end{split}\ee
Then
\be
\begin{split}
J(t) = t^{(\alpha - \frac{1}{2})n + (\alpha - \beta - \frac{3}{2}) |R_2| - 2} \int \prod_{e \in L_2} & d q_e \bigg [ \int \prod_{e \in E_2 \backslash L_2} d q_e \prod_{ \substack{ e \in E_2 (\mathbb{T}, \bar{v})\\ e \notin D_{\bar{v}} }} \frac{d \tau^e}{| (1 + q_e^2) \mathrm{i} + \tau^e |} \prod_{e \in L_2} \frac{1}{( t + |q_e|^2 )^\frac{\alpha}{2}}\\
\times \prod_{e \in R_2} ( t + |q_e|^2 )^\frac{\beta}{2} & \prod_{ \substack{ v \in V(\mathbb{T})\\ v \not= \bar{v} } } \delta \big ( \tau^{e^a_v} - \tau^{e^b_v} - \tau^{e^c_v} \big ) \prod_{v \in V} \delta \big (q_{e^a_v} - q_{e^b_v} - q_{e^c_v}\big ) | q_{e^b_v} + q_{e^c_v} | \bigg ]^2.
\end{split}
\ee
As shown in \eqref{eq:Intovertauvariables}, we have
\be\begin{split}
\int & \prod_{ \substack{ e \in E_2 (\mathbb{T}, \bar{v})\\ e \notin D_{\bar{v}} }} \frac{d \tau^e}{| (1 + q_e^2) \mathrm{i} + \tau^e |} \prod_{ \substack{ v \in V(\mathbb{T})\\ v \not= \bar{v} } } \delta ( \tau^{e^a_v} - \tau^{e^b_v} - \tau^{e^c_v} )\\
& \lesssim_\varepsilon \prod_{\substack{ v:\; e^a_v \notin R_2\\ v \not= \bar{v}}} \frac{1}{(1 + q^2_{e^b_v} + q^2_{e^c_v})^{1-\varepsilon}}
\prod_{\substack{ v :\; e^a_v \in R_2\\ v \not= \bar{v} }}\frac{1}{(1 + q^2_{e^a_v} + q_{e^b_v}^2 + q_{e^c_v}^2 )^{1-\varepsilon}}
\end{split}
\ee
where $0< \varepsilon <1$ will be fixed later. Thus, we have
\be\label{eq:CollisionitegralboundM}\begin{split}
J(t) \lesssim_\varepsilon & t^{(\alpha - \frac{1}{2})n + (\alpha - \beta - \frac{3}{2}) |R_2| - 2} \int \prod_{e \in E_2} d q_e \prod_{v \in V} \delta \big (q_{e^a_v} - q_{e^b_v} - q_{e^c_v}\big ) \prod_{e \in L_2} \frac{1}{( t + |q_e|^2 )^\alpha} \prod_{e \in R_2} ( t + |q_e|^2 )^\beta\\
& \times | q_{e^b_{\bar{v}}} + q_{e^c_{\bar{v}}} |^2 \prod_{\substack{v:\; e^a_v \notin R_2\\ v \not= \bar{v}}} \frac{| q_{e^b_v} + q_{e^c_v} |^2}{(1 + q^2_{e^b_v} + q^2_{e^c_v})^{2 (1-\varepsilon) }}
\prod_{ \substack{ v:\; e^a_v \in R_2\\ v \not= \bar{v}} } \frac{| q_{e^b_v} + q_{e^c_v} |^2}{(1 + q^2_{e^a_v} + q_{e^b_v}^2 + q_{e^c_v}^2 )^{2(1-\varepsilon) }}\\
\le & t^{2 (n-1) \varepsilon} \int \prod_{e \in E_2} d q_e \prod_{v \in V} \delta \big (q_{e^a_v} - q_{e^b_v} - q_{e^c_v}\big ) \prod_{e \in L_2} \frac{1}{( 1 + |q_e|^2 )^\alpha} \prod_{e \in R_2} ( 1 + |q_e|^2 )^\beta\\
& \times | q_{e^b_{\bar{v}}} + q_{e^c_{\bar{v}}} |^2 \prod_{\substack{v:\; e^a_v \notin R_2\\ v \not= \bar{v}}} \frac{| q_{e^b_v} + q_{e^c_v} |^2}{(1 + q^2_{e^b_v} + q^2_{e^c_v})^{2 (1-\varepsilon) }}
\prod_{ \substack{ v:\; e^a_v \in R_2\\ v \not= \bar{v}} } \frac{| q_{e^b_v} + q_{e^c_v} |^2}{(1 + q^2_{e^a_v} + q_{e^b_v}^2 + q_{e^c_v}^2 )^{2(1-\varepsilon) }}
\end{split}
\ee
for $0< t \le 1,$ where we have made variable substitution as $q_e \longrightarrow t^\frac{1}{2} q_e$ for all $e \in E_2$ in the last expression.

Next, we estimate the integration on the right hand side of the above inequality. Fix $\bar{v} \in M(\mathbb{T})$ and denote the integration by $\Xi (\mathbb{T}, \bar{v}),$ i.e.,
\be\begin{split}
\Xi (\mathbb{T}, \bar{v}) = & \int \prod_{e \in E_2} d q_e \prod_{v \in V} \delta \big (q_{e^a_v} - q_{e^b_v} - q_{e^c_v}\big ) \prod_{e \in L_2} \frac{1}{( 1 + |q_e|^2 )^\alpha} \prod_{e \in R_2} ( 1 + |q_e|^2 )^\beta\\
& \times | q_{e^b_{\bar{v}}} + q_{e^c_{\bar{v}}} |^2 \prod_{\substack{v:\; e^a_v \notin R_2\\ v \not= \bar{v}}} \frac{| q_{e^b_v} + q_{e^c_v} |^2}{(1 + q^2_{e^b_v} + q^2_{e^c_v})^{2 (1-\varepsilon) }}
\prod_{ \substack{ v:\; e^a_v \in R_2\\ v \not= \bar{v}} } \frac{| q_{e^b_v} + q_{e^c_v} |^2}{(1 + q^2_{e^a_v} + q_{e^b_v}^2 + q_{e^c_v}^2 )^{2(1-\varepsilon) }}
\end{split}\ee
Since the delta functions relate variables within the same trees, we may consider separately the integrations over all $q$-variables associated with each tree. Note that the integration associated with a tree $\mathrm{T}$ not including $\bar{v}$ is
\be\begin{split}
\Xi (\mathrm{T}) = \int \prod_{e \in E_2 (\mathrm{T})} d q_e \prod_{v \in V (\mathrm{T})} \delta \big (q_{e^a_v} - q_{e^b_v} - q_{e^c_v}\big ) & \frac{(1 + |q_{e^a_v = R(\mathrm{T})} |^2 )^\beta | q_{e^b_v} + q_{e^c_v} |^2}{(1 + q^2_{e^a_v = R(\mathrm{T})} + q_{e^b_v}^2 + q_{e^c_v}^2 )^{2 (1-\varepsilon)}}\\
\times \prod_{e \in L_2 (\mathrm{T})} \frac{1}{(1 + |q_e |^2)^\alpha} & \prod_{v \in V (\mathrm{T}):\; e^a_v \not= R(\mathrm{T})} \frac{| q_{e^b_v} + q_{e^c_v} |^2}{(1 + q^2_{e^b_v} + q^2_{e^c_v})^{2 (1-\varepsilon)}}.
\end{split}\ee
Again, we begin the integration with the $q$-variables of the leaves and proceed toward the root, and a vertex $v$ will be integrated only when all vertices $v'$ with $v \prec v'$ have already been integrated out.

Indeed, for a maximal vertex $v$ whose two daughter edges must be leaves, if $e^a_v \not= R$ then the associated integration is
\be\begin{split}
\Xi_v = & \int d q_{e^b_v} d q_{e^c_v} \delta \big (q_{e^a_v} - q_{e^b_v} - q_{e^c_v}\big ) \frac{1}{(1 + |q_{e^b_v} |^2)^\alpha} \frac{1}{(1 + |q_{e^c_v} |^2)^\alpha} \frac{| q_{e^b_v} + q_{e^c_v} |^2}{(1 + q^2_{e^b_v} + q^2_{e^c_v})^{2 (1-\varepsilon)}}\\
= & \int d q_{e^b_v} \frac{1}{(1 + |q_{e^b_v} |^2)^\alpha} \frac{1}{(1 + |q_{e^a_v} - q_{e^b_v} |^2)^\alpha} \frac{q_{e^a_v}^2}{(1 + q^2_{e^b_v} + (q_{e^a_v} - q_{e^b_v})^2 )^{2 (1-\varepsilon)}}.
\end{split}\ee
As shown in \eqref{eq:Integralonevertexbound}, taking $\varepsilon = \frac{1}{4} \min \{ 1, \alpha - \frac{3}{2} \}$ (noticing that $\alpha > \frac{3}{2}$), we have a constant $C_\alpha >0$ depending only on $\alpha$ such that
\be
\Xi_v \le C_\alpha \frac{1}{(1 + |q_{e^a_v} |^2)^\alpha}.
\ee
Subsequently, every vertex $v$ for which all vertices $v'$ with $v \prec v'$ have already been integrated out is associated with the integration of the form $\Xi_v$ as above when $e^a_v \not= R.$ Thus, we obtain
\be\begin{split}
\Xi (\mathrm{T}) \le C_\alpha^{|V(\mathrm{T})|-1} & \int d q_{e^a_v = R(\mathrm{T})} d q_{e^b_v} d q_{e^c_v} \delta \big (q_{e^a_v} - q_{e^b_v} - q_{e^c_v}\big ) (1 + |q_{e^a_v} |^2 )^\beta\\
& \times \frac{1}{(1 + |q_{e^b_v} |^2)^\alpha} \frac{1}{(1 + |q_{e^c_v} |^2)^\alpha} \frac{| q_{e^b_v} + q_{e^c_v} |^2}{(1 + q^2_{e^a_v = R(\mathrm{T})} + q_{e^b_v}^2 + q_{e^c_v}^2 )^{2 (1-\varepsilon)}}\\
\le C_\alpha^{|V(\mathrm{T})|-1} & \int d q_{e^a_v = R(\mathrm{T})} d q_{e^b_v} (1 + |q_{e^a_v} |^2 )^\beta \frac{1}{(1 + |q_{e^b_v} |^2)^\alpha} \frac{1}{(1 + |q_{e^a_v} - q_{e^b_v} |^2)^\alpha}\\
& \times \frac{| q_{e^a_v}|^2}{( q^2_{e^a_v = R(\mathrm{T})} + q_{e^b_v}^2 + (q_{e^a_v} - q_{e^b_v})^2 )^{2 (1-\varepsilon)}}\\
\le C_\alpha^{|V(\mathrm{T})|} \int & d q_{e^a_v} (1 + |q_{e^a_v} |^2 )^{\beta - \alpha}
\le C_{\beta, \alpha}^{|V(\mathrm{T})|}
\end{split}\ee
since $\alpha > \frac{1}{2}$ and $\beta < - \frac{3}{2}.$

Therefore, it remains to estimate the integration over the tree $\mathrm{T}$ containing $\bar{v},$ which we denote by $\Xi (\mathrm{T}, \bar{v}).$ If $e^a_{\bar{v}} = R (\mathrm{T}),$ then for $\alpha > \frac{3}{2}$ and $\beta < - \frac{3}{2}$ one has
\be\begin{split}
\Xi (\mathrm{T}, \bar{v}) = & \int d q_{e^a_{\bar{v}}} d q_{e^b_{\bar{v}}} d q_{e^c_{\bar{v}}} \frac{\delta ( q_{e^a_{\bar{v}}} - q_{e^b_{\bar{v}}} - q_{e^c_{\bar{v}}} ) | q_{e^b_{\bar{v}}} + q_{e^c_{\bar{v}}} |^2 (1 + |q_{e^a_{\bar{v}}}|^2)^\beta}{( 1 + |q_{e^b_{\bar{v}}}|^2 )^\alpha ( 1 + |q_{e^c_{\bar{v}}}|^2 )^\alpha}\\
& \le \int \frac{ |q_{e^a_{\bar{v}}}|^2 (1 + |q_{e^a_{\bar{v}}}|^2)^\beta d q_{e^a_{\bar{v}}}}{(1 + |q_{e^a_{\bar{v}}}|^2)^\alpha} \sup_{q_{e^a_{\bar{v}}} \in \mathbb{R}^3} \int d q_{e^b_{\bar{v}}} \frac{( 1 + |q_{e^a_{\bar{v}}}|^2 )^\alpha}{( 1 + |q_{e^b_{\bar{v}}}|^2 )^\alpha ( 1 + | q_{e^a_{\bar{v}}} - q_{e^b_{\bar{v}}} |^2 )^\alpha}\\
& \lesssim \int \frac{ |q_{e^a_{\bar{v}}}|^2 (1 + |q_{e^a_{\bar{v}}}|^2)^\beta d q_{e^a_{\bar{v}}}}{(1 + |q_{e^a_{\bar{v}}}|^2)^\alpha} < \8.
\end{split}\ee

It remains to deal with the case $e^a_{\bar{v}} \not= R (\mathrm{T}),$ where
\be\begin{split}
\Xi (\mathrm{T}) = \int \prod_{e \in E_2 (\mathrm{T})} d q_e \prod_{v \in V (\mathrm{T})} \delta \big (q_{e^a_v} - q_{e^b_v} - q_{e^c_v}\big ) & \frac{(1 + |q_{e^a_v = R(\mathrm{T})} |^2 )^\beta | q_{e^b_v} + q_{e^c_v} |^2}{(1 + q^2_{e^a_v = R(\mathrm{T})} + q_{e^b_v}^2 + q_{e^c_v}^2 )^{2 (1-\varepsilon)}}\\
\times | q_{e^b_{\bar{v}}} + q_{e^c_{\bar{v}}} |^2 \prod_{e \in L_2 (\mathrm{T})} \frac{1}{(1 + |q_e |^2)^\alpha} & \prod_{\substack{v \in V (\mathrm{T}):e^a_v \not= R(\mathrm{T})\\ v \not= \bar{v}} } \frac{| q_{e^b_v} + q_{e^c_v} |^2}{(1 + q^2_{e^b_v} + q^2_{e^c_v})^{2 (1-\varepsilon)}}.
\end{split}\ee
At first, the integration associated with $\bar{v}$ is
\be\begin{split}
\Xi_{\bar{v}} = & \int d q_{e^b_{\bar{v}}} d q_{e^c_{\bar{v}}} \delta \big ( q_{e^a_{\bar{v}}} - q_{e^b_{\bar{v}}} - q_{e^c_{\bar{v}}}\big ) \frac{ | q_{e^b_{\bar{v}}} + q_{e^c_{\bar{v}}} |^2}{ (1+ | q_{e^b_{\bar{v}}}|^2)^\alpha  (1+ | q_{e^c_{\bar{v}}}|^2)^\alpha} \lesssim C_\alpha < \8
\end{split}\ee
whenever $\alpha > \frac{3}{2}.$

Secondly, suppose that $v$ is a vertex with one of the leaves being $e^a_{\bar{v}},$ for instance $e^b_v = e^a_{\bar{v}}.$ If $e^a_v \not= R(\mathrm{T}),$ then for $\alpha > \frac{3}{2},$ the associated integration
\be\begin{split}
\Xi_v \le & C_\alpha \int d q_{e^b_v} d q_{e^c_v} \delta \big (q_{e^a_v} - q_{e^b_v} - q_{e^c_v}\big ) \frac{1}{(1 + |q_{e^c_v}|^2 )^\alpha} \frac{| q_{e^b_v} + q_{e^c_v} |^2}{(1 + q^2_{e^b_v} + q^2_{e^c_v})^{2 (1-\varepsilon)}}\\
\le & C_\alpha \int d q_{e^b_v} \frac{1}{(1 + |q_{e^a_v} - q_{e^b_v}|^2)^\alpha} \frac{| q_{e^b_v}|^2 + |q_{e^a_v} - q_{e^b_v}|^2}{(1 + q^2_{e^b_v} + (q_{e^a_v} - q_{e^b_v})^2 )^{2 (1-\varepsilon)}}\\
\lesssim & C_\alpha \int d q_{e^b_v} \frac{1}{(1 + |q_{e^a_v} - q_{e^b_v}|^2)^\alpha} \frac{1}{(1+ |q_{e^b_v}|^2 + |q_{e^a_v} - q_{e^b_v}|^2)^{1 - 2 \varepsilon}}\\
\lesssim & C_\alpha < \8
\end{split}\ee
provided $\varepsilon = \frac{1}{4} \min \{1, \alpha - \frac{3}{2} \}.$ Subsequently, for each $v \prec \bar{v}$ with $e^a_v \not= R(\mathrm{T}),$ the associated integration is the same as this $\Xi_v$ and so $\lesssim C_\alpha.$

Thirdly, for a maximal vertex $v \in V (\mathrm{T}) \backslash \{ \bar{v} \}$ whose two daughter edges must be leaves, the associated integration is
\be\begin{split}
\Xi_v = & \int d q_{e^b_v} d q_{e^c_v} \delta \big (q_{e^a_v} - q_{e^b_v} - q_{e^c_v}\big ) \frac{1}{(1 + | q_{e^b_v} |^2 )^\alpha (1 + | q_{e^c_v} |^2 )^\alpha} \frac{| q_{e^b_v} + q_{e^c_v} |^2}{(1 + q^2_{e^b_v} + q^2_{e^c_v})^{2 (1-\varepsilon)}}\\
\le & \frac{C_\alpha}{(1 + | q_{e^a_v} |^2)^\alpha}
\end{split}\ee
provided $\alpha > \frac{3}{2}$ and $\varepsilon = \frac{1}{4} \min \{ 1, \alpha - \frac{3}{2} \} >0,$ as shown in \eqref{eq:Integralonevertexbound}. Subsequently, every vertex $v$ with $v \nprec \bar{v}$ for which all vertices $v'$ with $v \prec v'$ have already been integrated out is associated with the integration of the form $\Xi_v$ as above.

Finally, we will arrive the vertex $v \prec \bar{v}$ with $e^a_v = R(\mathrm{T}).$ Then the associated integration (without loss of generality, we assume that $e^b_v$ is at the route towards $\bar{v}$)
\be\begin{split}
\Xi_v \le & C_\alpha^{|V(\mathrm{T})|-1} \int d q_{e^a_v = R(\mathrm{T})} d q_{e^b_v} d q_{e^c_v} \delta \big (q_{e^a_v} - q_{e^b_v} - q_{e^c_v}\big )  \frac{1}{( 1 + |q_{e^c_v} |^2 )^\alpha} \frac{(1 + |q_{e^a_v}|^2)^\beta | q_{e^b_v} + q_{e^c_v} |^2}{(1 + q^2_{e^a_v} + q^2_{e^b_v} + q^2_{e^c_v})^{2 (1-\varepsilon)}}\\
\le & \int d q_{e^a_v} d q_{e^b_v} (1 + |q_{e^a_v}|^2)^\beta \frac{1}{( 1 + |q_{e^a_v} - q_{e^b_v}|^2 )^\alpha (1 + q^2_{e^a_v} + q^2_{e^b_v} + (q_{e^a_v} - q_{e^b_v})^2 )^{1-2 \varepsilon}} \le  C_{\beta, \alpha}^{|V(\mathrm{T})|}
\end{split}\ee
when $\beta < - \frac{3}{2}$ and $\alpha > \frac{3}{2}$ with $\varepsilon = \frac{1}{4} \min \{1, \alpha - \frac{3}{2} \}.$

In conclusion, if $\beta < - \frac{3}{2}$ and $\alpha > \frac{3}{2}$ with $\varepsilon = \frac{1}{4} \min \{1, \alpha - \frac{3}{2} \},$ we have
\be
\Xi (\mathrm{T}, \bar{v}) \le C^{|V (\mathrm{T})|}
\ee
and hence
\be
\Xi (\mathbb{T}, \bar{v}) = \Xi (\mathrm{T}, \bar{v}) \prod_{\bar{v} \notin \mathrm{T}} \Xi (\mathrm{T}) \le C^n
\ee
because $V ( \mathbb{T}) =n.$ Note that $|M (\mathbb{T})| \le n.$ Therefore, we conclude \eqref{eq:Erroroperatorspace-spacetimeSobolevbound}.
\end{proof}

\section{Appendix}\label{App}

\subsection{Binary trees with marked edges}\label{Binarytree}

In Section \ref{GraphicNShierachy} above, we have used binary trees to express the Duhamel expansion of the Navier-Stokes hierarchy. For the sake of convenience, we present the details of the binary trees following \cite{ESY2007}.

We begin with the definition of a binary tree.

\begin{defi}\label{df:Binarytree}\rm
For a nonnegative integer $n \ge 0,$ a {\it binary tree} of $n$ order $\mathrm{T}$ consist of a root, $n$ vertices, and $n+1$ leaves such that each vertex is adjacent to three edges. The root and the leaves are not regards as vertices, which are instead identified with the unique edge they are adjacent to. We denote by $V(\mathrm{T})$ the set of vertices and by $E(\mathrm{T})$ the set of edges. The root is denoted by $R = R(\mathrm{T}),$ the set of $n+1$ leaves is denoted by $L= L(\mathrm{T}).$ They are called external edges and denoted by $E_{\mathrm{ext}} (\mathrm{T}).$ We denote by $E_{\mathrm{in}} (\mathrm{T})= E(\mathrm{T}) \setminus E_{\mathrm{ext}} (\mathrm{T})$ the internal edges.
\end{defi}

\begin{rk}\label{rk:trivialtree}\rm
For $n=0,$ namely there is no vertex, there is only one single edge, that is the root and the single leaf at the same time; but we count this edge twice when counting the external edges. This tree is called {\it trivial}.
\end{rk}

At every vertex, the one of the three edges that is closest to the root is called {\it mother-edge}, the other two are called {\it daughter-edges} of this vertex with one of which being marked. For illustration, we draw such a tree so that the marked daughter-edge goes straight through, and the unmarked daughter-edge joins from below, i.e., the root is on the left, and the leaves are on the right of the graph (see e.g. Fig.\ref{fig:binarytree5vertices}).
\begin{figure}[htb]
\begin{minipage}[t]{15cm}
 \centering
 \includegraphics[height=6cm,width=12cm]{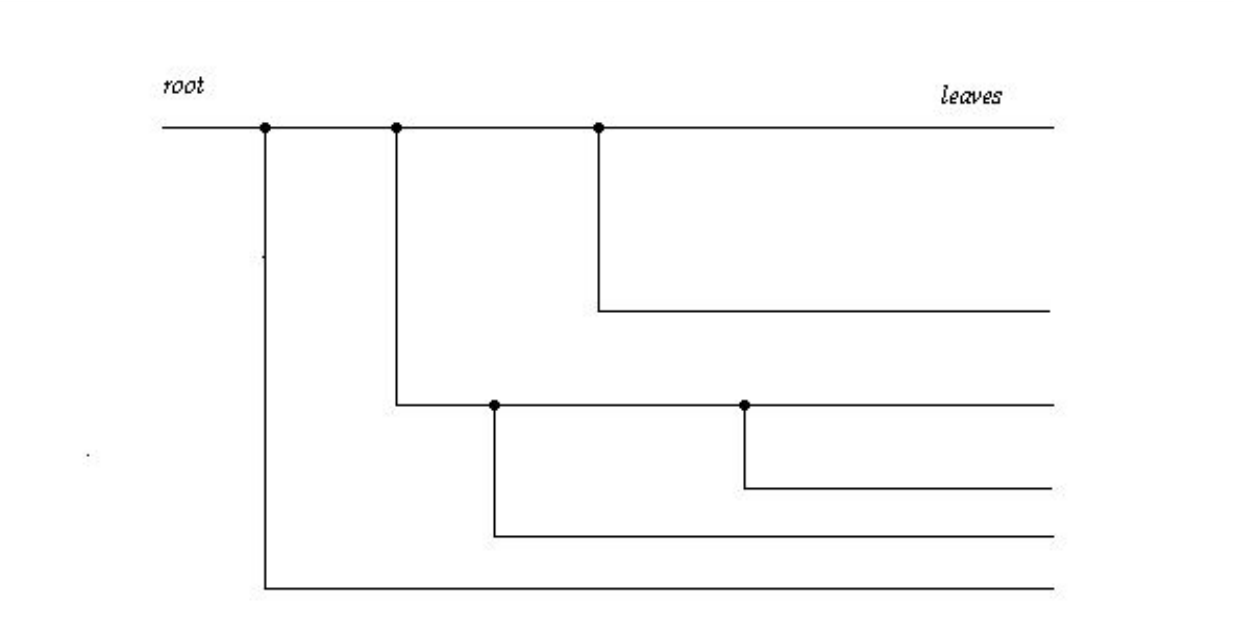}
\caption{\label{fig:binarytree5vertices} Example of a rooted, marked binary tree with $n=5$ vertices.}
\end{minipage}
\end{figure}

The set of all marked binary trees of $n$ order is denote by $\mathfrak{T}_n.$ Two trees $\mathrm{T}_1$ and $\mathrm{T}_2$ in $\mathfrak{T}_n$ are said to be equivalent if there exists a one-to-one map between the edges and the vertices of $\mathrm{T}_1$ and $\mathrm{T}_2$ such that all adjacency relations and all marked edges of the vertices are preserved. In what follows, any element $\mathrm{T}$ in $\mathfrak{T}_n$ is simply called a binary tree (of $n$ order). As shown in \cite[\S 9.1.1]{ESY2007}, the number of (inequivalent) marked binary trees of $n$ order, $C_n = |\mathfrak{T}_n|,$ is equal to the so-called $n$-th Catalan number
\be
C_n = \frac{1}{n+1} \binom{2n}{n}
\ee
and can be estimated by $C_n \le 4^n.$

\begin{defi}\label{df:Order}\rm
For a given tree $\mathrm{T} \in \mathfrak{T}_n,$ we define a partial order $\prec$ on the vertices $V(\mathrm{T})$ as follows: for any $v, v' \in V(\mathrm{T}),$ we have $v \prec v'$ if $v$ lies on the (unique) route from $v'$ to $R (\mathrm{T}).$
\end{defi}

We need to define the notion of forests that is used in the graphic representation of the Navier-Stokes hierarchy.

\begin{defi}\label{df:Forest}\rm
For $n \ge 0$ and $k \ge 1,$ we denote by $\mathfrak{T}_{n,k}$ the set of $k$-tuples of binary trees, $(\mathrm{T}_1, \cdots, \mathrm{T}_k),$ so that the total number of vertices equals to $n,$ i.e., $\sum^k_{j=1} |V(\mathrm{T}_j)| =n.$ An element $\mathbb{T}= (\mathrm{T}_1, \cdots, \mathrm{T}_k)$ in $\mathfrak{T}_{n,k}$ is said to be a {\it forest}. Again, for $\mathbb{T} \in \mathfrak{T}_{n,k}$ we denote by $V(\mathbb{T})$ the set of the vertices of $\mathbb{T},$ by $E(\mathbb{T})$ the edges, and by $E_{\mathrm{int}} (\mathbb{T})$ (resp. $E_{\mathrm{ext}} (\mathbb{T})$) the set of internal (resp. external) edges. We also denote by $R (\mathbb{T})$ the set of its $k$ roots and by $L(\mathbb{T})$ the set of its $n+k$ leaves.
\end{defi}

For a forest $\mathbb{T} \in \mathfrak{T}_{n,k},$ a root and the corresponding component will be called trivial if it contains no vertex. The set of trivial roots of $\mathbb{T}$ will be denoted by $R_1 (\mathbb{T}),$ and we set $R_2 (\mathbb{T}):= R(\mathbb{T}) \setminus R_1 (\mathbb{T}).$ Let $L_1 (\mathbb{T})$ be the set of leaves of the trivial components of $\mathbb{T},$ which is naturally identified with $R_1 (\mathbb{T}).$ We set $L_2(\mathbb{T}) = L (\mathbb{T}) \setminus L_1 (\mathbb{T}).$ See Fig.\ref{fig:Forestn=9k=5} for an example of a forest in $\mathfrak{T}_{9,5}.$
\begin{figure}[htb]
\begin{minipage}[t]{12cm}
\vspace{0pt}
 \centering
 \includegraphics[height=8cm,width=10cm]{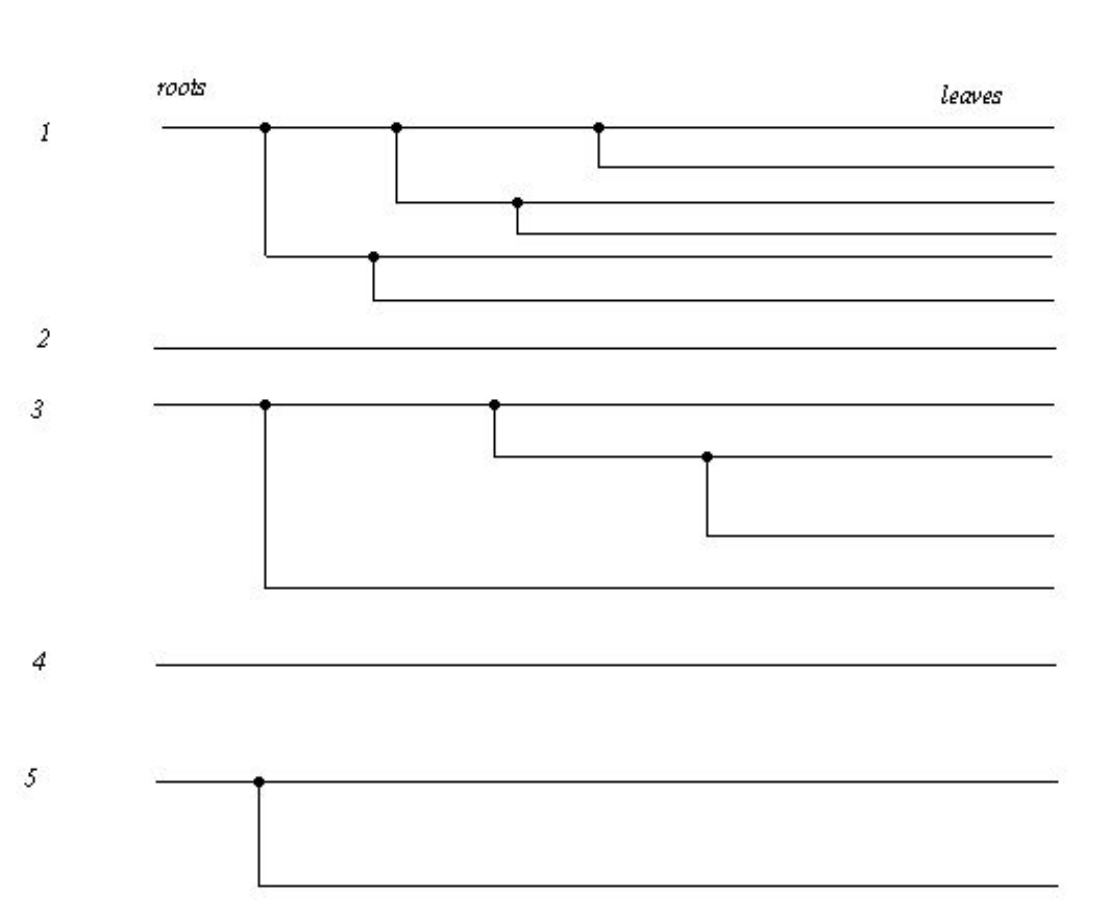}
\caption{\label{fig:Forestn=9k=5} Example of a forest with $n=9$ and $k=5.$}
\end{minipage}
\end{figure}

\begin{defi}\label{df:labelling}\rm
For a forest $\mathbb{T} = (\mathrm{T}_1, \cdots, \mathrm{T}_k) \in \mathfrak{T}_{n,k},$ we define the labelling map $\pi_1: R (\mathbb{T}) \mapsto \{1,\ldots, k \}$ so that $\pi_1 (r)=j$ if $r = R(\mathrm{T}_j),$ the root of $\mathrm{T}_j.$ Any one-to-one map $\pi_2: L(\mathbb{T}) \mapsto \{1, \ldots, n+k \}$ is called a labelling for the leaves of $\mathbb{T}.$ There are $(n+k)!$ different labelling maps $\pi_2.$
\end{defi}

Two forests $\mathbb{T}, \mathbb{T}' \in \mathfrak{T}_{n,k}$ are equivalent if $\mathrm{T}_j = \mathrm{T}'_j$ for all $1 \le j \le k.$ Then, the permutation of the trees in a forest usually results in inequivalent forests. Note that the number of inequivalent forests in $\mathfrak{T}_{n,k}$ is given by
\be
| \mathfrak{T}_{n,k}| = \sum_{\substack{(n_1, \ldots , n_k)\\ \sum_i n_i =n}} C_{n_1} \cdots C_{n_k}
\ee
where the summation takes over all $k$-tuples of nonnegative integers that add up to $n.$ Thus, this number is bounded by
\beq\label{eq:ForestNo}
| \mathfrak{T}_{n,k}| \le 4^n \binom{n+k-1}{k-1} \le 2^{3 n +k},
\eeq
as noted in  \cite[\S 9.1.1]{ESY2007}.

Given $\mathbb{T} \in \mathfrak{T}_{n,k},$ we again define a partial order $\prec$ on $V(\mathbb{T})$ as follows: for any $v, v' \in V(\mathbb{T}),$ we have $v \prec v'$ if they are within the same tree and $v \prec v'$ on this tree as defined above; there is no order relation between $v$ and $v'$ if they are in different trees.

\subsection{Preliminary estimates}\label{PreEst}

For the sake of convenience, we collect some estimates used in the body of the paper.

\begin{lem}\label{lem:WoperatordfWell}
Suppose that $h$ is a nonnegative function in $\mathcal{D} (\mathbb{R}^3)$ supported in the unit ball of $\mathbb{R}^3$ such that $\int h d x =1.$ For any $\epsilon >0,$ we let $\delta_\epsilon (x) = \epsilon^{-3} h ( \epsilon^{-1} x).$ Given $k \ge 1,$ for $\phi \in \mathcal{D} (\mathbb{R}^{3 k})$ and $u \in \mathcal{S} (\mathbb{R}^{3(k+1)})$ we have, for every $1 \le j \le k,$
\beq\label{eq:Woperatorappr1}
\begin{split}
\Big | \int d x_{k+1} \big \langle \partial_{x^i_j} \phi, [ & \delta_\epsilon (x_j - x_{k+1}) - \delta (x_j - x_{k+1}) ] u (\cdot, x_{k+1}) \big \rangle_{L^2 (\mathbb{R}^{3 k})} \Big |\\
& \le C \epsilon^{\frac{1}{2}} \| \phi \|_{H^2 (\mathbb{R}^{3k})} \| (1 - \triangle_j)^{\frac{1}{2}} (1 - \triangle_{k+1})^{\frac{1}{2}} u \|_{L^2 (\mathbb{R}^{3(k+1)})}
\end{split}
\eeq
and
\beq\label{eq:Woperatorappr2}
\begin{split}
\Big | \int d x_{k+1} \big \langle R_{x^{\iota}_j} R_{x^{\el}_j} \partial_{x^i_j} \phi, & [ \delta_\epsilon (x_j - x_{k+1}) - \delta (x_j - x_{k+1}) ] u (\cdot, x_{k+1}) \big \rangle_{L^2 (\mathbb{R}^{3 k})} \Big |\\
& \le C \epsilon^{\frac{1}{2}} \| \phi \|_{H^2 (\mathbb{R}^{3k})} \| (1 - \triangle_j)^{\frac{1}{2}} (1 - \triangle_{k+1})^{\frac{1}{2}} u \|_{L^2 (\mathbb{R}^{3(k+1)})}
\end{split}
\eeq
for all $1 \le i, \el, \iota \le 3,$ where $C>0$ is a universal constant.
\end{lem}

\begin{proof}
The proof follows the argument of \cite[Lemma 8.2]{ESY2007}. We first prove \eqref{eq:Woperatorappr1}. Given $\phi \in \mathcal{S} (\mathbb{R}^{3 k})$ and $u \in \mathcal{S} (\mathbb{R}^{3(k+1)}),$ note that for any $\epsilon >0,$
\be\begin{split}
\int d x_{k+1} & \Big \langle \partial_{x^i_j} \phi, u (\cdot, x_{k+1}) \big [ \delta_\epsilon (x_j - x_{k+1}) - \delta (x_j - x_{k+1}) \big ] \Big \rangle_{L^2 (\mathbb{R}^{3 k})}\\
= & \int \big [ \Phi ( \vec{x}_{k,j \to k+1} ) u (\vec{x}_{k, j \to k+1}, x_{k+1}) - \Phi (\vec{x}_k) u (\vec{x}_{k+1}) \big ] \delta_\epsilon (x_j - x_{k+1}) d \vec{x}_{k+1}
\end{split}\ee
where $\Phi = \partial_{x^i_j}\phi$ and $\vec{x}_{k, j \to k+1} = (x_1, \ldots, x_{j-1}, x_{k+1}, x_{j+1}, \ldots, x_k).$ Since $\delta_\epsilon (x) \le \frac{C}{|B|} \chi_B (x)$ with $B = \{ x \in \mathbb{R}^3 :\; |x| \le \epsilon \},$ by a Poincar\'{e}-type inequality (cf. \cite[Lemma 7.16]{GT2001}) we have
\be\begin{split}
\bigg | \int \big [ \Phi ( \vec{x}_{k, j \to k+1} ) u (\vec{x}_{k, j \to k+1}, x_{k+1}) & - \Phi (\vec{x}_k) u (\vec{x}_k, x_{k+1}) \big ] \delta_\epsilon (x_j - x_{k+1}) d x_j \bigg |\\
\le & C \int_{|x_j - x_{k+1}| \le \epsilon} \frac{\big | \bigtriangledown_j \big [ \Phi ( \vec{x}_k ) u (\vec{x}_k, x_{k+1}) \big ] \big |}{|x_j - x_{k+1}|^2} d x_j
\end{split}\ee
for all $\vec{x}_{k, j \to k+1},$ where $C$ is a universal constant. Inserting this inequality on the right hand side of the proceeding equality and applying the Schwarz inequality, we have
\be\begin{split}
\bigg | \int & d x_{k+1} \Big \langle \partial_{x^i_j} \phi, u (\cdot, x_{k+1}) \big [ \delta_\epsilon (x_j - x_{k+1}) - \delta (x_j - x_{k+1}) \big ] \Big \rangle_{L^2 (\mathbb{R}^{3 k})} \bigg | \\
\lesssim & \int \Big [ | \bigtriangledown_j \Phi( \vec{x}_k ) | |u (\vec{x}_{k+1})| + | \Phi( \vec{x}_k ) | | \bigtriangledown_j u (\vec{x}_{k+1})| \Big ] \frac{\chi_{\{|x_j - x_{k+1}| \le \epsilon \}} (\vec{x}_{k+1})d \vec{x}_{k+1}}{|x_j - x_{k+1}|^2}\\
\lesssim & \bigg ( \int | \bigtriangledown_j \Phi( \vec{x}_k ) |^2 \frac{\chi_{\{|x_j - x_{k+1}| \le \epsilon \}} (\vec{x}_{k+1}) d \vec{x}_{k+1}}{|x_j - x_{k+1}|^2} \bigg )^{\frac{1}{2}} \bigg ( \int |u (\vec{x}_{k+1})|^2 \frac{\chi_{\{|x_j - x_{k+1}| \le \epsilon \}} (\vec{x}_{k+1}) d \vec{x}_{k+1} }{|x_j - x_{k+1}|^2} \bigg )^{\frac{1}{2}}\\
& \; + \bigg ( \int | \Phi( \vec{x}_k ) |^2 \frac{\chi_{\{|x_j - x_{k+1}| \le \epsilon \}} (\vec{x}_{k+1}) d \vec{x}_{k+1} }{|x_j - x_{k+1}|^2} \bigg )^{\frac{1}{2}} \bigg ( \int |\bigtriangledown_j  u (\vec{x}_{k+1})|^2 \frac{\chi_{\{|x_j - x_{k+1}| \le \epsilon \}} (\vec{x}_{k+1})d \vec{x}_{k+1}}{|x_j - x_{k+1}|^2} \bigg )^{\frac{1}{2}}.
\end{split}\ee
In the terms containing $\Phi$ we perform the $x_{k+1}$ integration and obtain
\be
\bigg ( \int \frac{\chi_{\{|x_j - x_{k+1}| \le \epsilon \}} (\vec{x}_{k+1})}{|x_j - x_{k+1}|^2} | \bigtriangledown_j \Phi( \vec{x}_k ) |^2 d \vec{x}_{k+1} \bigg )^{\frac{1}{2}} \le \epsilon^{\frac{1}{2}} \| \phi \|_{H^2 (\mathbb{R}^{3k})}
\ee
and
\be
\bigg ( \int \frac{\chi_{\{|x_j - x_{k+1}| \le \epsilon \}} (\vec{x}_{k+1})}{|x_j - x_{k+1}|^2} | \Phi( \vec{x}_k ) |^2 d \vec{x}_{k+1} \bigg )^{\frac{1}{2}} \le \epsilon^{\frac{1}{2}} \| \phi \|_{H^1 (\mathbb{R}^{3k})}.
\ee
In the terms containing $u,$ dropping the restriction $\chi_{\{|x_j - x_{k+1}| \le \epsilon \}}$ and applying the Hardy inequality to the $x_{k+1}$-integration, we obtain
\be\begin{split}
\bigg ( \int & \frac{\chi_{\{|x_j - x_{k+1}| \le \epsilon \}} (\vec{x}_{k+1})}{|x_j - x_{k+1}|^2} |u (\vec{x}_{k+1})|^2 d \vec{x}_{k+1} \bigg )^{\frac{1}{2}}\\
& \le \bigg ( \int |\bigtriangledown_{k+1} u (\vec{x}_{k+1})|^2 d \vec{x}_{k+1} \bigg )^{\frac{1}{2}} \le  \| (1- \triangle_j)^{\frac{1}{2}} (1- \triangle_{k+1})^{\frac{1}{2}} u \|_{L^2 (\mathbb{R}^{3 (k+1)})}
\end{split}\ee
and
\be\begin{split}
\bigg ( \int & \frac{\chi_{\{|x_j - x_{k+1}| \le \epsilon \}} (\vec{x}_{k+1})}{|x_j - x_{k+1}|^2} |\bigtriangledown_j u (\vec{x}_{k+1})|^2 d \vec{x}_{k+1} \bigg )^{\frac{1}{2}}\\
& \le \bigg ( \int |\bigtriangledown_j \bigtriangledown_{k+1} u (\vec{x}_{k+1})|^2 d \vec{x}_{k+1} \bigg )^{\frac{1}{2}} \le  \| (1- \triangle_j)^{\frac{1}{2}} (1- \triangle_{k+1})^{\frac{1}{2}} u \|_{L^2 (\mathbb{R}^{3 (k +1)})}.
\end{split}\ee
In summary, we get the inequality \eqref{eq:Woperatorappr1}.

Thanks to the fact that Riesz transforms $R_{x^i}$ are bounded on $L^2 (\mathbb{R}^3),$ replacing $\phi$ by $R_{x^{\iota}_j} R_{x^{\el}_j} \phi$ in \eqref{eq:Woperatorappr1} we obtain \eqref{eq:Woperatorappr2}.
\end{proof}

\begin{lem}\label{lem:Potentialestimate} {\rm (cf. \cite[Lemma A.3]{ESY2007})}
Let $V$ be a nonnegative function in $L^1 (\mathbb{R}^3).$ Then there exists a universal constant $C>0$ such that
\be
\int d x dy V(x-y) | u (x, y) |^2 \le C \| V \|_1 \| u \|^2_{\mathrm{H}^1_{(2)}}
\ee
for all $ u \in \mathrm{H}^1_{(2)} (\mathbb{R}^3).$
\end{lem}

\begin{lem}\label{lem:Frequencyestimate} {\rm (cf. \cite[Lemma 10.1]{ESY2007})}
For three nonnegative numbers $\epsilon, \lambda, \eta$ satisfying $0 \le \epsilon< \lambda < 1$ and $0 < \eta < \lambda - \epsilon,$ there exists a constant $C_{\epsilon, \lambda, \eta}>0$ such that
\be
\int^\8_{- \8} \frac{d \beta}{\langle \alpha - \beta \rangle^{1 - \epsilon} |\beta|^\lambda} \le C_{\epsilon, \lambda, \eta} \frac{1}{\langle \alpha \rangle^{\lambda - \epsilon - \eta}}
\ee
for all $ \alpha \in \mathbb{R}.$
\end{lem}

\begin{lem}\label{lem:FrequencyestimateM}
For three nonnegative numbers $\epsilon, \lambda, \eta$ satisfying $0 \le \epsilon< \lambda < 1$ and $0 < \eta < \lambda - \epsilon,$ there exists a constant $C_{\epsilon, \lambda, \eta}>0$ such that
\be
\int^\8_{- \8} \frac{d t}{| s - t + a \mathrm{i}|^{1 - \epsilon} |t + b \mathrm{i} |^\lambda} \le C_{\epsilon, \lambda, \eta} \frac{1}{| s + (a + b) \mathrm{i}|^{\lambda - \epsilon - \eta}}
\ee
for all $a, b \ge 1$ and $ s \in \mathbb{R}.$
\end{lem}

\begin{proof}
Without loss of generality, we can assume that $a \ge b \ge 1.$ By Lemma \ref{lem:Frequencyestimate}, we have
\be\begin{split}
\int^\8_{- \8} \frac{d t}{| s - t + a \mathrm{i}|^{1 - \epsilon} |t + b \mathrm{i} |^\lambda}
= & \frac{1}{a^{\lambda - \epsilon}} \int^\8_{- \8} \frac{d t}{| s/a - t + \mathrm{i}|^{1 - \epsilon} |t + \mathrm{i} b/a |^\lambda}\\
\le & \frac{1}{a^{\lambda - \epsilon}} \int^\8_{- \8} \frac{d t}{\langle s/a - t \rangle | t|^\lambda} \le C_{\epsilon, \lambda, \eta} \frac{1}{| s + (a + b) \mathrm{i}|^{\lambda - \epsilon - \eta}}
\end{split}\ee
where we use in the last inequality the conditions $a \ge b \ge 1$ and $\eta \ge 0.$
\end{proof}

\

{\it Acknowledgment.}\; The author is grateful to Chuangye Liu for making helpful suggestions on the draft of the manuscript. This work was partially supported by the Natural Science Foundation of China under Grant No.11431011.


\begin{thebibliography}{99}

\bibitem{Chen2013} Z. Chen,
Local well-posedness for Gross-Pitaevskii hierarchies,
{\it Acta Anal. Funct. Appl.} {\bf 15} (2013), 291-305.

\bibitem{DS1964} N. Dunford, J. T. Schwartz,
{\it Linear Operators, Part I: General Theory,}
Interscience Publishers, INC., New York, 1964.

\bibitem{DU1977} J. Diestel, J.J. Uhl,Jr., {\it Vector Measures,} American Mathematical Society,
Providence, Rhode Island, 1977.

\bibitem{ESY2007} L. Erd\"os, B. Schlein, H. T. Yau,
Derivation of the cubic non-linear Schr\"odinger equation from quantum dynamics of many-body systems,
{\it Invent. Math.} {\bf 167} (2007), 515-614.

\bibitem {GT2001} D. Gilbarg, N. S. Trudinger,
{\it Elliptic Partial Differential Equations of Second Order {\rm (Second edition)},}
Springer-Verlag, Berlin, 2001.





\bibitem{Leray1934} J. Leray,
Sur le mouvement d'un fluide visqueux emplissant l'espace,
{\it Acta Math.} {\bf 63} (1934), 193-248.

\bibitem{LR2002} P. G. Lemari\'{e}-Rieusset,
{\it Recent Developments in the Navier-Stokes Problem,}
Chapman \& Hall/CRC, Boca Raton, 2002.

\bibitem{LR2016} P. G. Lemari\'{e}-Rieusset,
{\it The Navier-Stokes Problem in the 21st Century,}
CRC Press, Boca Raton, 2016.






\end{thebibliography}
\end{document}